\documentclass[12pt]{article}
\usepackage[left=1in,top=1in,right=1in,bottom=1in,nohead]{geometry}
\usepackage{amsmath,amssymb,comment}
\usepackage{mathtools}
\usepackage{rotating}
\usepackage{graphicx}

\usepackage{amsmath,amssymb}
\usepackage{amssymb}
\usepackage{bm}                
\usepackage{mathrsfs}      
\usepackage{amsmath, amsthm}
\usepackage{graphicx,epsf,subfigure,pstricks,pst-node,psfrag,amsthm,amssymb,amsmath}
\usepackage{float}
\usepackage{array}
\usepackage{authblk}
\usepackage{natbib}
\usepackage{url}
\usepackage{hyperref}
\usepackage{rotating}
\usepackage{multicol}
\usepackage{multirow}
\newtheorem{theorem}{Theorem}
\newtheorem{lemma}[theorem]{Lemma}
\newtheorem{assumption}[theorem]{Assumption}

\newcommand{\Ind}{1\!\mathrm{l}}

\begin{document}
	\title{Horseshoe shrinkage methods for Bayesian fusion estimation}
	\date{} 
	\author{Sayantan Banerjee \footnote{Corresponding author. Address: IIM Indore, Rau-Pithampur Road, Indore, M.P. - 453 556, India. e-mail: \url{sayantanb@iimidr.ac.in}}}
	
	\affil{\small Operations Management and Quantitative Techniques Area\\ Indian Institute of Management Indore\\ Indore, M.P., India}
	\maketitle
	
	\begin{abstract}
		We consider the problem of estimation and structure learning of high dimensional signals via a normal sequence model, where the underlying parameter vector is piecewise constant, or has a block structure. We develop a Bayesian fusion estimation method by using the Horseshoe prior to induce a strong shrinkage effect on successive differences in the mean parameters, simultaneously imposing sufficient prior concentration for non-zero values of the same. The proposed method thus facilitates consistent estimation and structure recovery of the signal pieces. We provide theoretical justifications of our approach by deriving posterior convergence rates and establishing selection consistency under suitable assumptions. We also extend our proposed method to signal de-noising over arbitrary graphs and develop efficient computational methods along with providing theoretical guarantees. We demonstrate the superior performance of the Horseshoe based Bayesian fusion estimation method through extensive simulations and two real-life examples on signal de-noising in biological and geophysical applications. We also demonstrate the estimation performance of our method on a real-world large network for the graph signal de-noising problem.
\end{abstract}

	\begin{quotation}

	\end{quotation}\par

\section{Introduction}

With modern technological advances and massive data storage capabilities, large datasets are becoming increasingly common in a plethora of areas including finance, econometrics, bioinformatics, engineering, signal-processing, among others. Flexible modeling of such datasets often require parameters whose dimension exceed the available sample size. Plausible inference is possible in such scenarios by finding a lower dimensional embedding of the high dimensional parameter.

Sparsity plays an important role in statistical learning in a high-dimensional scenario, where the underlying parameters are `nearly black' \citep{donoho1992maximum}, implying that only a small subset of the same is non-zero. Several regularization based methods have been proposed in this regard in the literature, that induce sparsity in the models. From a frequentist perspective, such methods include penalization based methods like ridge regression \citep{hoerl1970ridge}, lasso \citep{tibshirani1996regression}, elastic net \citep{zou2005regularization}, SCAD \citep{fan2001variable}, fused lasso \citep{tibshirani2005sparsity, rinaldo2009properties}, graphical lasso \citep{friedman2008sparse}, among others; see \cite{buhlmann2011statistics} for frequentist methods in high dimensions. Bayesian methods in high dimensions have been developed more recently, where sparsity is induced via suitable prior distributions on the parameters. These include spike-and-slab priors \citep{mitchell1988bayesian, ishwaran2005spike}, Bayesian lasso \citep{park2008bayesian}, Bayesian graphical lasso \citep{wang2012bayesian}, and other shrinkage priors like spike-and-slab lasso \citep{rovckova2018spike}, Normal Exponential Gamma \citep{griffin2010inference}, Horseshoe \citep{carvalho2009handling, carvalho2010horseshoe} and other variants like Horseshoe+ \citep{bhadra2017horseshoe+} and Horseshoe-like \citep{bhadra2019horseshoe}, Dirichlet-Laplace \citep{bhattacharya2015dirichlet}, R2-D2 \citep{zhang2016high}, among others. Theoretical guarantees of such Bayesian procedures have been developed recently as well. We refer the readers to \cite{banerjee2021bayesian} for a comprehensive overview of Bayesian methods in a high-dimensional framework.

Piecewise constant signal estimation and structure recovery is an important and widely studied problem in this regard. Such signals occur in numerous applications, including bioinformatics, remote sensing, digital image processing, finance, and geophysics. We refer the readers to \cite{little2011generalized} for a review on generalized methods for noise removal in piecewise constant signals, along with potential applications. In this paper, we consider a noisy high-dimensional signal modeled via the normal sequence model given by
\begin{equation}
	y_i = \theta_i + \epsilon_i,\; i = 1,\ldots,n,
	\label{eqn:gaussian-means}
\end{equation}
where $\epsilon_i \stackrel{iid}{\sim}N(0, \sigma^2),$ $\sigma^2$ being the noise (or, error) variance, $y = (y_1,\ldots,y_n)^T$ is the observed noisy signal, and $\theta = (\theta_1,\ldots,\theta_n)^T \in \mathbb{R}^n$ is the vector of mean parameters. We further assume that the true parameter $\theta_0 = (\theta_{0,1},\ldots,\theta_{0,n})^T \in \mathbb{R}^n$ is piecewise constant, or has an underlying block structure, in the sense that the transformed parameters $\eta_{0,j} = \theta_{0,j} - \theta_{0,j-1},\, 2 \leq j \leq n$ belong to a nearly-black class $l_0[s] = \{\eta \in \mathbb{R}^{n-1}: \#\{j: \eta_j \neq 0,\, 2 \leq j \leq n\} \leq s,\, 0 \leq s = s(n) \leq n\}$. Here $\#$ denotes the cardinality of a finite set. The number of true blocks in the parameter vector is $s_0 := \#\{j: \eta_{0,j} \neq 0,\, 2 \leq j \leq n\}.$ Our fundamental goals are two-fold -- (i) estimating the mean parameter $\theta$, and (ii) recovering the underlying piecewise constant structure.
 
 Frequentist methods in addressing the piecewise constant signal de-noising problem include penalization methods like the fused lasso method \citep{tibshirani2005sparsity} and the $L_1$-fusion method \citep{rinaldo2009properties}. The Bayesian equivalent of the fused lasso penalty is using a Laplace shrinkage prior on the successive differences $\eta$ \citep{kyung2010penalized}. However, the Laplace prior leads to posterior inconsistency \citep{castillo2015bayesian}. To overcome this problem, \cite{song2017nearly} used a heavy-tailed shrinkage prior for the coefficients in a linear regression framework. Motivated by this, \cite{song2020bayesian} proposed to use a $t$-shrinkage prior for Bayesian fusion estimation. \cite{shimamura2019bayesian} use a Normal Exponential Gamma (NEG) prior in this context, and make inference based on the posterior mode. 
 
 Motivated by strong theoretical guarantees regarding estimation and structure learning induced by the Horseshoe prior in a regression model \citep{datta2013asymptotic, van2014horseshoe, van2017adaptive}, we propose to investigate the performance of the same in our fusion estimation framework. The Horseshoe prior induces sparsity via an infinite spike at zero, and also possess a heavy tail to ensure consistent selection of the underlying pieces or blocks. Furthermore, the global scale parameter that controls the level of sparsity in the model automatically adapts to the actual level of sparsity in the true model, as opposed to choosing the scale parameter in the $t$-shrinkage prior on the basis of the underlying dimension $n$. The conjugate structure of the posterior distribution arising out of modeling via a Horseshoe prior allows fast computation along with working with the full posterior distribution, so that the samples obtained via Markov Chain Monte Carlo (MCMC) can further be used for uncertainty quantification.
 
 The piecewise constant normal sequence model can be further extended to piecewise constant signals over arbitrary graphs. Note that the former is a special case of signals over graphs, when the underlying graph is a linear chain. We consider the normal sequence model as given by (\ref{eqn:gaussian-means}), but in this case, the signal $\theta = (\theta_1,\ldots,\theta_n)^T$ is defined over an undirected graph $G = (V,E)$, where $V = \{1,\ldots,n\}$ is the vertex set of the graph and $E$ is the corresponding edge-set. The true signal is then assumed to be piecewise constant over the graph $G$, so that the parameter space is now given by $l_0[G,s] := \{\theta \in \mathbb{R}^n: \#\{(i,j) \in E: \theta_i \neq \theta_j \} \leq s, \, 0 \leq s(n) = s \leq n\}$. The resulting graph de-noising problem has several potential applications, including multiple change-point detection in linear chain graphs, image segmentation, and anomaly detection in large networks. We refer the readers to \cite{Fan2018} and references therein for related work in this field. In this work, we propose a Bayesian fusion method via Horseshoe prior specification on adjacent edges of a transformed graph obtained from the original graph and illustrate its excellent inference and posterior concentration properties.
 
 The paper is organized as follows. In the next section, we specify the Bayesian model along with the prior specifications, followed by posterior computations in Section~\ref{sec:posterior}. In Section~\ref{sec:theory}, we provide theoretical guarantees of our proposed method via determining posterior convergence rates and establishing posterior selection consistency of the signal pieces. We demonstrate the numerical performance of our method along with other competing methods via simulations in Section~\ref{sec:simulation}, followed by real-life applications in two different areas -- DNA copy number analysis and Solar X-Ray flux data analysis. The methodological extension of our proposed method to the graph signal de-noising problem, along with theoretical guarantees and real-life illustration are given in Section~\ref{sec:graph de-noising}. We conclude our paper with a brief discussion of our proposed fusion estimation method along with future directions for research. Additional lemmas, proofs of main results, and a practical solution for block structure recovery along with numerical results for the same are provided in the Appendix.
 
 The notations used in the paper are as follows. For real-valued sequences $\{a_n\}$ and $\{b_n\}$, $a_n = O(b_n)$ implies that $a_n/b_n$ is bounded, and $a_n = o(b_n)$ implies that $a_n/b_n \rightarrow 0$ as $n \rightarrow \infty.$ By $a_n \lesssim b_n$, we mean that $a_n = O(b_n)$, while $a_n \asymp b_n$ means that both $a_n \lesssim b_n$ and $b_n \lesssim a_n$ hold. $a_n \prec b_n$ means $a_n = o(b_n).$ For a real vector $x = (x_1,\ldots, x_n)^T$, the $L_r$-norm of $x$ for $r > 0$ is defined as $\|x\|_r = \left( \sum_{i=1}^n|x_i|^r\right) ^{1/r}.$ We denote the cardinality of a finite set $S$ as $\#S$. The indicator function is denoted by $\Ind.$

\section{Bayesian modeling and prior specification}
\label{sec:model-and-prior}
We consider the normal sequence model (\ref{eqn:gaussian-means}) and assume that the successive differences of the means belong to a nearly-black class $l_0[s]$. As discussed earlier, frequentist procedures induce sparsity in the model via suitable regularization based procedures like penalization of the underlying parameters, whereas Bayesian methods usually induce sparsity via imposing suitable prior distributions on the same. For example, for learning a high-dimensional parameter $\theta \in \mathbb{R}^n$, a general version of a penalized optimization procedure can be written as 
$\arg \min_{\theta \in \mathbb{R}^n}\{l(\theta; y) + \pi(\theta) \},$
where $l(\theta;y)$ is the empirical risk and $\pi(\theta)$ is the penalty function. If $l(\theta;y)$ is defined as the negative log-likelihood (upto a constant) of the observations $y$, the above optimization problem becomes equivalent to finding the mode of the posterior distribution $p(\theta \mid y)$, corresponding to the prior density $p(\theta) \propto \exp\left( -\pi(\theta)\right).$ 

In our context, \cite{tibshirani2005sparsity} proposed the \emph{fused lasso} estimator $\hat{\theta}^{FL}$ that induces sparsity on both $\theta$ and $\eta$, defined as
\begin{equation}
	\hat{\theta}^{FL} = \arg \min_{\theta \in \mathbb{R}^n}\left\lbrace \dfrac{1}{2}\|y-\theta\|_2^2 + \lambda_1 \|\theta\|_1 + \lambda_2\|\eta\|_1\right\rbrace,
\end{equation}
for suitable penalty parameters $\lambda_1$ and $\lambda_2$. \cite{rinaldo2009properties} considered the $L_1$-fusion estimator $\hat{\theta}^F$ with penalization of the successive differences only, given by,
\begin{equation}
	\hat{\theta}^{F} = \arg \min_{\theta \in \mathbb{R}^n}\left\lbrace \dfrac{1}{2}\|y-\theta\|_2^2 + \lambda\|\eta\|_1\right\rbrace,
	\label{eqn:fusedlasso}
\end{equation}
where $\lambda$ is the corresponding penalty (tuning) parameter. A Bayesian equivalent of the fused lasso estimator (\ref{eqn:fusedlasso}) can be obtained by putting a $Laplace(\lambda/\sigma)$ prior on the successive differences $\eta$ and finding the corresponding posterior mode. As in a normal regression model with Bayesian lasso \citep{park2008bayesian}, the Bayesian fused lasso estimator given by the posterior mode will converge to the true $\eta_0$ at a nearly optimal rate. However, the induced posterior distribution has sub-optimal contraction rate \citep{castillo2015bayesian}, owing to insufficient prior concentration near zero.

The posterior inconsistency of the Bayesian fused lasso method motivates us to explore other approaches that would mitigate the problems leading to the undesirable behavior of the posterior distribution. Shrinkage priors qualify as naturally good choices for our problem, as they can address the dual issue of shrinking true zero parameters to zero, and retrieving the `boundaries' of the blocks effectively by placing sufficient mass on the non-zero values of successive differences of the normal means. In the normal sequence model, optimal posterior concentration has been achieved via using shrinkage prior distributions with polynomially decaying tails \citep{song2017nearly}. This is in contrast to the Laplace prior, that has exponentially light tails. 

The Horseshoe prior is a widely acclaimed choice as a shrinkage prior owing to its infinite spike at 0 and simultaneously possessing a thick tail. The tails decay like a second-order polynomial, and hence the corresponding penalty function behaves like a logarithmic penalty, and is non-convex (see \cite{carvalho2010horseshoe, bhadra2019horseshoe} for more details). The Horseshoe prior can be expressed as a scale mixture of normals with half-Cauchy prior, thus acting as a global-local shrinkage prior. We put a Horseshoe prior on the pairwise differences in the parameters $\eta_i = \theta_i - \theta_{i-1},\; i = 2,\ldots,n.$ We also need to put suitable priors on the mean parameter $\theta_1$ and the error variance $\sigma^2$. The full prior specification is given by,
\begin{eqnarray}
	\theta_1 \mid \lambda_1^2, \sigma^2 \sim  N(0, \lambda_1^2\sigma^2),\; \eta_i \mid \lambda_i^2,\tau^2, \sigma^2 \stackrel{ind}{\sim}  N(0, \lambda_i^2\tau^2\sigma^2), \; 2 \leq i \leq n,\nonumber \\
	\lambda_i \stackrel{iid}{\sim}  C^+(0,1),\; 2 \leq i \leq n,\; \tau \sim  C^+(0,1),\; \sigma^2 \sim  IG(a_\sigma, b_\sigma).
	\label{eqn:prior}
\end{eqnarray}
Here $C^+(\cdot, \cdot)$ and $IG(\cdot, \cdot)$ respectively denote the half Cauchy density and Inverse Gamma density.
The level of sparsity induced in the model is controlled by the global scale parameter $\tau$, and choosing the same is a non-trivial problem. Using a plug-in estimate for $\tau$ based on empirical Bayes method suffers from a potential danger of resulting in a degenerate solution resulting in a heavily sparse model. There are several works \citep{carvalho2010horseshoe, piironen2017sparsity, piironen2017hyperprior} that suggest effective methods regarding the choice of $\tau$. In this paper, we have proposed to take a fully Bayesian approach as suggested in \cite{carvalho2010horseshoe, piironen2017sparsity} and use a half-Cauchy prior. 

The half-Cauchy distribution can further be written as a scale-mixture of Inverse-Gamma distributions. For a random variable $X \sim C^+(0,\psi)$, we can write,
$$X^2 \mid \phi \sim IG(1/2, 1/\phi),\; \phi \sim IG(1/2,1/\psi^2).$$
Thus, the full hierarchical prior specification for our model is given by,
\begin{eqnarray}
	\theta_1 \mid \lambda_1^2, \sigma^2 \sim  N(0, \lambda_1^2\sigma^2),\;\eta_i \mid \lambda_i^2,\tau^2, \sigma^2 \stackrel{ind}{\sim}  N(0, \lambda_i^2\tau^2\sigma^2), \; 2 \leq i \leq n,\nonumber \\
	\lambda_i^2 \mid \nu_i \stackrel{ind}{\sim}  IG(1/2, 1/\nu_i),\; 2 \leq i \leq n,\;\tau^2 \mid \xi \sim  IG(1/2,1/\xi), \nonumber \\
	\nu_2,\ldots,\nu_n, \xi \stackrel{iid}{\sim}  IG(1/2,1),\; \sigma^2 \sim  IG(a_\sigma, b_\sigma)
	\label{eqn:prior-2}
\end{eqnarray}
The hyperparameters $a_\sigma$ and $b_\sigma$ may be chosen in such a way that the corresponding prior becomes non-informative. The local scale parameter $\lambda_1$ is considered to be fixed as well.

\section{Posterior computation}
\label{sec:posterior}

The conditional posterior distributions of the underlying parameters can be explicitly derived via exploring the conjugate structure. Hence, the posterior computations can be accomplished easily via Gibbs sampling. We present the conditional posterior distributions of the parameters below.

The normal means have the conditional posterior distribution
\begin{equation}
\theta_i \mid \cdots \sim N(\mu_i, \zeta_i),\; 1 \leq i \leq n,
\end{equation}

where $\mu_i$ and $\zeta_i$ are given by,
$$\zeta_1^{-1} = \dfrac{1}{\sigma^2}\left(1 + \dfrac{1}{\lambda_{2}^2\tau^2} + \dfrac{1}{\lambda_1^2} \right),\; \mu_1 = \dfrac{\zeta_1}{\sigma^2}\left(y_1 + \dfrac{\theta_{2}}{\lambda_{2}^2\tau^2}\right) ,$$
$$\zeta_i^{-1} = \dfrac{1}{\sigma^2}\left(1 + \dfrac{1}{\lambda_{i+1}^2\tau^2} + \dfrac{1}{\lambda_i^2\tau^2} \right),\; \mu_i = \dfrac{\zeta_i}{\sigma^2}\left(y_i + \dfrac{\theta_{i+1}}{\lambda_{i+1}^2\tau^2} + \dfrac{\theta_{i-1}}{\lambda_i^2\tau^2}\right),\; 2 \leq i \leq n.$$ 
Here $\lambda_{n+1}$ is considered to be infinity. The conditional posteriors for the rest of the parameters are given by,

\begin{eqnarray}
\lambda_i^2 \mid  \cdots &\sim & IG\left(1, \dfrac{1}{\nu_i} + \dfrac{(\theta_i - \theta_{i-1})^2}{2\tau^2\sigma^2}  \right),\; 2 \leq i \leq n, \nonumber \\
\sigma^2 \mid \cdots &\sim & IG\left(n + a_\sigma, b_\sigma + \dfrac{1}{2}\left[\sum_{i=1}^{n}(y_i - \theta_i)^2 + \dfrac{1}{\tau^2}\sum_{i=2}^{n}\dfrac{(\theta_i - \theta_{i-1})^2}{\lambda_i^2} + \dfrac{\theta_1^2}{\lambda_1^2}\right]\right), \nonumber \\
\tau^2 \mid \cdots &\sim & IG\left(\dfrac{n}{2}, \dfrac{1}{\xi} + \dfrac{1}{2\sigma^2}\sum_{i=2}^{n}\dfrac{(\theta_i - \theta_{i-1})^2}{\lambda_i^2}\right), \nonumber \\
\nu_i \mid \cdots &\sim & IG\left(1, 1 + \dfrac{1}{\lambda_i^2}\right),\; 2 \leq i \leq n, \nonumber \\
\xi \mid \cdots &\sim & IG\left(1, 1 + \dfrac{1}{\tau^2}\right).
\end{eqnarray}

\section{Theoretical results}
\label{sec:theory}

In this section, we present the theoretical validations of using a Horseshoe prior for fusion estimation. We first discuss the result involving posterior convergence rate of the mean parameter $\theta$ under certain assumptions and then proceed to discuss the result on posterior selection of the underlying true block structure of the mean parameter.

\begin{assumption}
	The number of true blocks in the model satisfies $s_0 \prec n/\log n$.
	\label{assump:true-blocks}
\end{assumption}

\begin{assumption}
	The true mean parameter vector $\theta_0 = (\theta_{0,1},\ldots, \theta_{0,n})^T$ and the true error variance $\sigma_0^2$ satisfy the following conditions:
	\begin{enumerate}
		\item[(i)] Define $\eta_{0,j} = \theta_{0,j} - \theta_{0,j-1},\, 2 \leq j \leq n.$ Then, $\max_j |\eta_{0,j}/\sigma_0| < L,$ where $\log L = O(\log n).$
		\item[(ii)] $\theta_{0,1}/(\lambda_1^2\sigma_0^2) + 2\log \lambda_1 = O(\log n),$ where $\lambda_1$ is the prior hyperparameter appearing in the prior for $\theta_1$ in (\ref{eqn:prior}).
	\end{enumerate}
\label{assump:true-param}
\end{assumption}

\begin{assumption}
	The global scale parameter $\tau$ in the prior specification (\ref{eqn:prior}) satisfies $\tau < n^{-(2 + b)}$ for some constant $b > 0$, and $-\log \tau = O(\log n)$.
	\label{assump:prior}
\end{assumption}

The first assumption above involves the true block size, that is ubiquitous in structure recovery problems in high-dimensions. In Assumption~\ref{assump:true-param}, we have considered $L$ as a bound on the maximum value of $|\eta_{0,j}/\sigma_0|$. Such a restriction is necessary for block recovery at a desired contraction rate. We shall see (in Lemma~\ref{lemma:priorthickness}) that the aforesaid condition on $L$, along with one of the conditions on the global scale parameter $\tau$ as in Assumption~\ref{assump:prior} would guarantee that the tail of the prior is not too sharp. This would ensure a minimum prior concentration for large non-zero values of the successive difference in the means. An equivalent prior condition cannot hold uniformly over an unbounded parameter space, thus prompting a restriction on the successive differences, and also on $\theta_{0,1}$. Similar restrictions in the context of linear and generalized regression models have been considered in the literature; for example, see \cite{wei2020contraction, song2020minimax}. Assumption~\ref{assump:prior} involves an upper bound on the global scale parameter $\tau$, that would be necessary for ensuring that the prior puts sufficient mass around zero for the successive differences in means, thus leading to effective fusion estimation via imposing sparsity. However, $\tau$ should not be too small, otherwise it would lead to degeneracy in inference by picking up too sparse a model. The lower bound on $\tau$, along with the condition on $L$ in Assumption~\ref{assump:true-param} guarantees that the Horseshoe prior is `thick' at non-zero parameter values, so that it is not too sharp. We now present our main result on posterior contraction rate.

\begin{theorem}
	Consider the Gaussian means model (\ref{eqn:gaussian-means}) with prior specification as in (\ref{eqn:prior}), and suppose that assumptions \ref{assump:true-blocks}, \ref{assump:true-param} and \ref{assump:prior} hold. Then the posterior distribution of $\theta$, given by $\Pi^n(\cdot \mid y)$, satisfies
	\begin{equation}
		\Pi^n(\|\theta - \theta_0\|_2/\sqrt{n} \geq M\sigma_0 \epsilon_n \mid y) \rightarrow 0, \; \mathrm{as}\, n \rightarrow \infty,
	\end{equation}
in probability or in $L_1$ wrt the probability measure of $y$, for $\epsilon_n \asymp \sqrt{s_0\log n/n}$ and a constant $M > 0$.
\label{theorem:posteriorconvergencerate}
\end{theorem}

The above result implies that the posterior convergence rate for $\|\theta - \theta_0\|_2/\sqrt{n}$ is of the order $\sigma_0 \sqrt{s_0\log n/n}$. When the exact piecewise constant structure is known, the proposed Bayesian fusion estimation method achieves the optimal convergence rate $O(\sigma_0\sqrt{s/n})$ upto a logarithmic term in $n$. The posterior convergence rate also adapts to the (unknown) size of the pieces. The rate directly compares with the Bayesian fusion estimation method as proposed in \cite{song2020bayesian}.

The Horseshoe prior (and other global-local shrinkage priors) are continuous shrinkage priors, and hence exact block structure recovery is not possible. However, we can consider discretization of the posterior samples via the posterior projection of the samples $(\theta, \sigma)$ to a discrete set $S(\theta, \sigma) = \{2 \leq j \leq n: |\theta_j - \theta_{j-1}|/\sigma < \epsilon_n/n\}.$ The number of false-positives resulting from such a discretization can be expressed as the cardinality of the set $A(\theta, \sigma) = S^c(\theta, \sigma) - \{2 \leq j \leq n: \theta_{0,j} - \theta_{0,j-1} \neq 0\}.$ The induced posterior distribution of $S(\theta,\sigma)$ (and hence that of $A(\theta, \sigma)$) can be shown to be `selection consistent', in the sense that the number of false-positives as defined above is bounded in probability. We formally present this result below.

\begin{theorem}
	Under the assumptions of Theorem~\ref{theorem:posteriorconvergencerate}, the posterior distribution of $A(\theta, \sigma)$ satisfies
	\begin{equation}
		\Pi^n(\#A(\theta,\sigma) > Ks_0 \mid y ) \rightarrow 0,
	\end{equation}
in probability or in $L_1$ wrt the measure of $y$ for some fixed constant $K > 0$.
\label{theorem:structure-recovery}
\end{theorem}

Note that the thresholding rule depends on the posterior contraction rate, that involves the number $(s_0)$ of true blocks. However, in practical scenarios, knowledge of $s_0$ may not be readily available. To tackle such situations, we propose to use a threshold analogous to the concept of shrinkage weights as in the sparse normal sequence model. We outline the details in the appendix.

\section{Simulation studies}
\label{sec:simulation}

To evaluate the performance of our method and compare with other competing approaches, we carry out simulation studies for varying signal and noise levels. Three different types of piecewise constant functions of length $n = 100$ are considered -- (i) 10 evenly spaced signal pieces, with each of the pieces having length 10, (ii) 10 unevenly spaced signal pieces, with the shorter pieces each having length 5, and (iii) 10 very unevenly spaced signal pieces, with the shorter pieces each having length 2. The response variable $y$ is generated from an $n$-dimensional Gaussian distribution with mean $\theta_{0}$ and variance $\sigma^2$, where $\theta_0$ is the true signal vector, and $\sigma \in \{0.1, 0.3, 0.5\}.$ 

We estimate the true signal using our Horseshoe prior based fusion estimation approach, and compare with three other approaches -- (a) Bayesian $t$-fusion method, as proposed in \cite{song2020bayesian}, (b) Bayesian fusion approach based on a Laplace prior, and (c) the $L_1$-fusion method. For the Bayesian methods, we consider 5000 MCMC samples, with initial 500 samples as burn-in. The MCMC-details including Gibbs sampler updates for the $t$-fusion and Laplace fusion approaches, and the choice of the corresponding scale parameters for the above priors are taken as suggested in \cite{song2020bayesian}. The hyperparameter values for the prior on error variance are taken as $a_{\sigma} = b_{\sigma} = 0.5$ across all the Bayesian methods, with the local scale parameter $\lambda_1 = 5$. For the frequentist fused lasso method based on $L_1$-fusion penalty, we use the \texttt{genlasso} package in \texttt{R}, and choose the fusion penalty parameter using a 5-fold cross-validation approach. To evaluate the performance of the estimators, we use the Mean Squared Error (MSE) and the adjusted MSE, respectively defined as, 
\begin{eqnarray}
	\mathrm{MSE} &=& \|\hat{\theta} - \theta_0\|_2^2/n,\nonumber \\
	\mathrm{adj. MSE} &=& \|\hat{\theta} - \theta_0\|_2^2/\|\theta_0\|_2^2,\nonumber 
\end{eqnarray}
where $\hat{\theta}$ is the estimated signal, given by the posterior mean in case of the Bayesian methods and the fused lasso estimate for the frequentist method. All the computations were performed in \texttt{R} on a laptop having an Intel(R) Core(TM) i7-10750H CPU @ 2.60GHz with 16GB RAM and a 64-bit OS, x64-based processor. The \texttt{R} codes to implement our proposed method are available at \url{https://github.com/sayantanbiimi/HSfusion}.

The summary measures for the performance of the four methods using MSE, adjusted MSE and their associated standard errors (in brackets) based on 100 Monte-Carlo replications are presented in Table~\ref{table:simu-results}. We find that our proposed Horseshoe based fusion estimation method has excellent signal estimation performance across all the different types of signal pieces and noise levels. Within a Bayesian framework, the performance of the Laplace based fusion estimation method is not at all promising. Though the estimation performances for the Horseshoe-fusion, $t$-fusion and fused lasso are comparable in case of low noise levels, our proposed method performs much better in the higher noise scenario. Additionally, the Bayesian methods provide credible bands as well that can be utilized further for uncertainty quantification. In that regard, we observe from Figures~\ref{fig1}, \ref{fig2} and \ref{fig3} that the Horseshoe based estimates have comparatively narrower credible bands as compared to other Bayesian competing methods considered here. Overall, we could successfully demonstrate the superiority of using a heavy-tailed prior for successive differences in the signals in contrast to an exponentially lighter one. In addition to that, using a prior with a comparatively sharper spike at zero results in better signal and structure recovery, especially in situations where the noise level is higher.

\begin{sidewaystable}[h]
	\small
	\begin{tabular}{ll|ll|ll|ll|ll}
		\hline 
		&       & \multicolumn{2}{c}{HS-fusion} & \multicolumn{2}{c}{$t$-fusion} & \multicolumn{2}{c}{Laplace fusion} & \multicolumn{2}{c}{$L_1$ fusion}  \\
		Signal      & $\sigma$ & MSE            &  adj MSE & MSE            &  adj MSE      & MSE              & adj MSE         & MSE          & adj MSE         \\
		\hline 
		& 0.1 & 0.002 (0.000)     & 0.000 (0.000)     & 0.002 (0.000)     & 0.001 (0.000)     & 0.504 (0.002) & 0.114 (0.004) & 0.005 (0.003) & 0.001 (0.001) \\
	Even	& 0.3 & 0.020 (0.001) & 0.004 (0.000)     & 0.030 (0.002) & 0.007 (0.001) & 0.516 (0.002) & 0.116 (0.004) & 0.035 (0.019) & 0.008 (0.004) \\
		& 0.5 & 0.062 (0.002) & 0.014 (0.000) & 0.133 (0.009) & 0.030 (0.002) & 0.539 (0.003) & 0.121 (0.004) & 0.087 (0.039) & 0.019 (0.009) \\
		&&&&&&&&& \\
		& 0.1 & 0.002 (0.000)     & 0.001 (0.000)     & 0.002 (0.000)     & 0.001 (0.000)  & 0.566 (0.002) & 0.249 (0.003) & 0.005 (0.003) & 0.002 (0.001) \\
	Uneven	& 0.3 & 0.019 (0.001) & 0.008 (0.000)     & 0.036 (0.002) & 0.016 (0.001) & 0.575 (0.002) & 0.253 (0.004) & 0.037 (0.019) & 0.016 (0.008) \\
		& 0.5 & 0.060 (0.002) & 0.026 (0.001) & 0.212 (0.012) & 0.094 (0.006) & 0.596 (0.003) & 0.263 (0.004) & 0.090  (0.044) & 0.039 (0.019) \\
		&&&&&&&&&\\
		& 0.1 & 0.002 (0.000)  & 0.002 (0.000)     & 0.005 (0.000)     & 0.005 (0.000)     & 0.479 (0.001) & 0.520  (0.001) & 0.009 (0.002) & 0.010  (0.002) \\
	V. Uneven	& 0.3 & 0.020 (0.001) & 0.022 (0.001) & 0.114 (0.007) & 0.124 (0.007) & 0.489 (0.001) & 0.531 (0.001) & 0.066 (0.024) & 0.072 (0.026) \\
		& 0.5 & 0.064 (0.002) & 0.070 (0.002) & 0.581 (0.020)  & 0.631 (0.021) & 0.507 (0.002) & 0.551 (0.002) & 0.143 (0.057) & 0.155 (0.062)\\
		\hline 
			\end{tabular}
	\caption{MSE and adjusted MSE values (with associated standard errors in parentheses) for Horseshoe-fusion, $t$-fusion, Laplace fusion, and fused lasso method, when the true signal is evenly spaced (``Even''), unevenly spaced (``Uneven''), and very unevenly spaced (``V. Uneven'').}
	\label{table:simu-results}
\end{sidewaystable}

\begin{figure}
	\begin{tabular}{lccc}
		& Evenly spaced pieces &  Unevenly spaced pieces &  Very unevenly spaced pieces \\
		True&&&\\
		&\includegraphics[width=40mm]{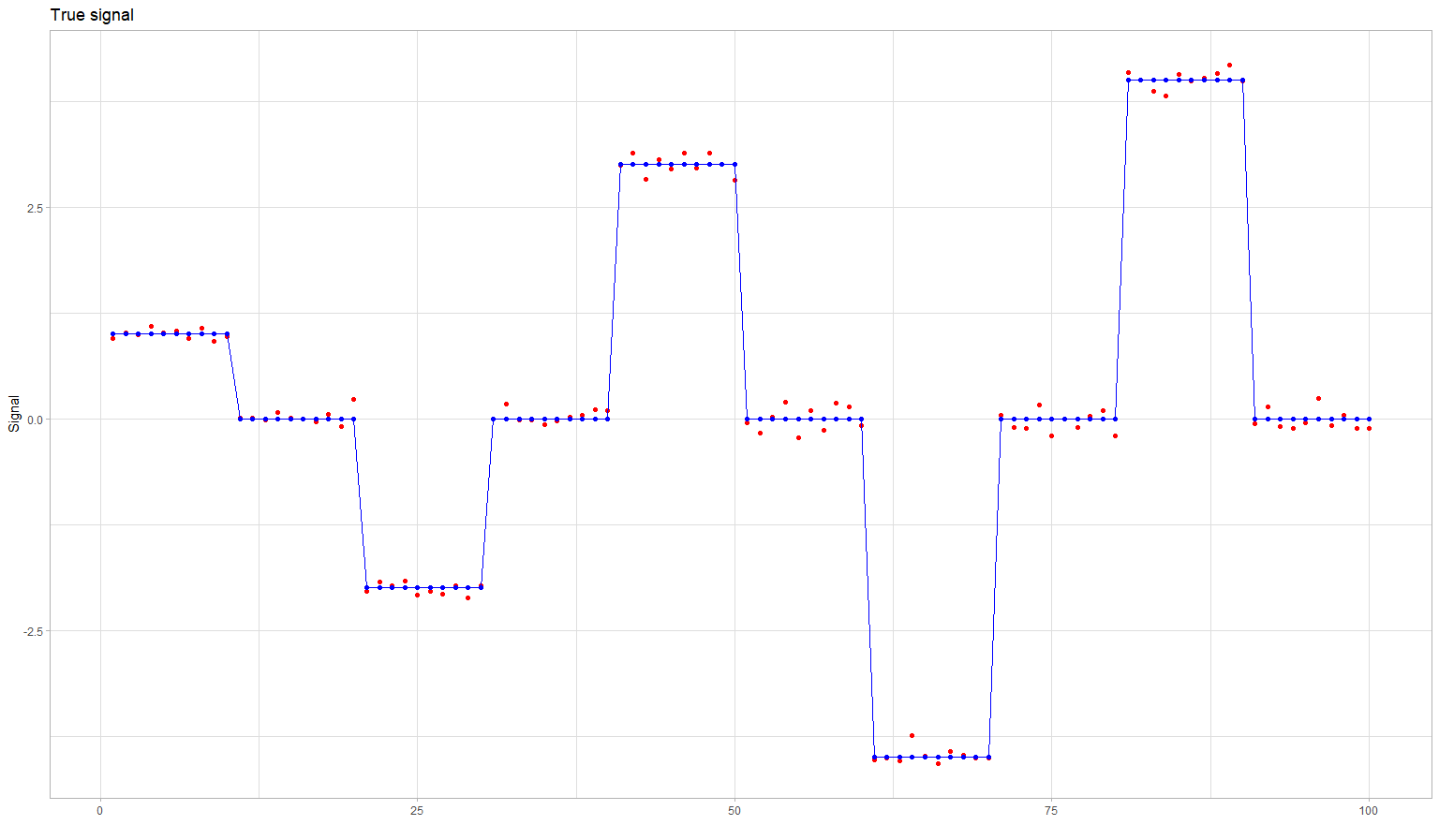} &   \includegraphics[width=40mm]{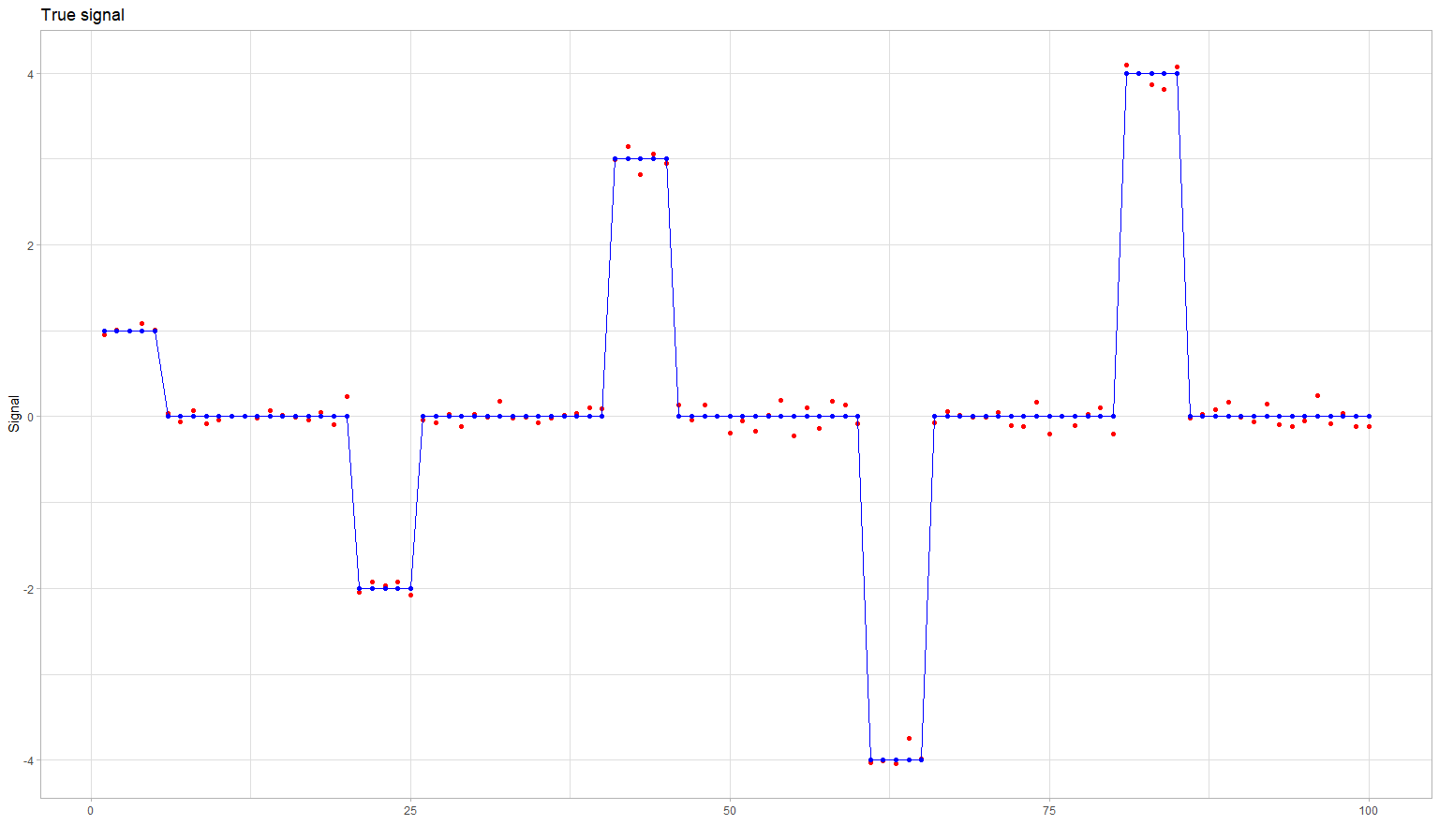} &   \includegraphics[width=40mm]{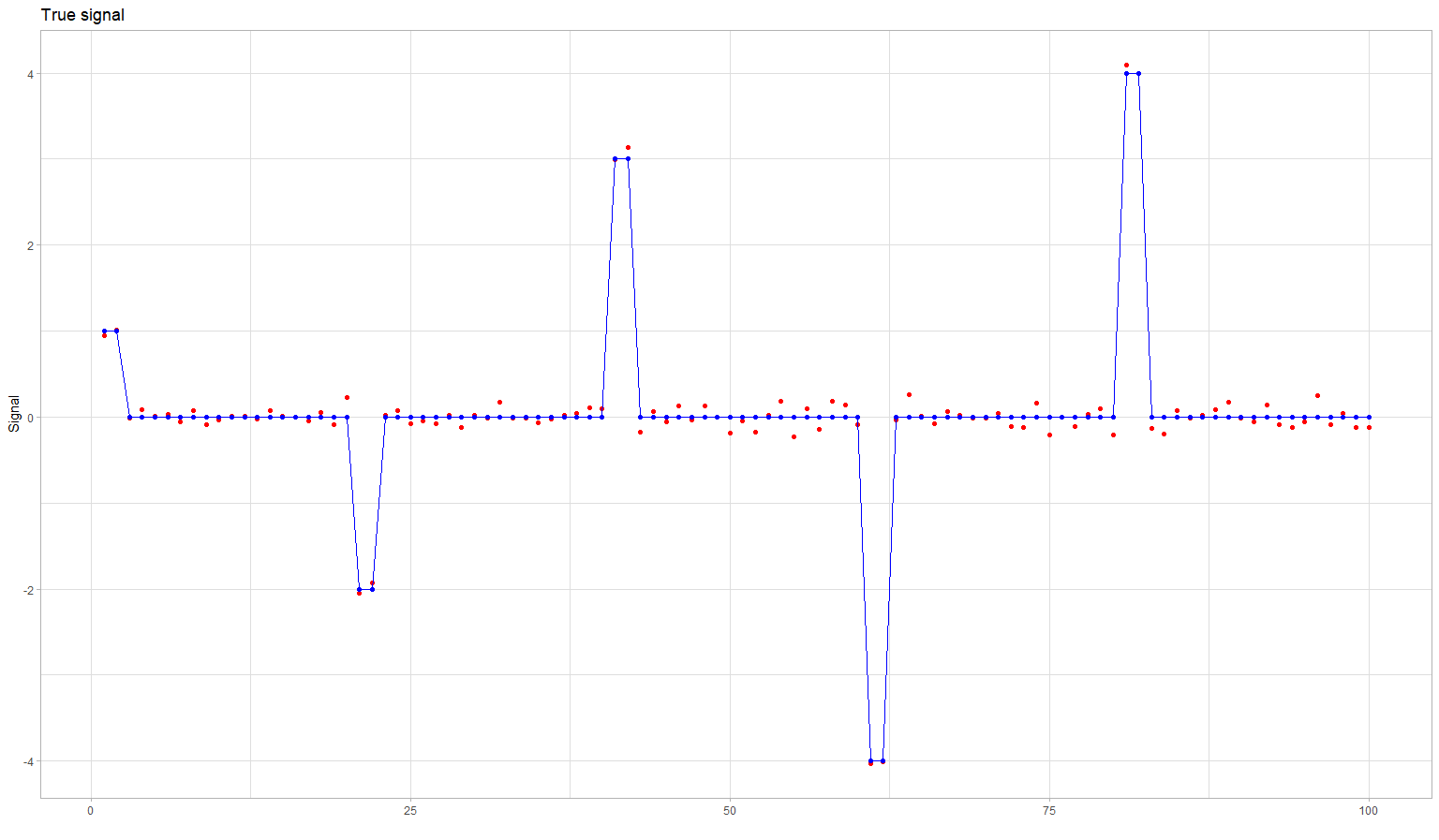} \\
		HS &&&\\
		& \includegraphics[width=40mm]{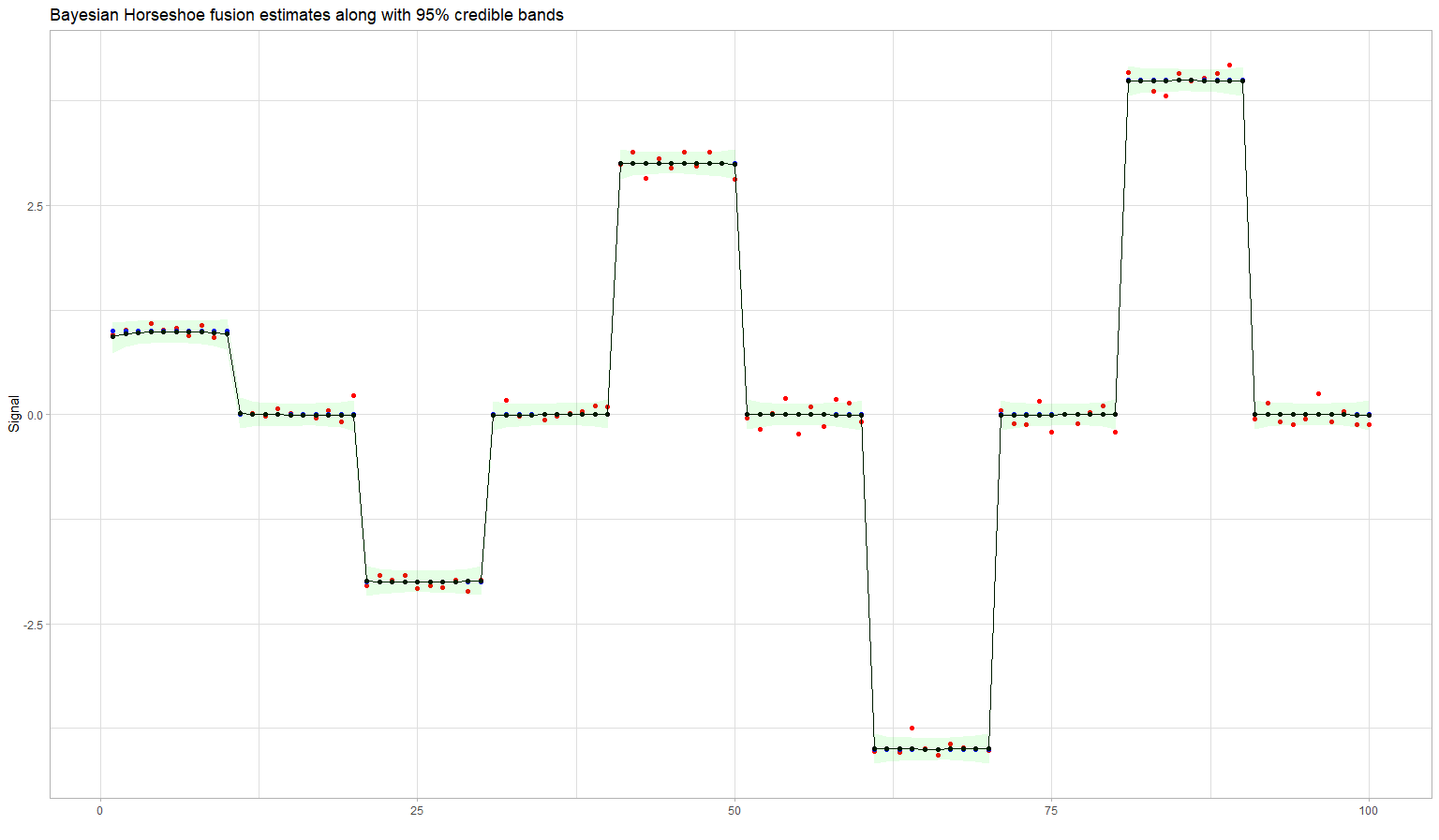} &   \includegraphics[width=40mm]{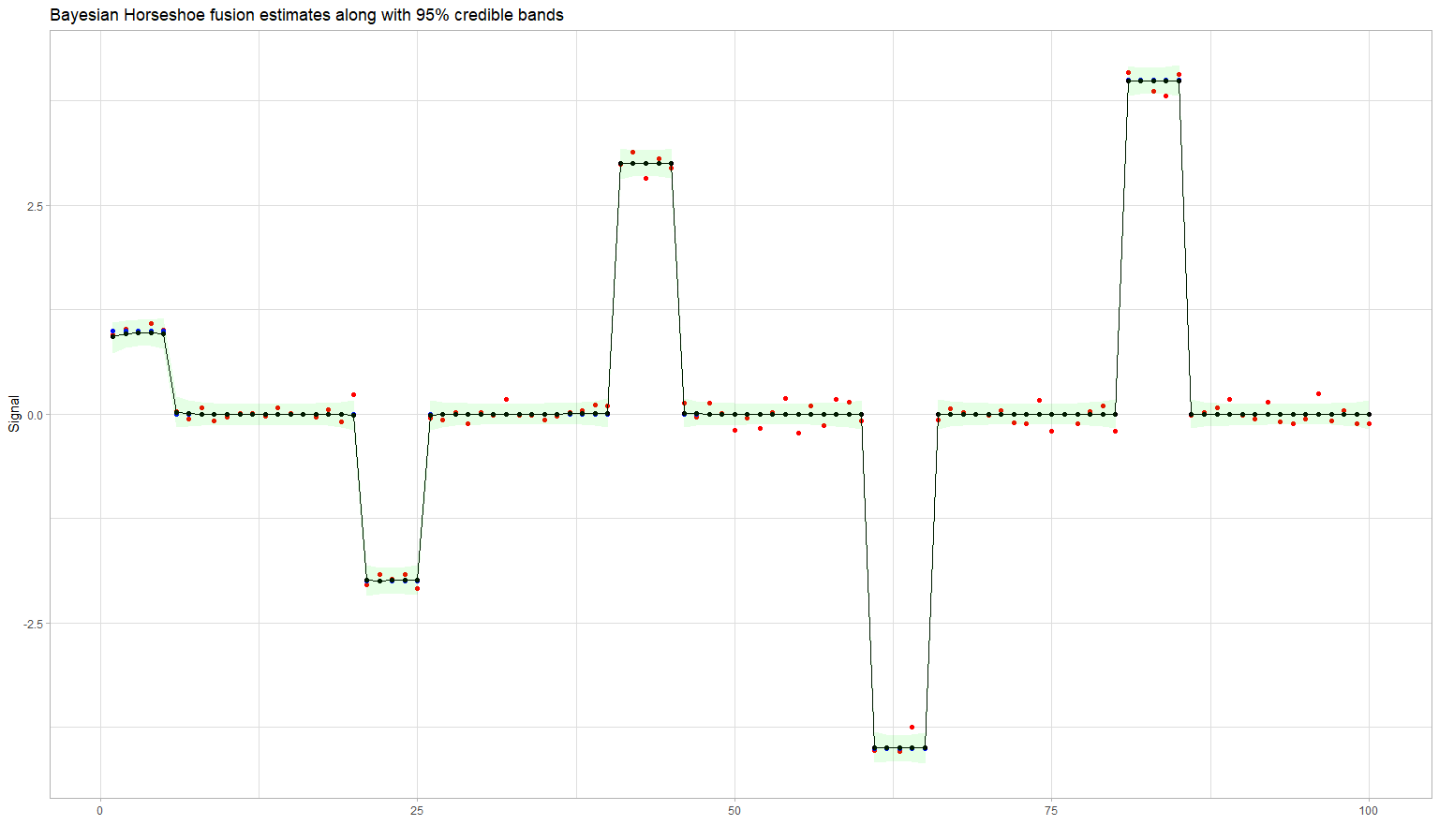} &   \includegraphics[width=40mm]{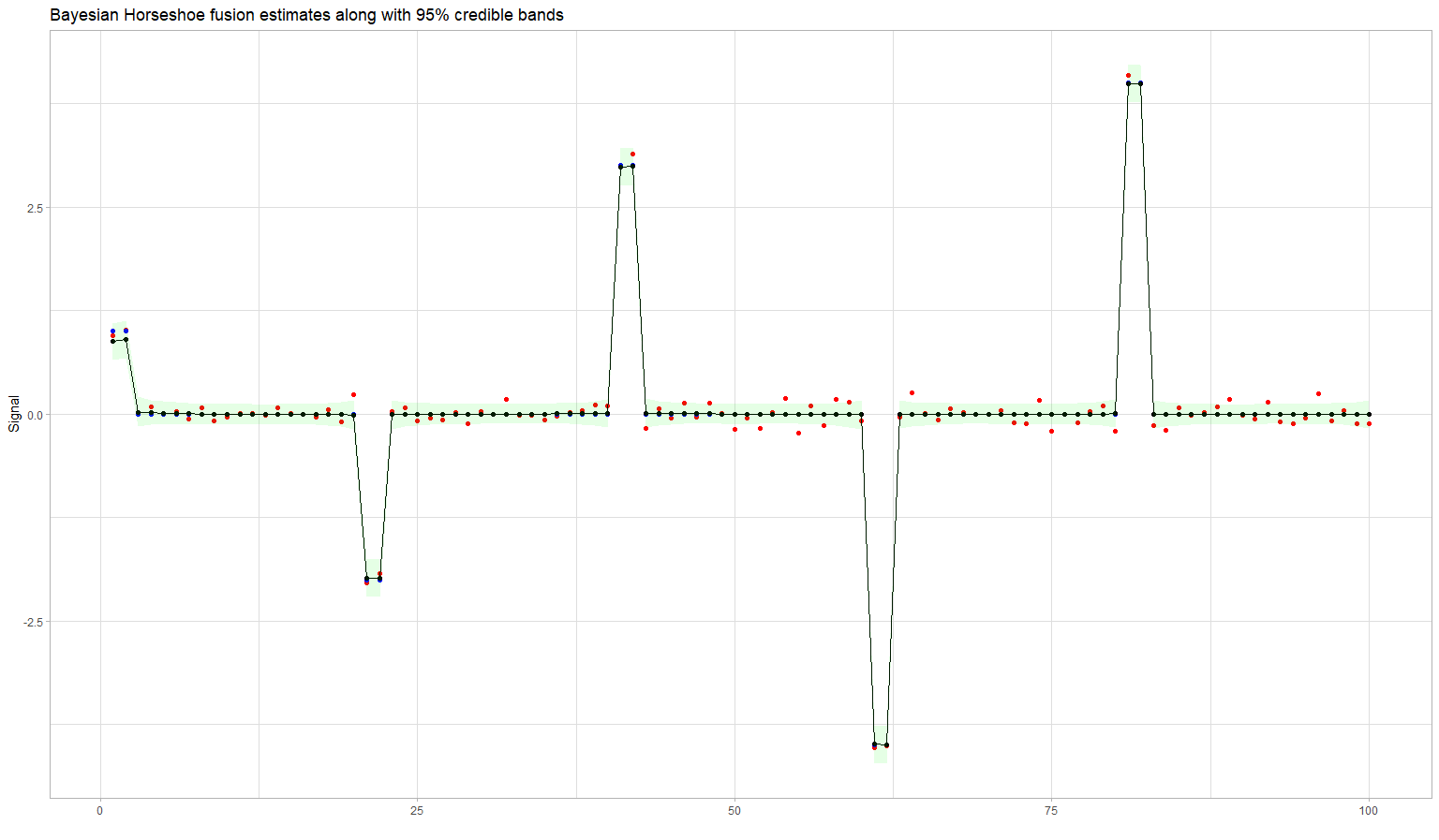} \\
	$t$ &&&\\
	&\includegraphics[width=40mm]{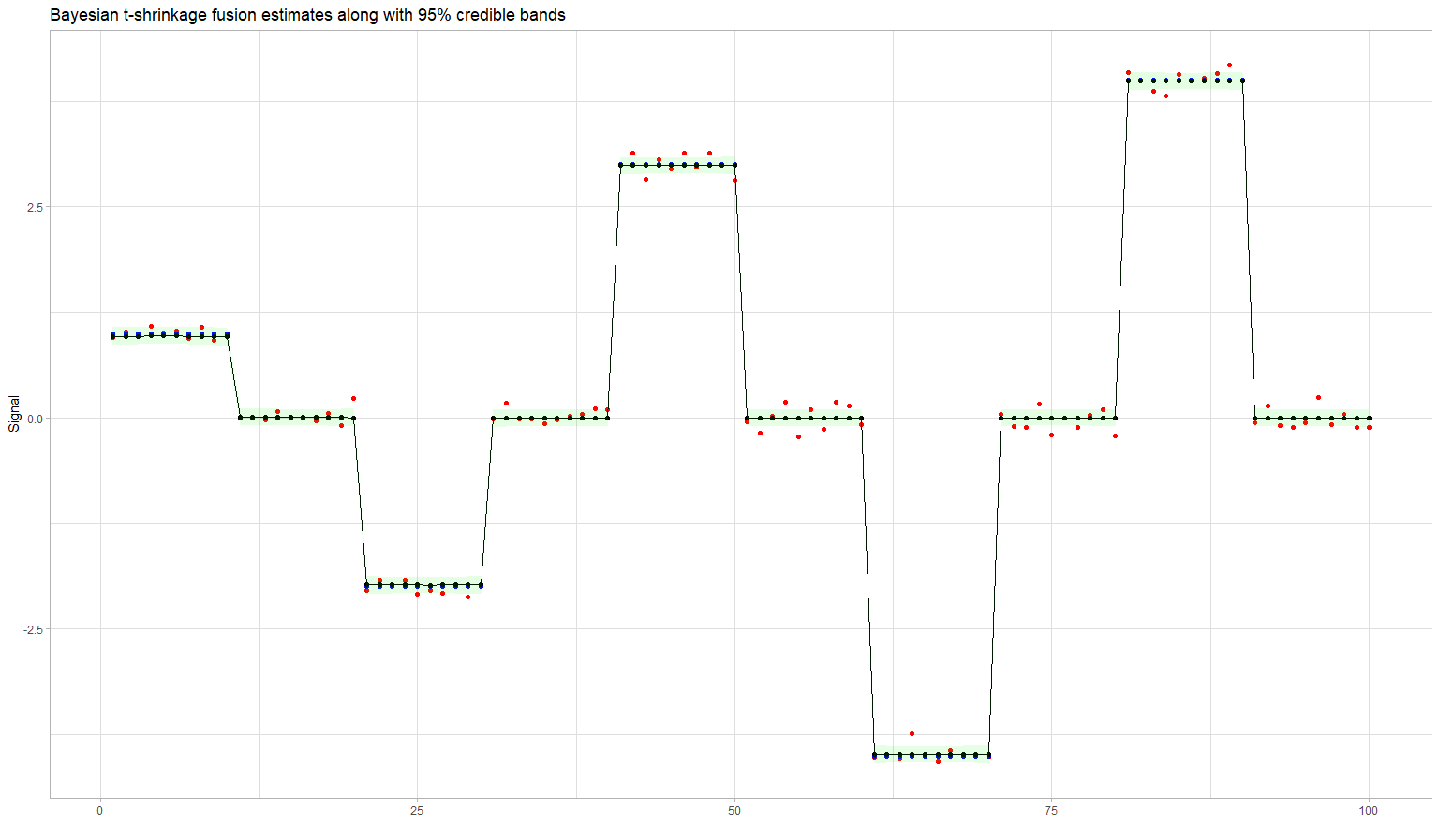} &   \includegraphics[width=40mm]{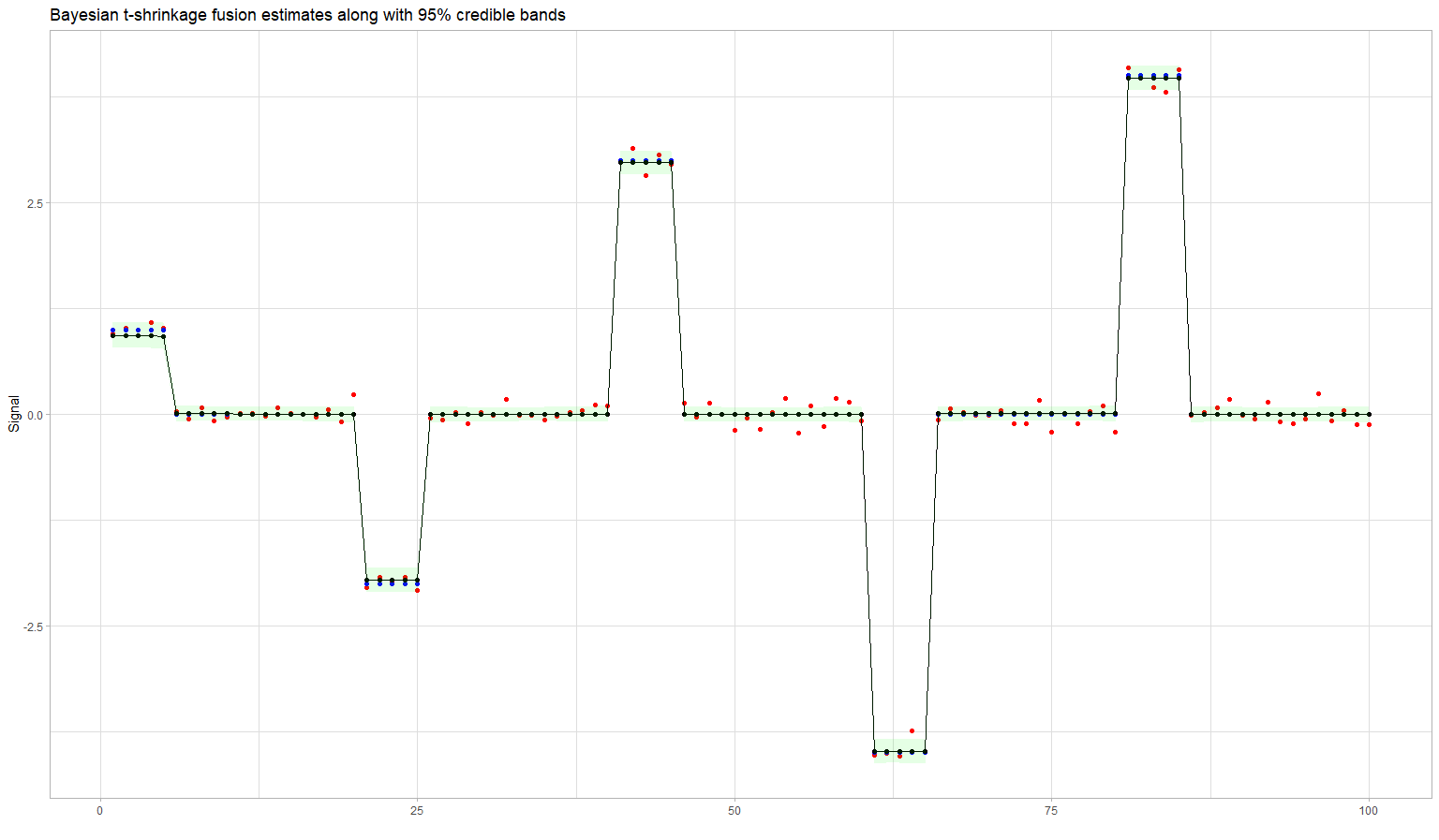} &   \includegraphics[width=40mm]{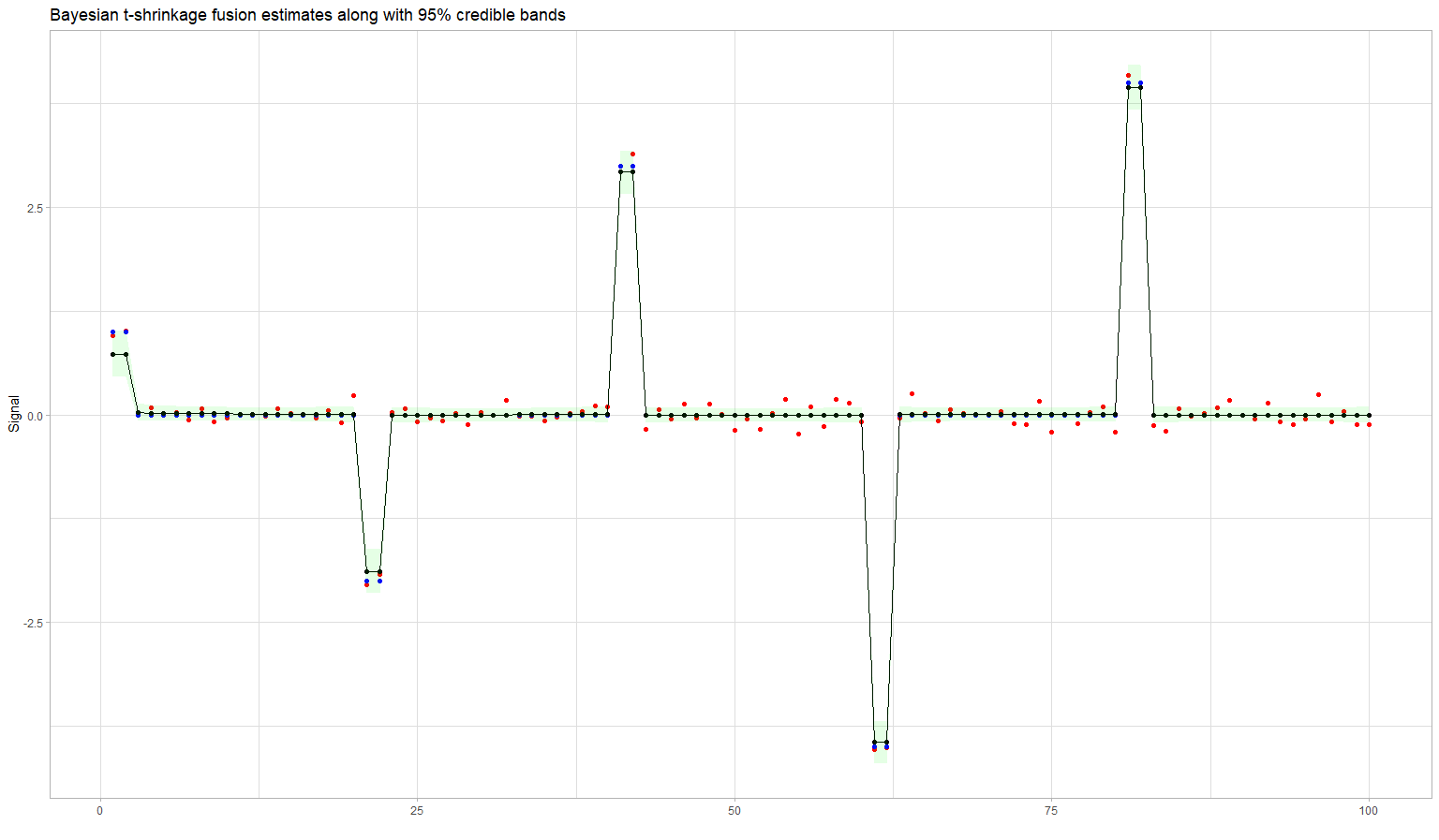} \\
	Laplace &&&\\
	 &	\includegraphics[width=40mm]{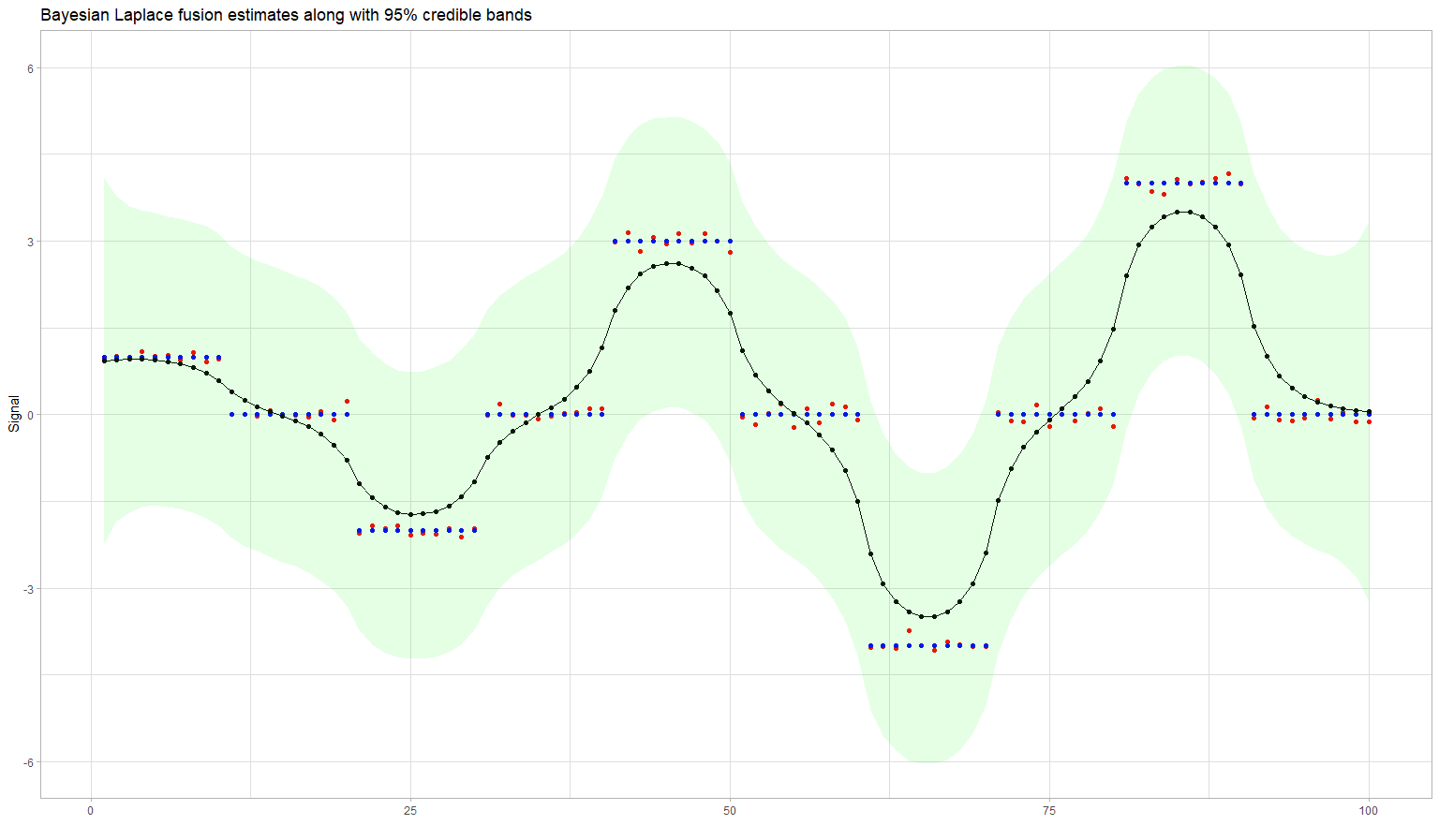} &   \includegraphics[width=40mm]{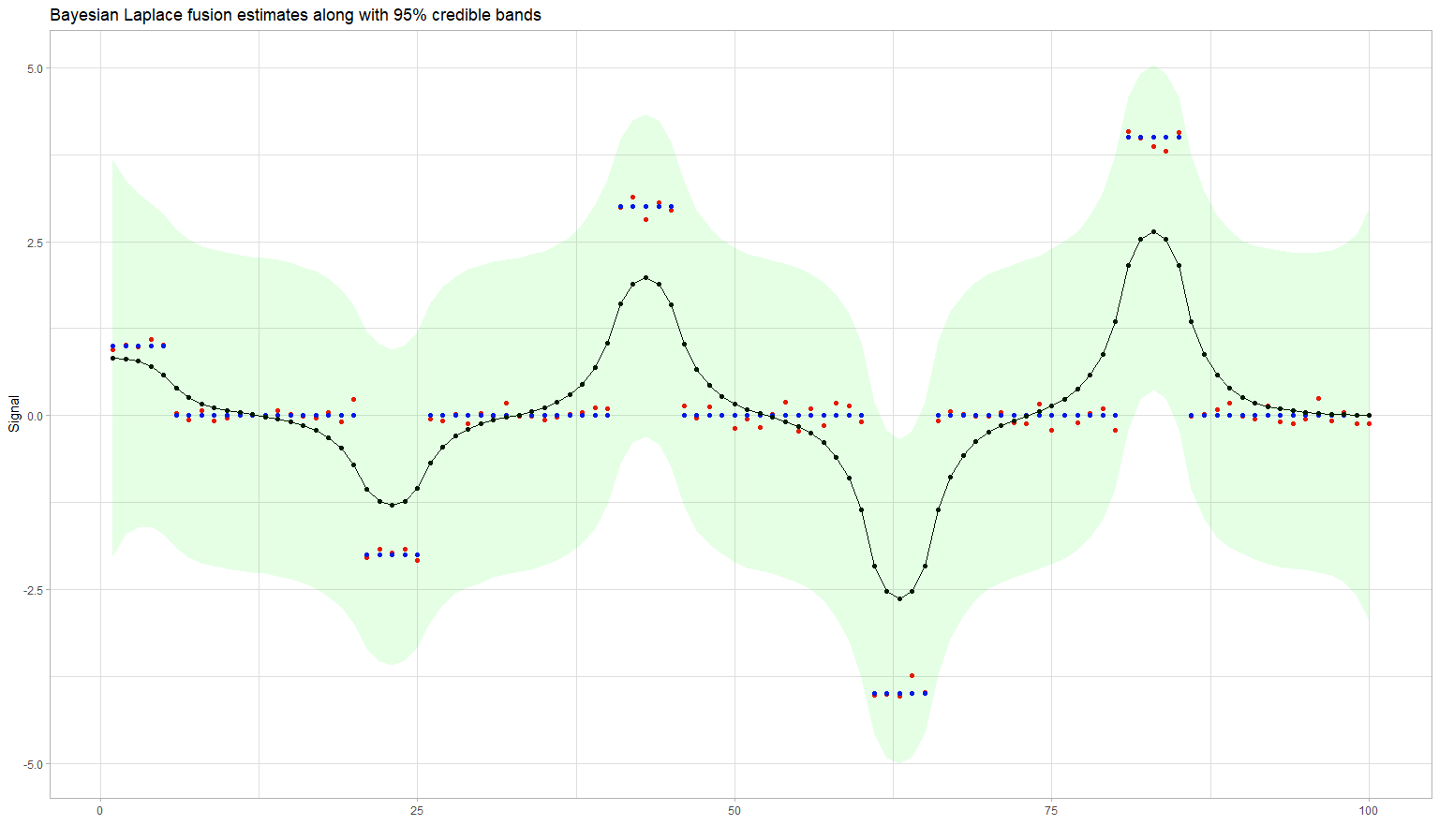} &   \includegraphics[width=40mm]{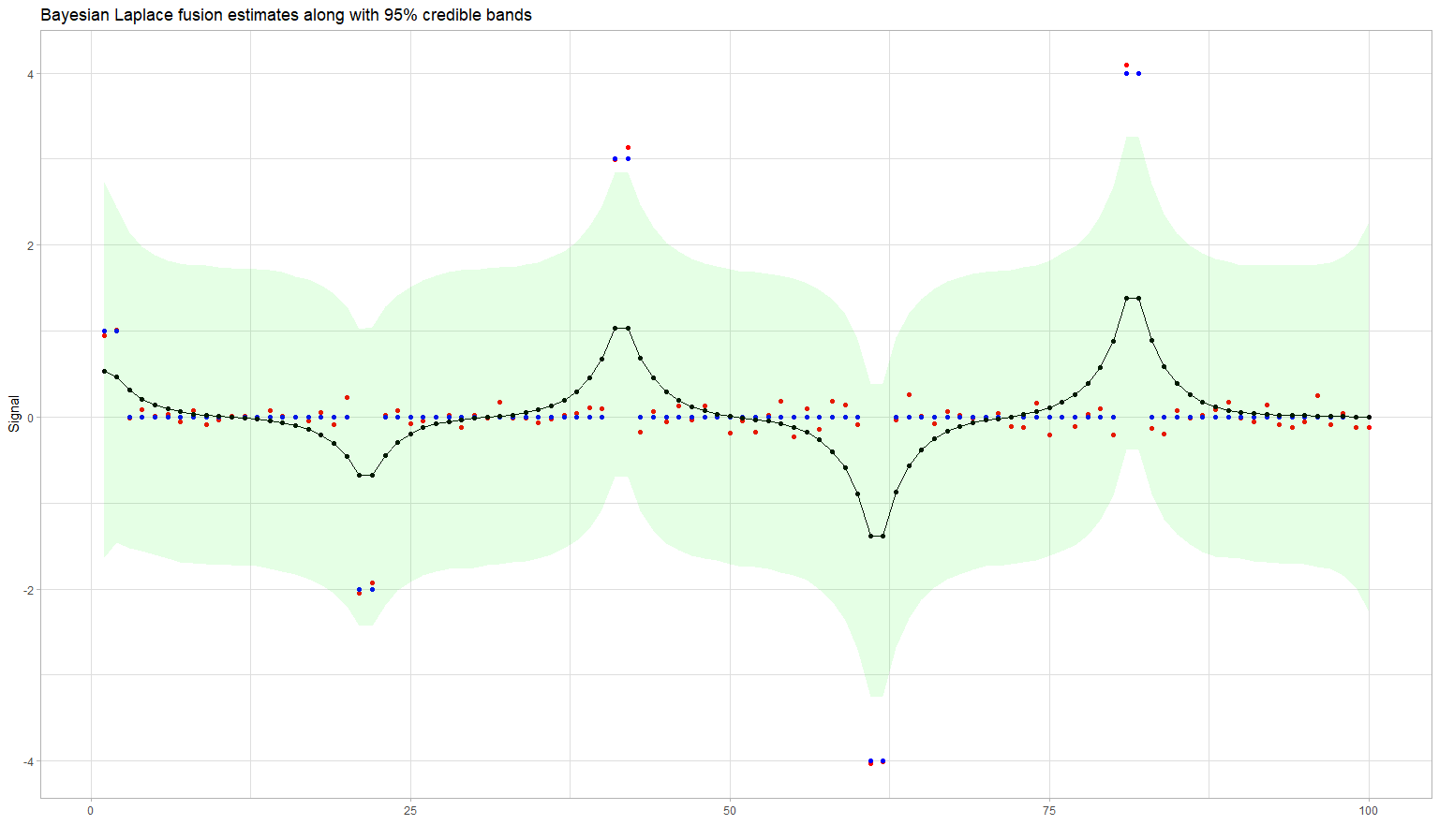} \\
	Fused &&&\\
	 & \includegraphics[width=40mm]{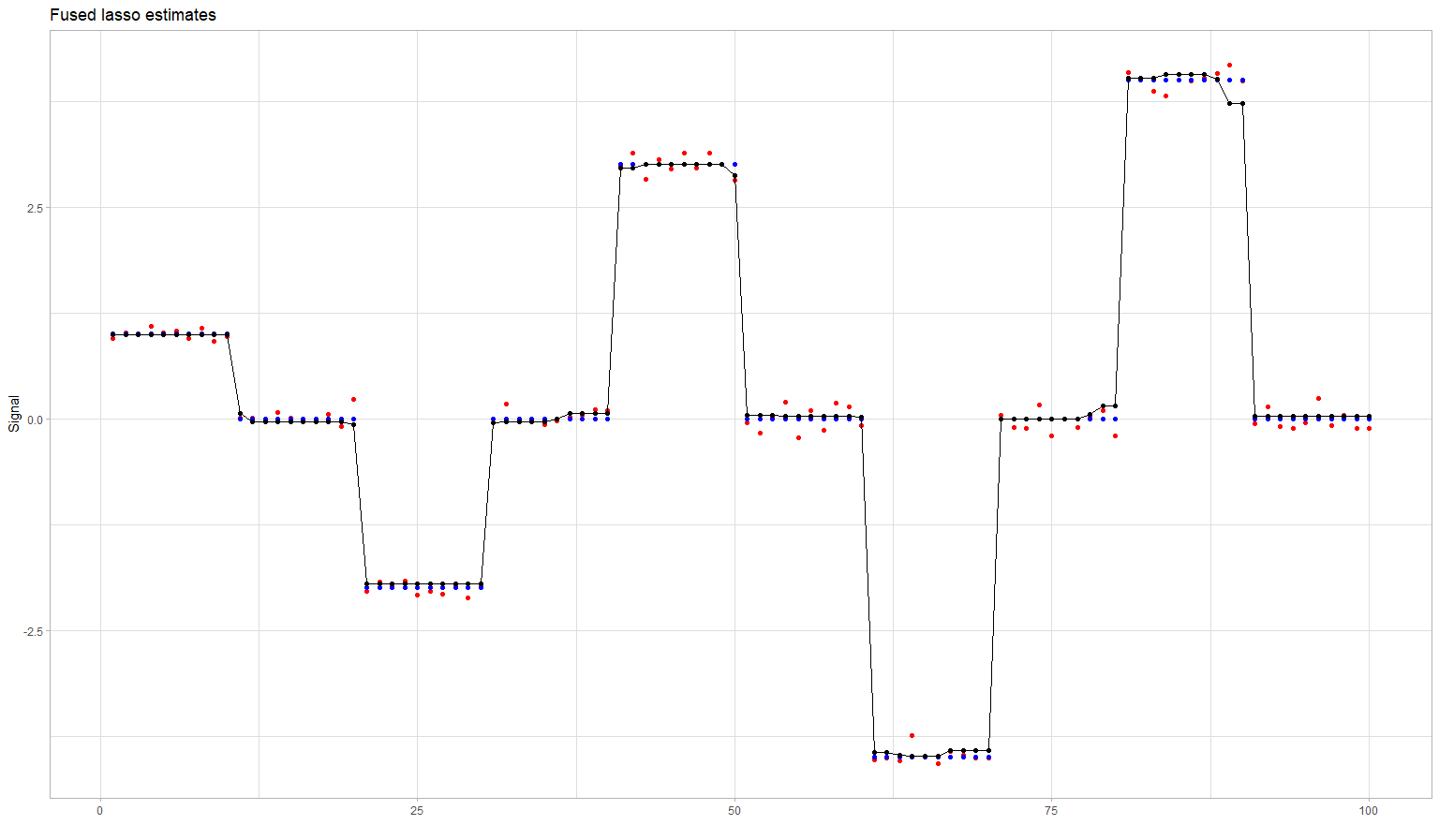} &   \includegraphics[width=40mm]{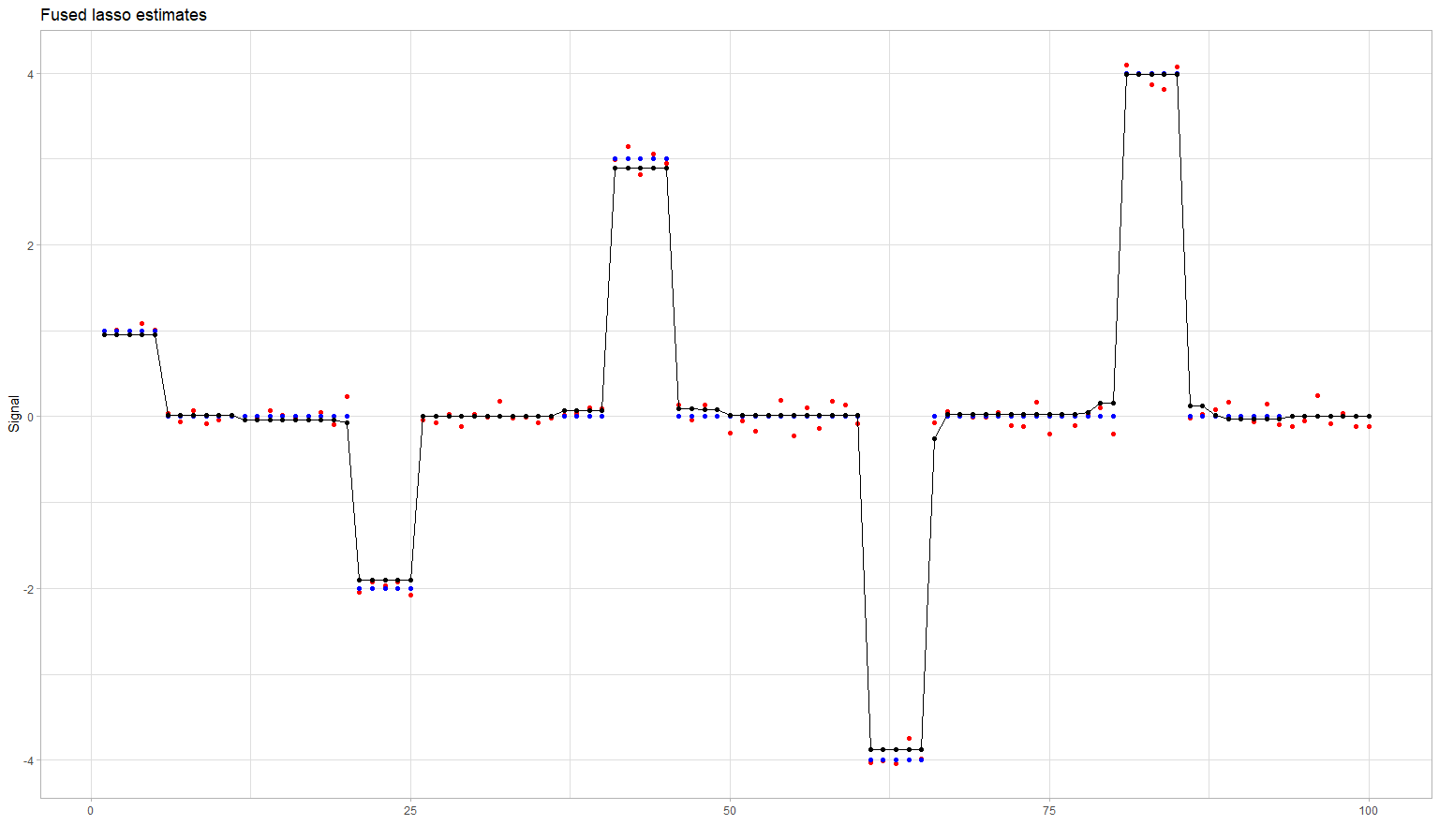} &   \includegraphics[width=40mm]{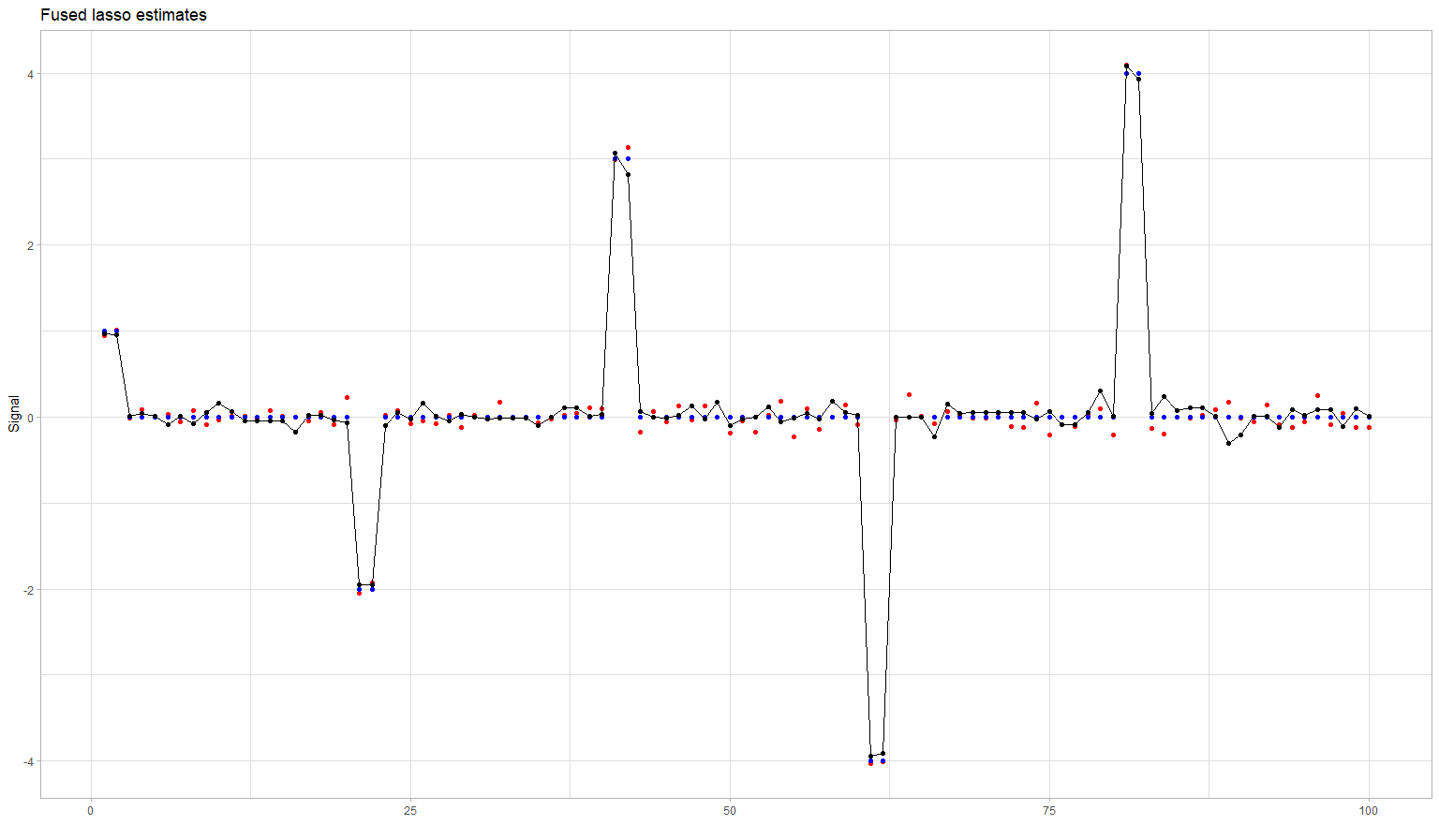} \\	
		\end{tabular}
	\caption{Fusion estimation performance with differently spaced signals and error sd $\sigma = 0.1.$ Observations are represented in red dots, true signals in blue dots, point estimates in black dots, and 95\% credible bands of the Bayesian procedures in green.}
	\label{fig1}
\end{figure}

\begin{figure}
	\begin{tabular}{lccc}
		& Evenly spaced pieces &  Unevenly spaced pieces &  Very unevenly spaced pieces \\
		True &&&\\
		&\includegraphics[width=40mm]{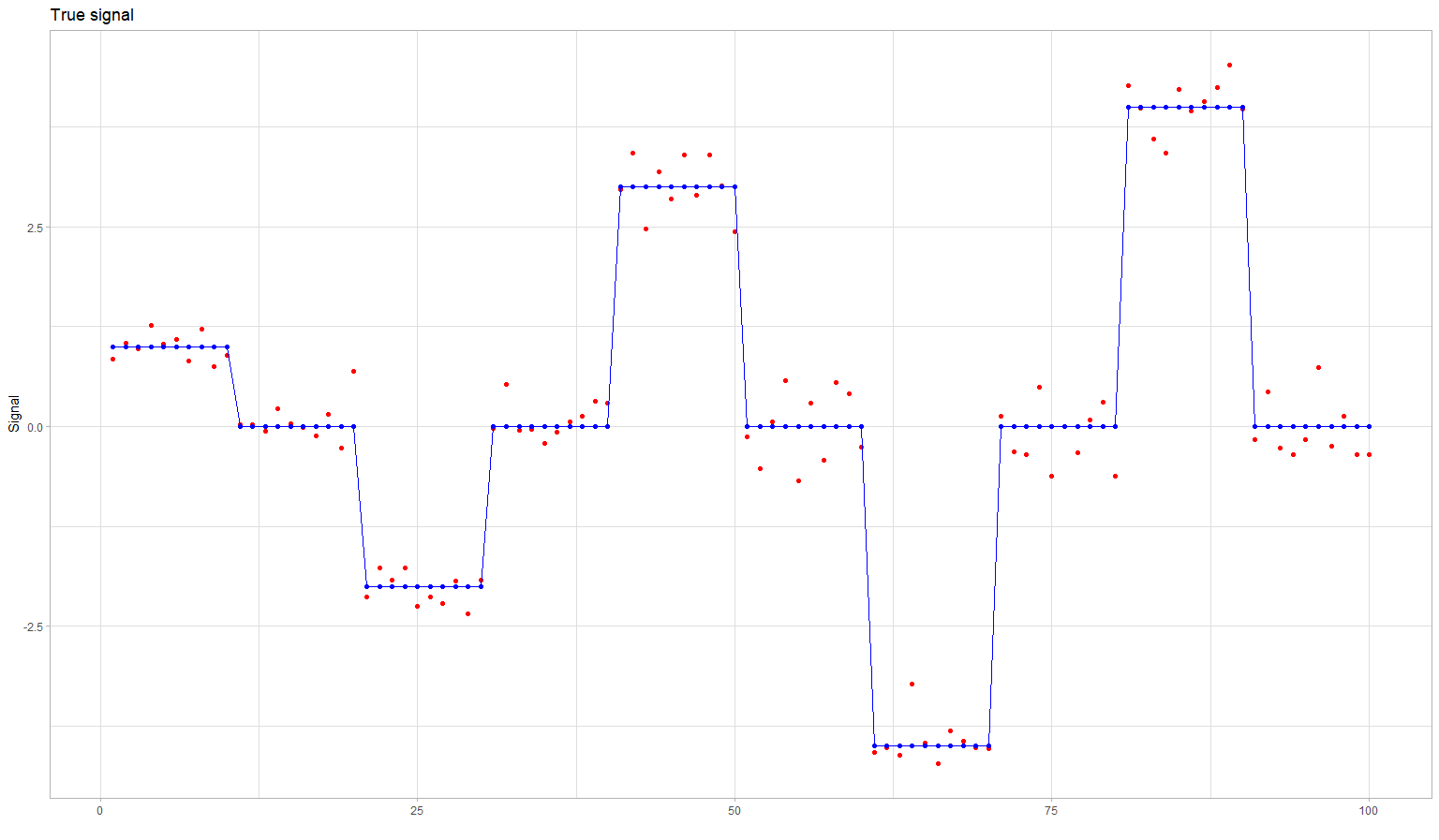} &   \includegraphics[width=40mm]{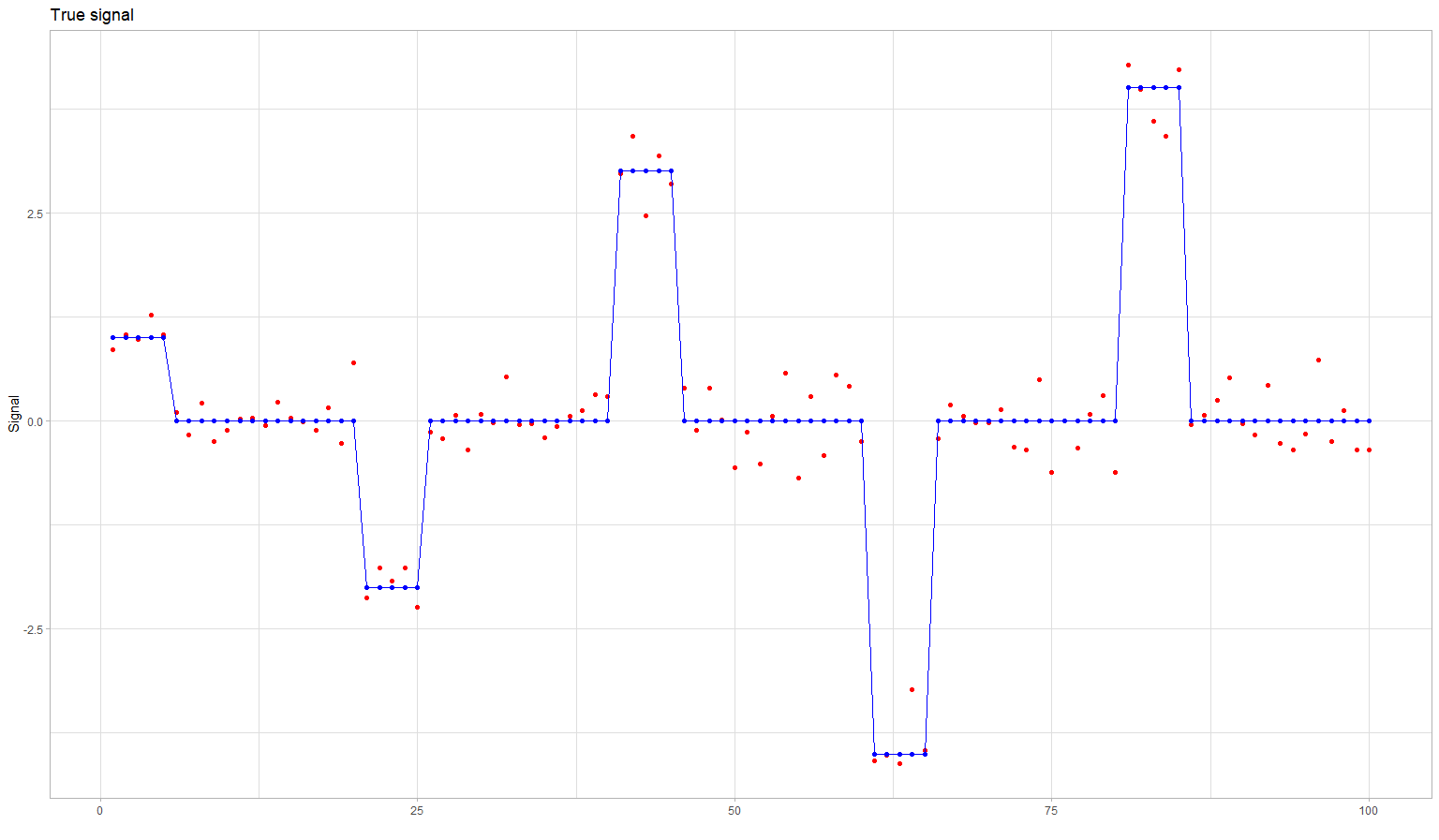} &   \includegraphics[width=40mm]{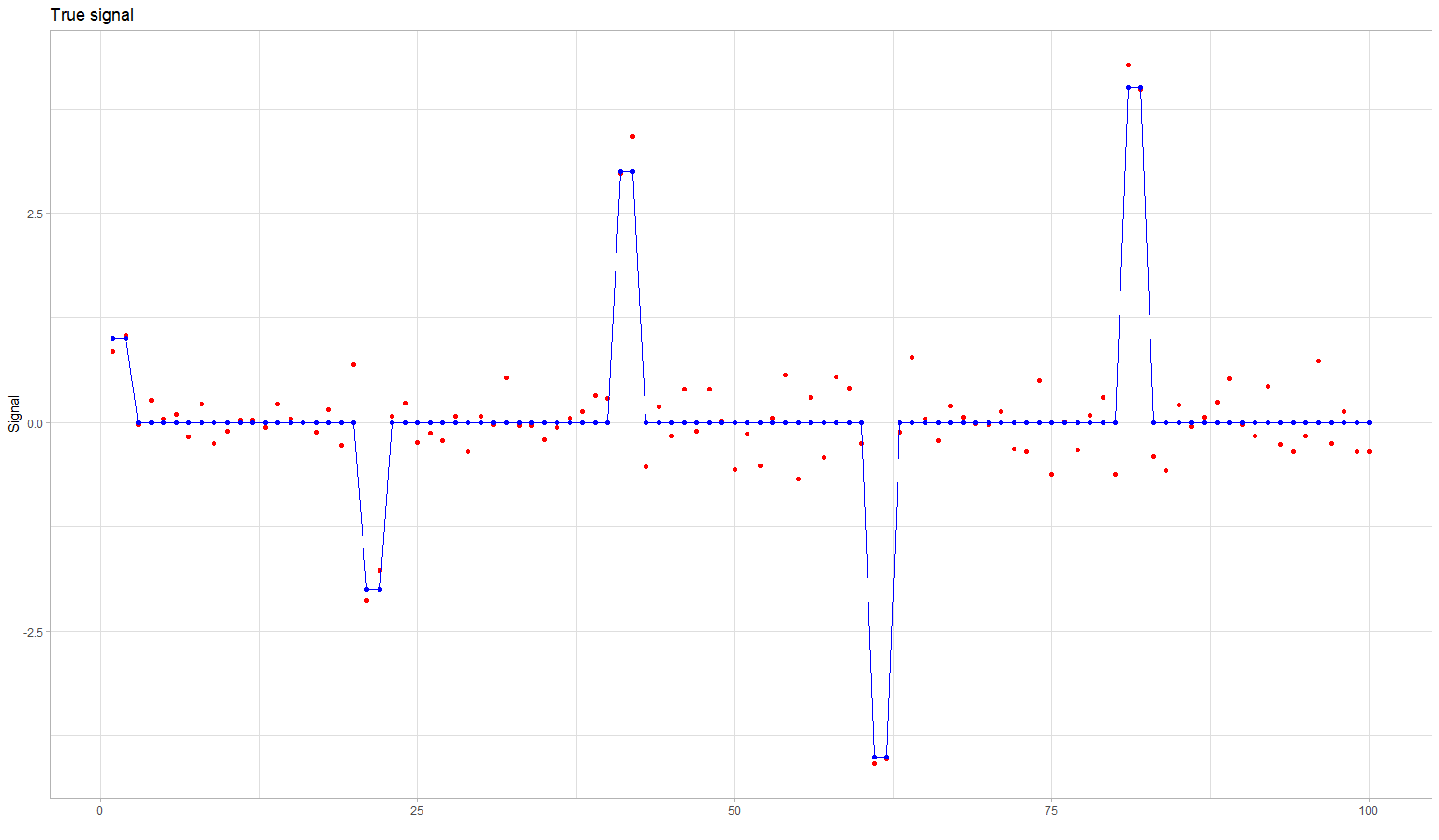} \\
		HS &&&\\
		& \includegraphics[width=40mm]{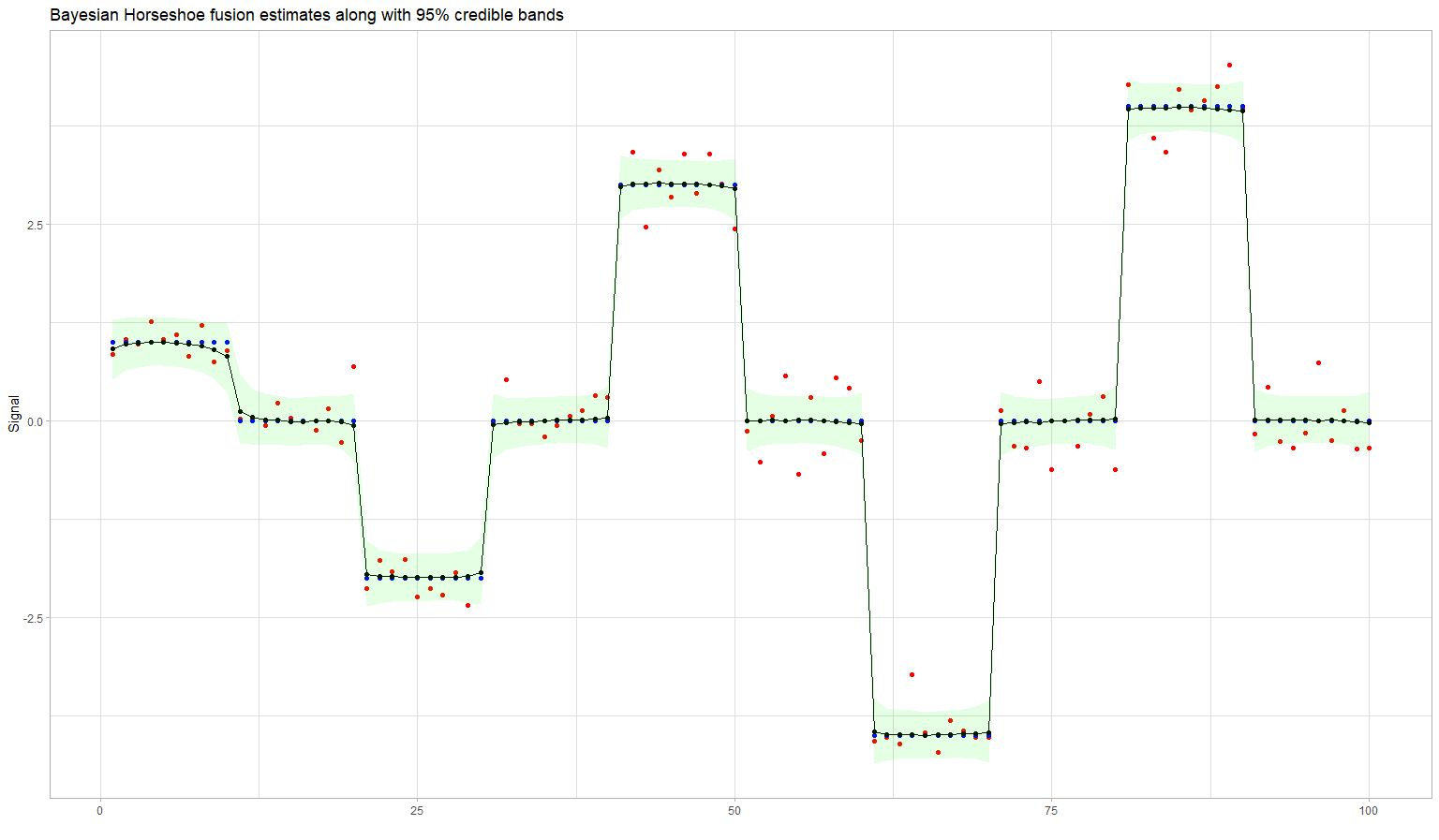} &   \includegraphics[width=40mm]{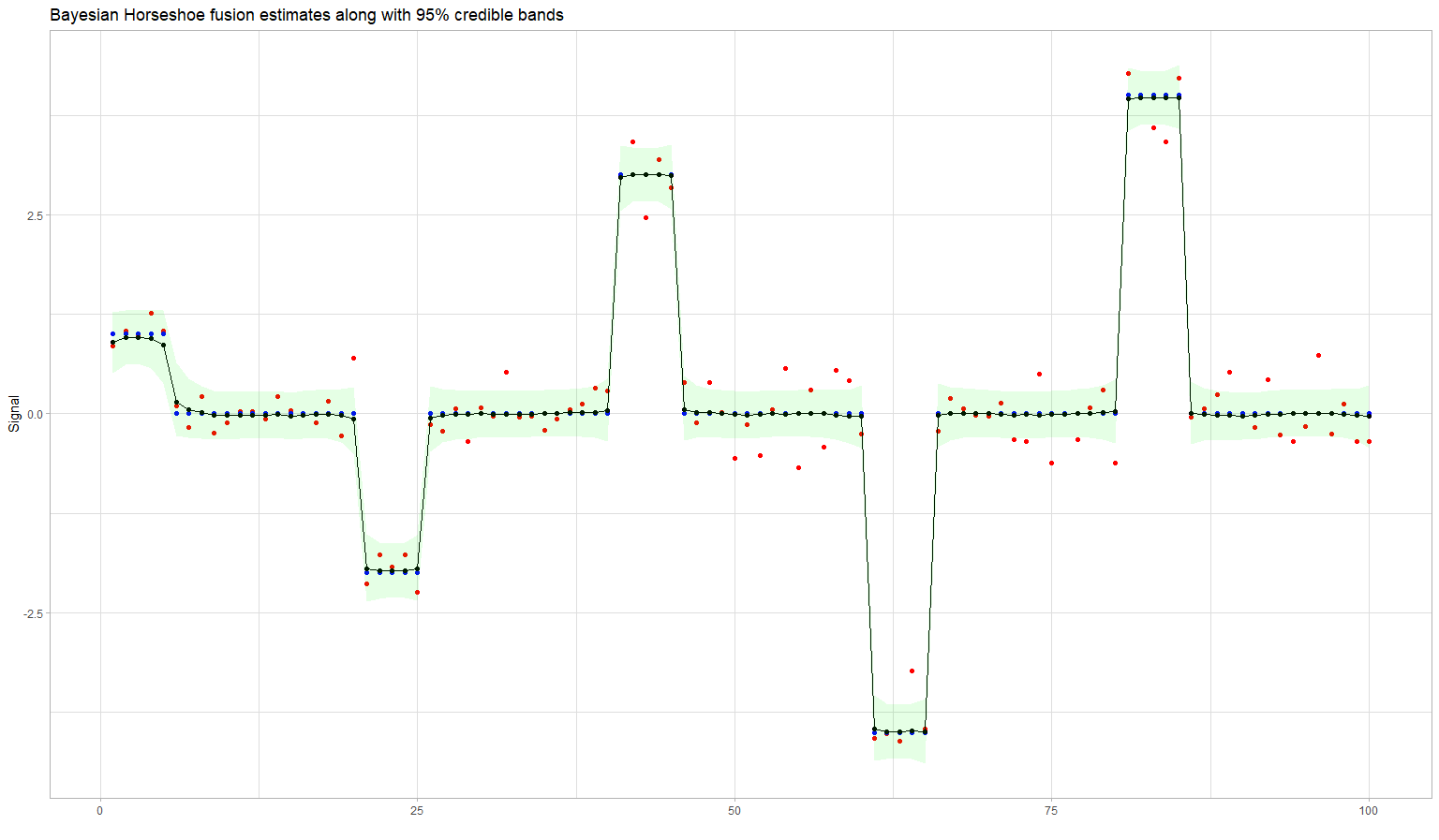} &   \includegraphics[width=40mm]{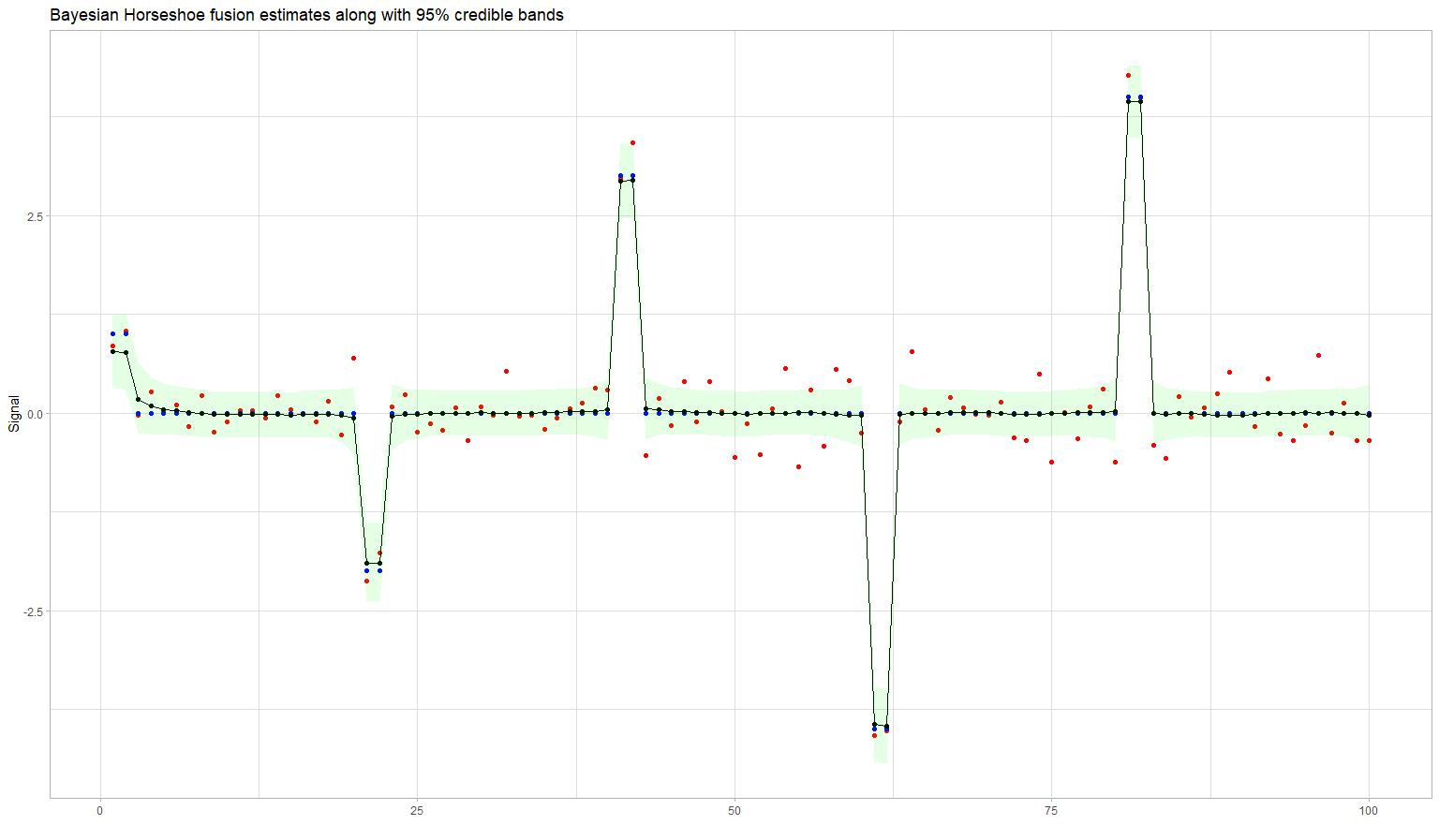} \\
		$t$ &&&\\
		&\includegraphics[width=40mm]{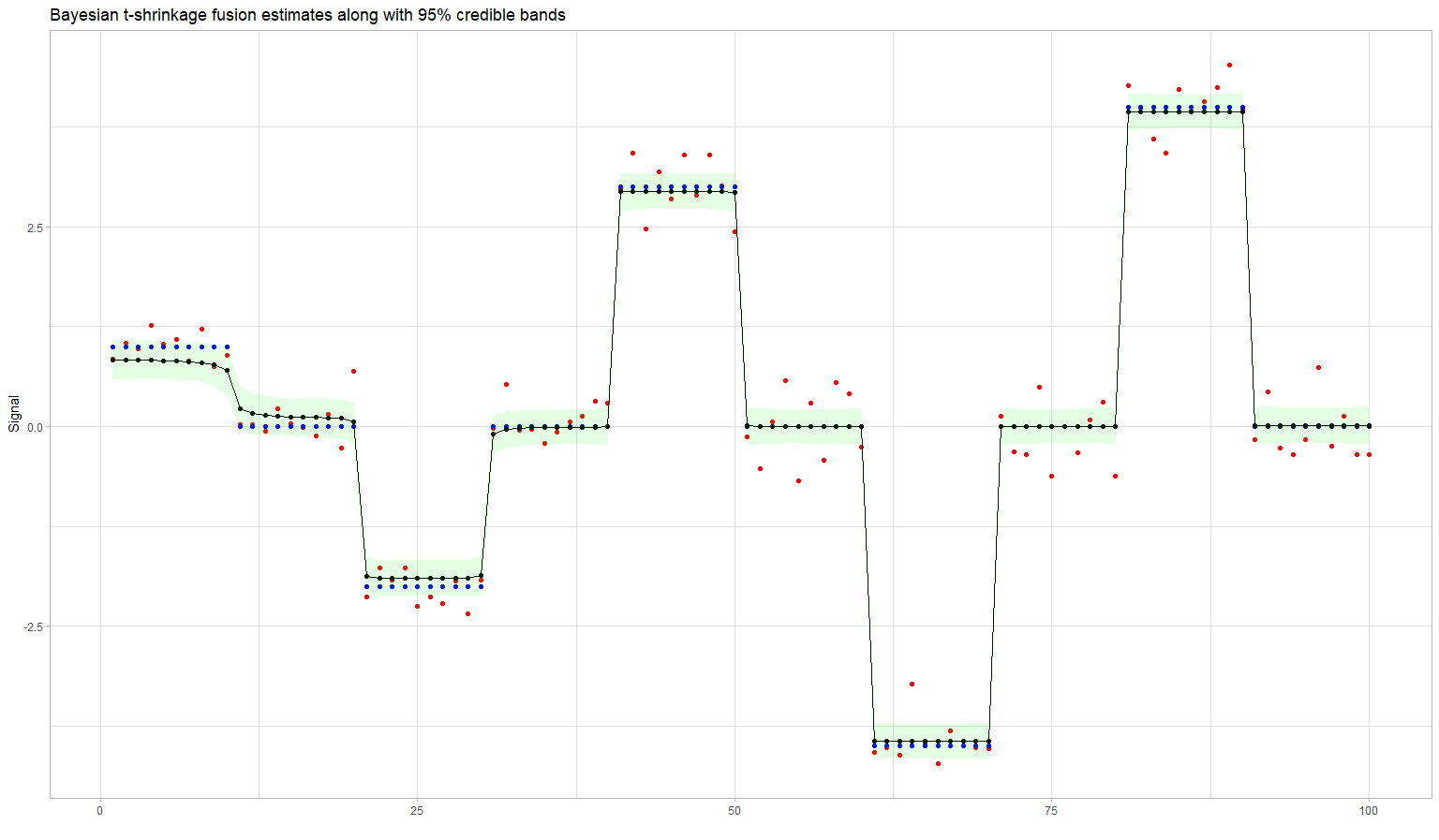} &   \includegraphics[width=40mm]{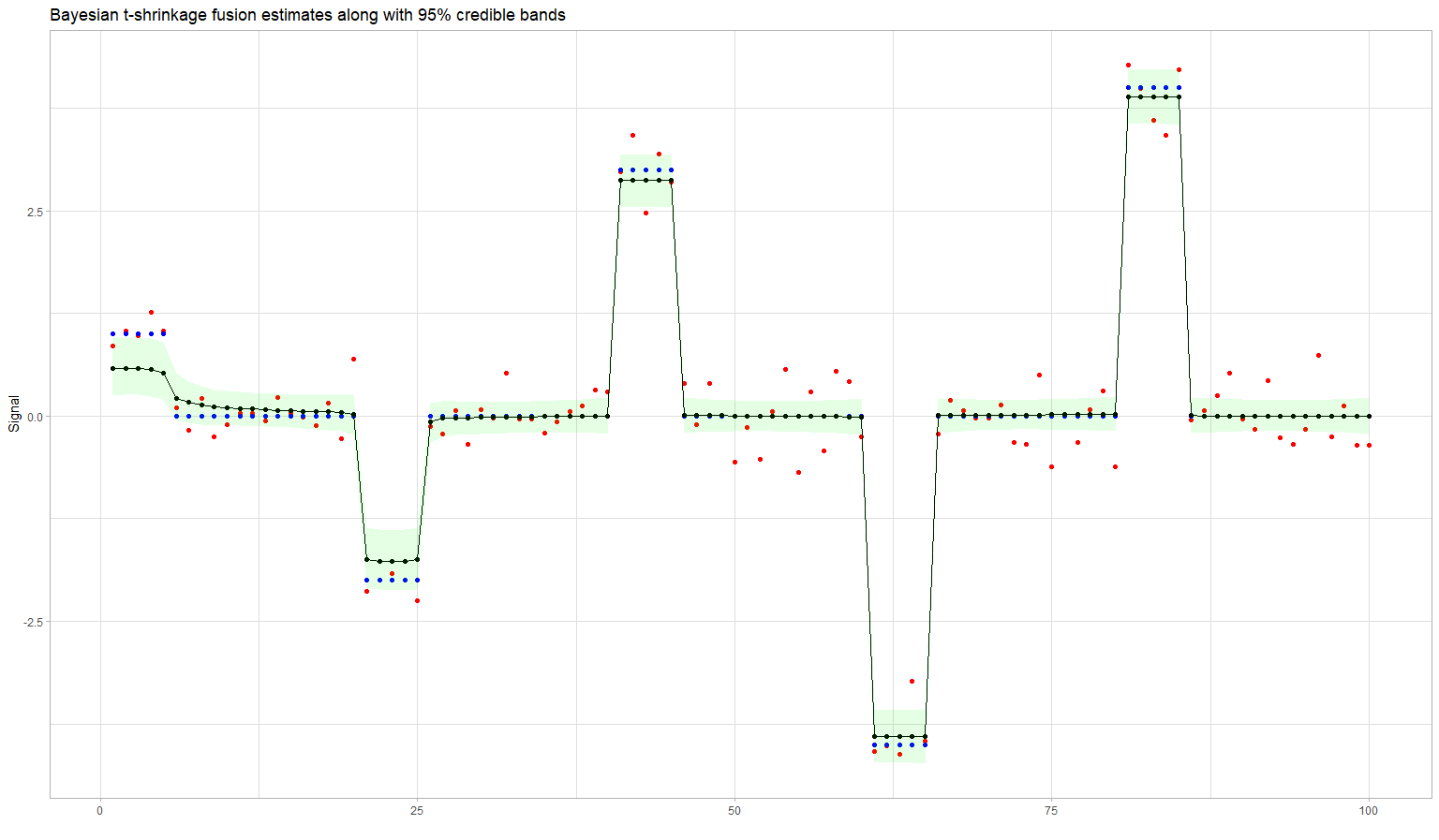} &   \includegraphics[width=40mm]{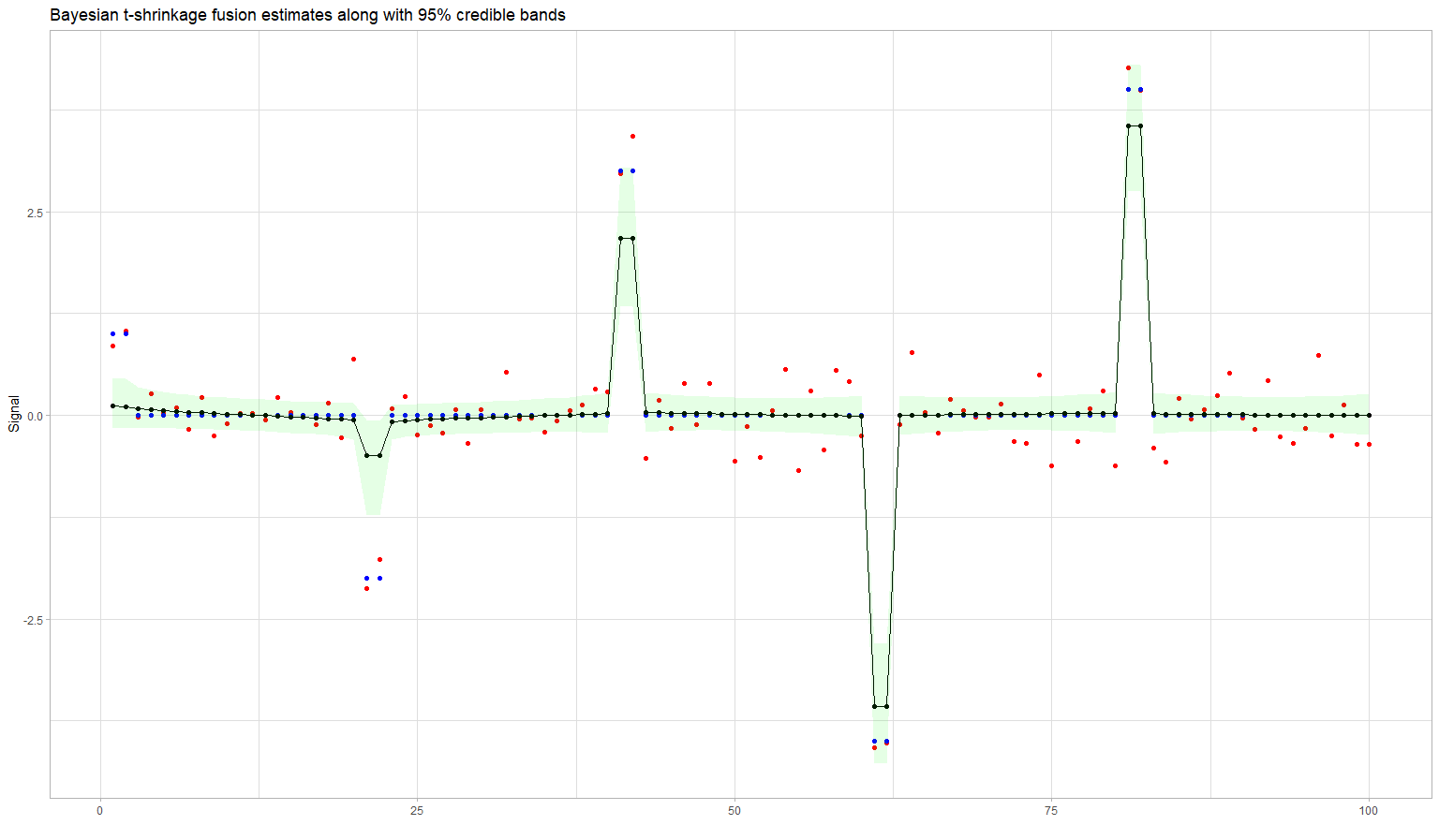} \\
		Laplace &&&\\
		&	\includegraphics[width=40mm]{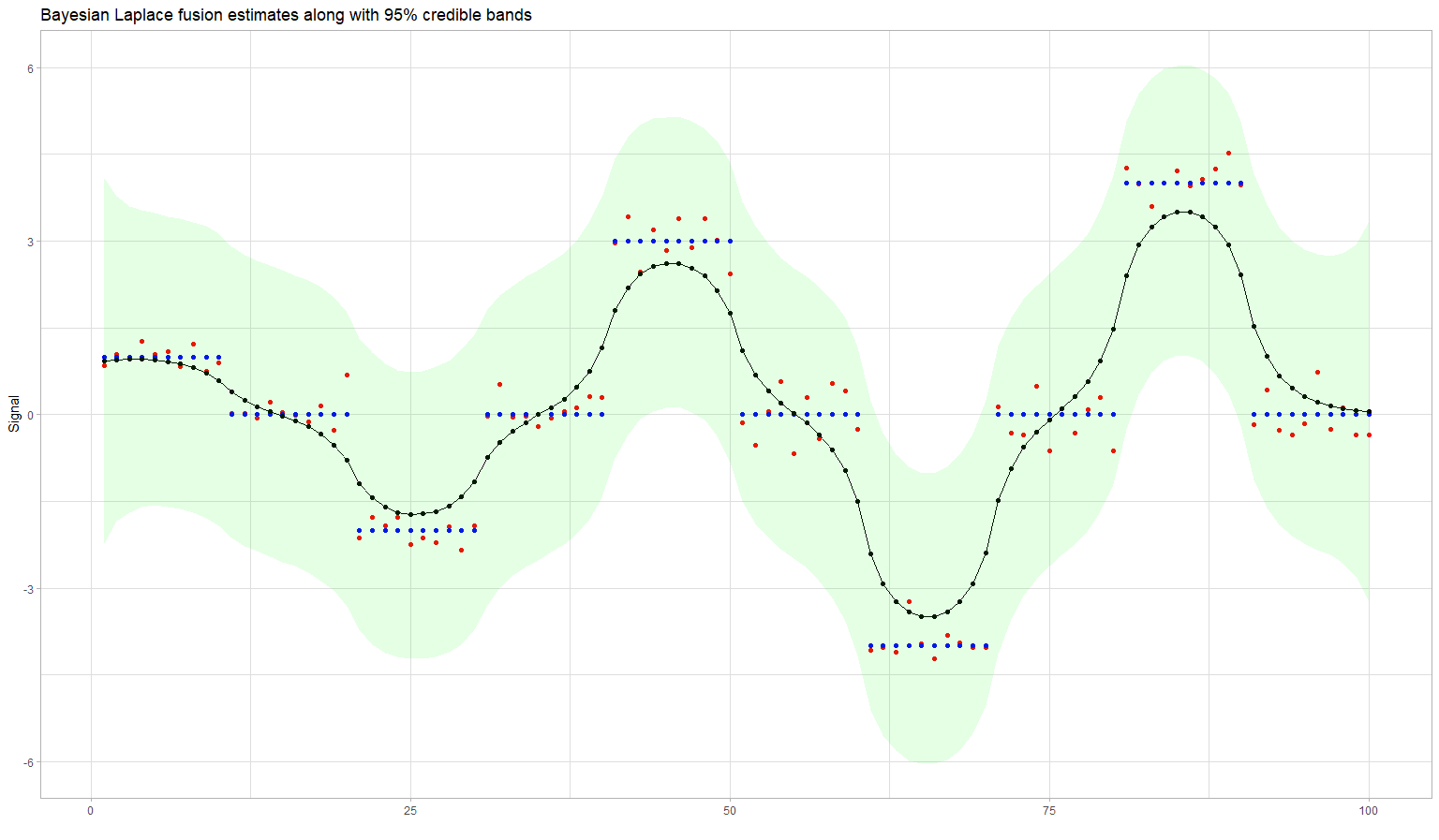} &   \includegraphics[width=40mm]{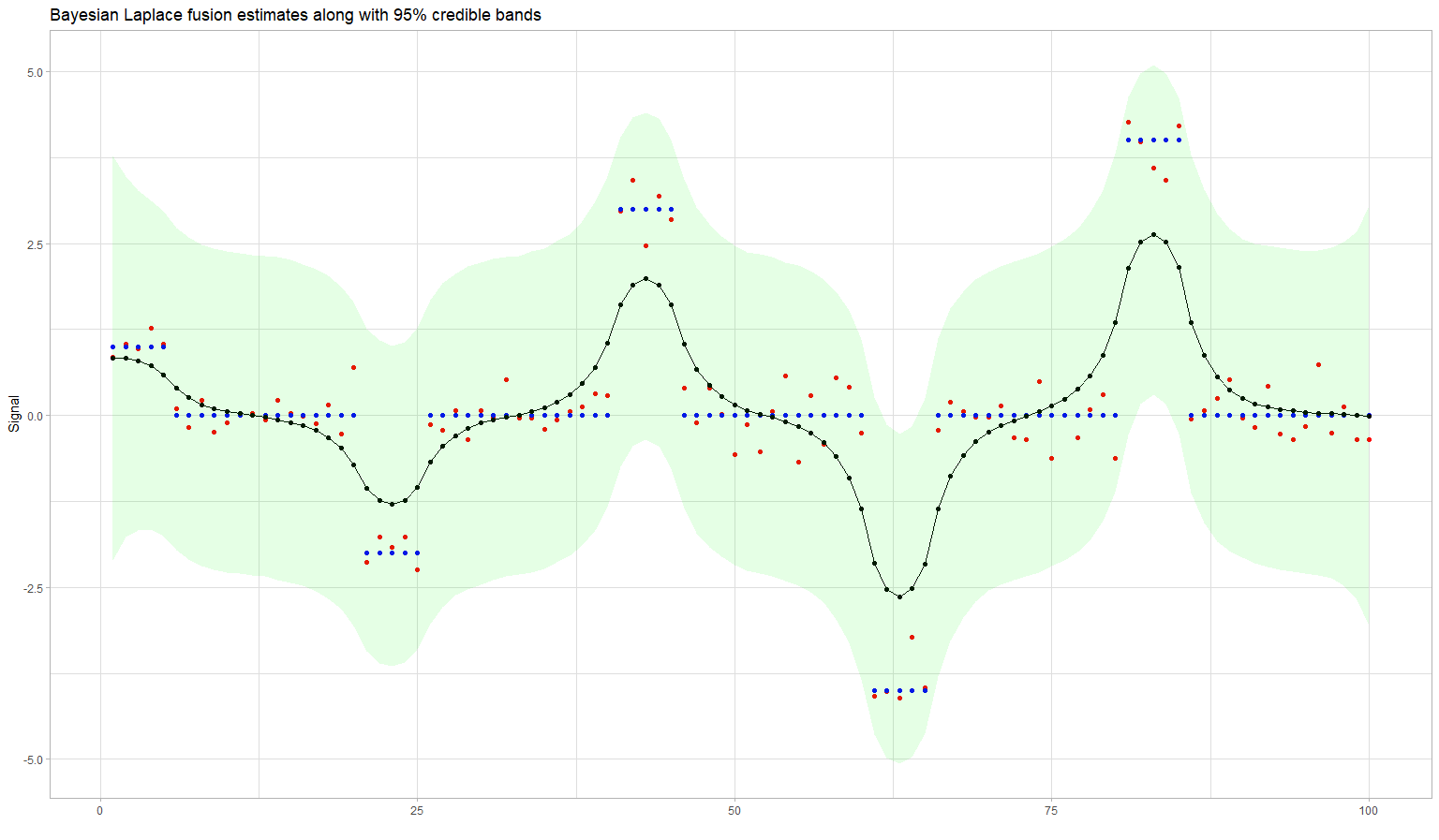} &   \includegraphics[width=40mm]{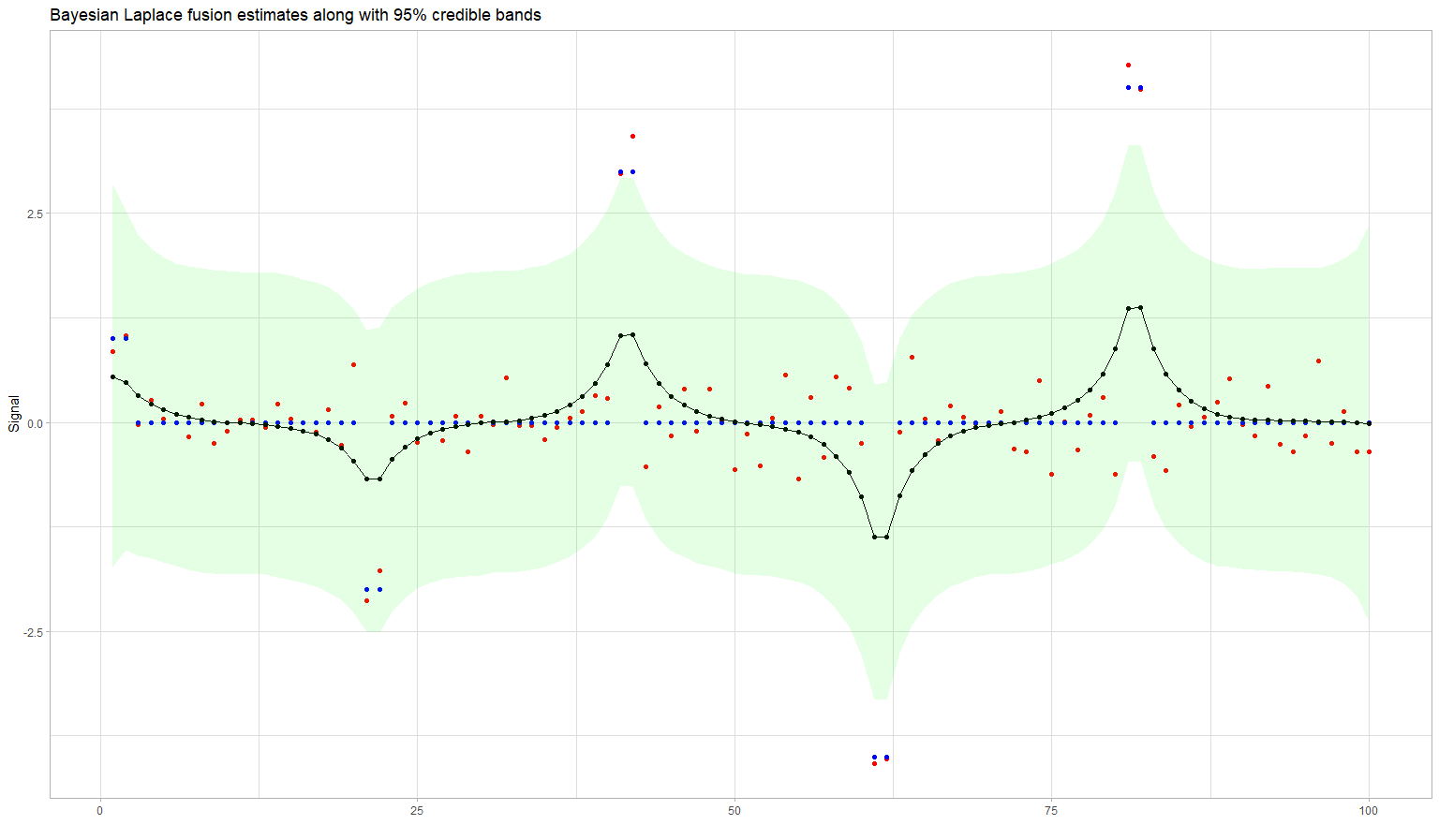} \\
		Fused  &&&\\
		& \includegraphics[width=40mm]{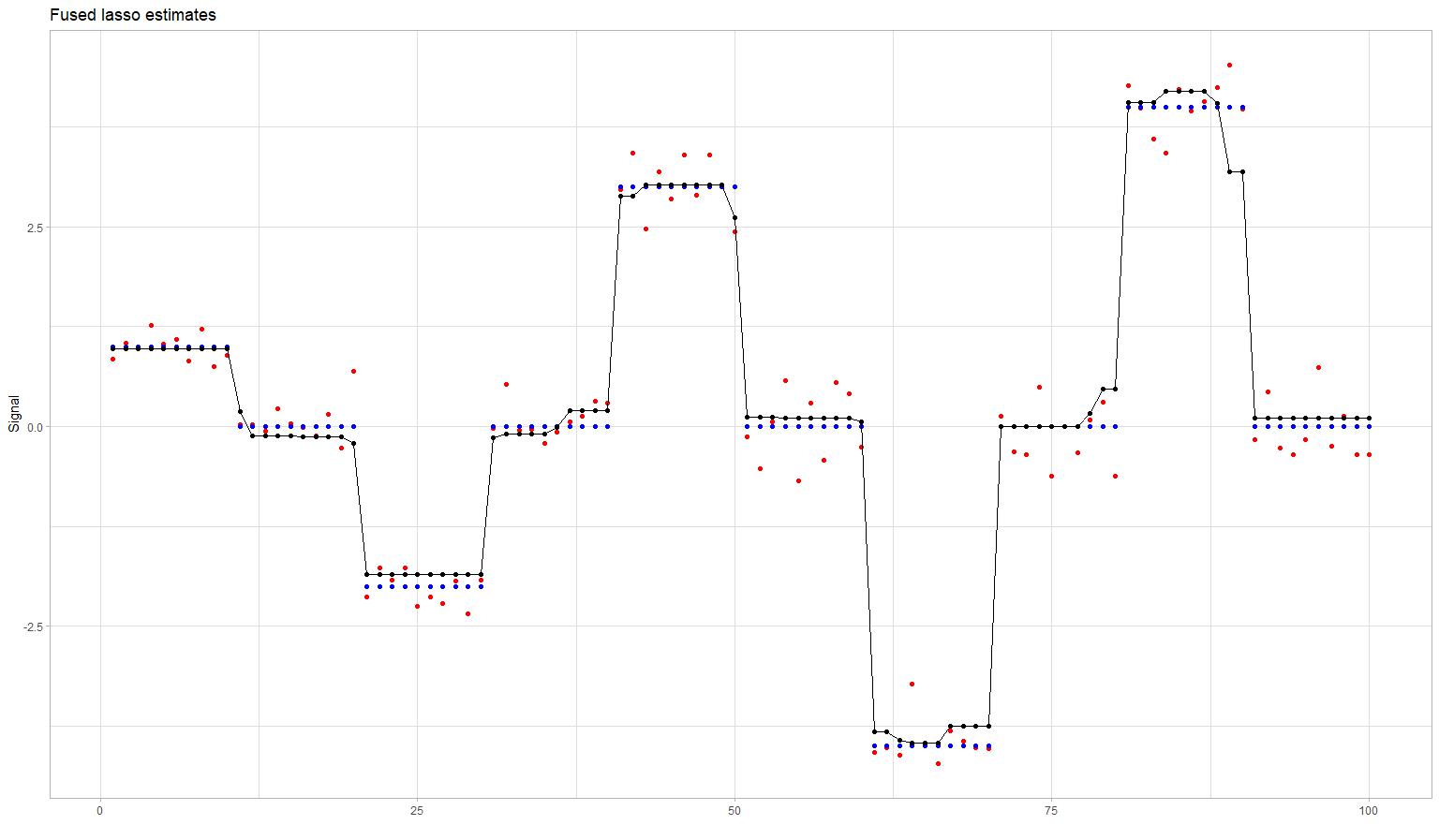} &   \includegraphics[width=40mm]{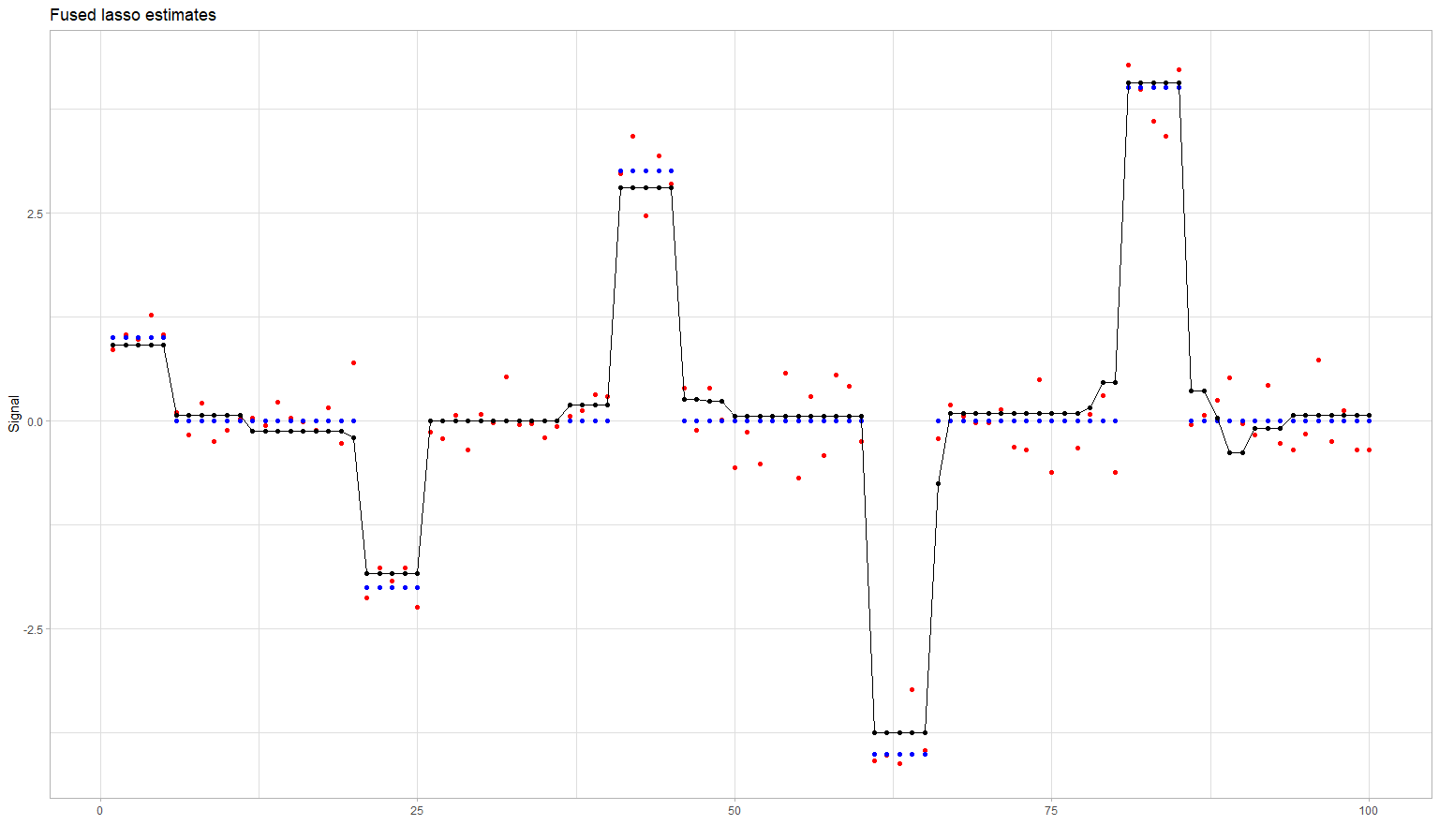} &   \includegraphics[width=40mm]{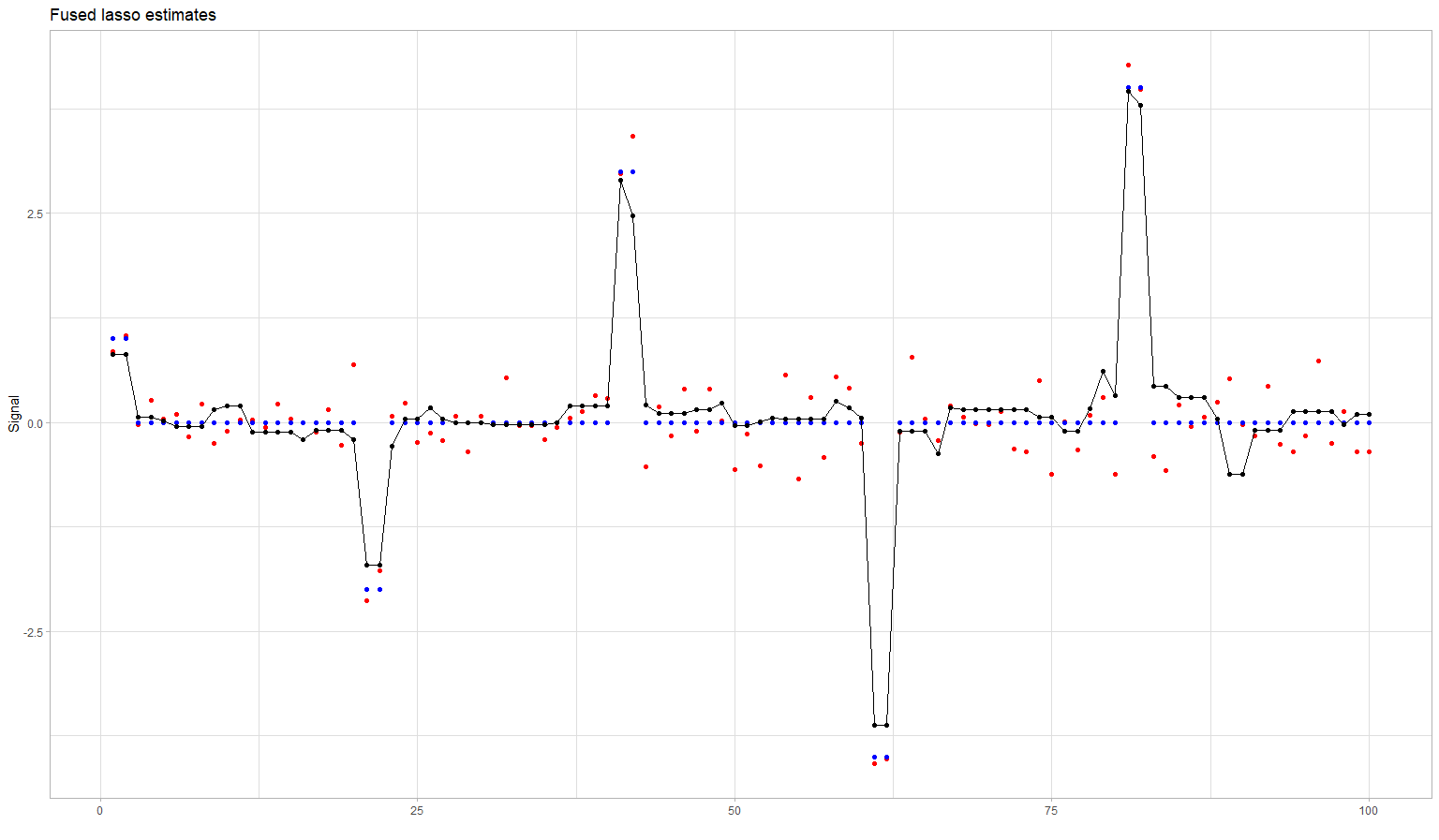} \\	
	\end{tabular}
	\caption{Fusion estimation performance with differently spaced signals and error sd $\sigma = 0.3.$ Observations are represented in red dots, true signals in blue dots, point estimates in black dots, and 95\% credible bands of the Bayesian procedures in green.}
	\label{fig2}
\end{figure}

\begin{figure}
	\begin{tabular}{lccc}
		& Evenly spaced pieces &  Unevenly spaced pieces &  Very unevenly spaced pieces \\
		True&&&\\
		&\includegraphics[width=40mm]{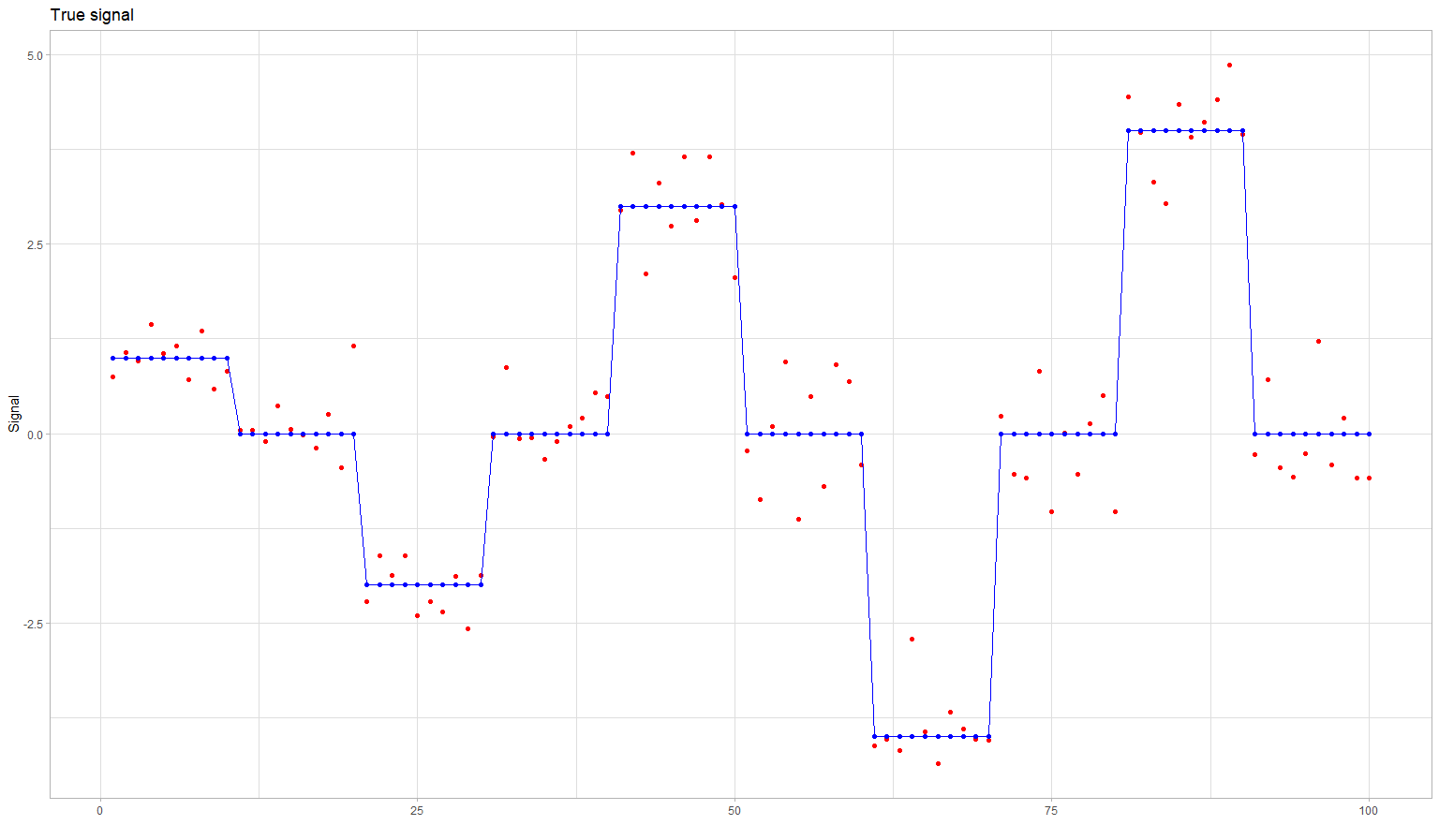} &   \includegraphics[width=40mm]{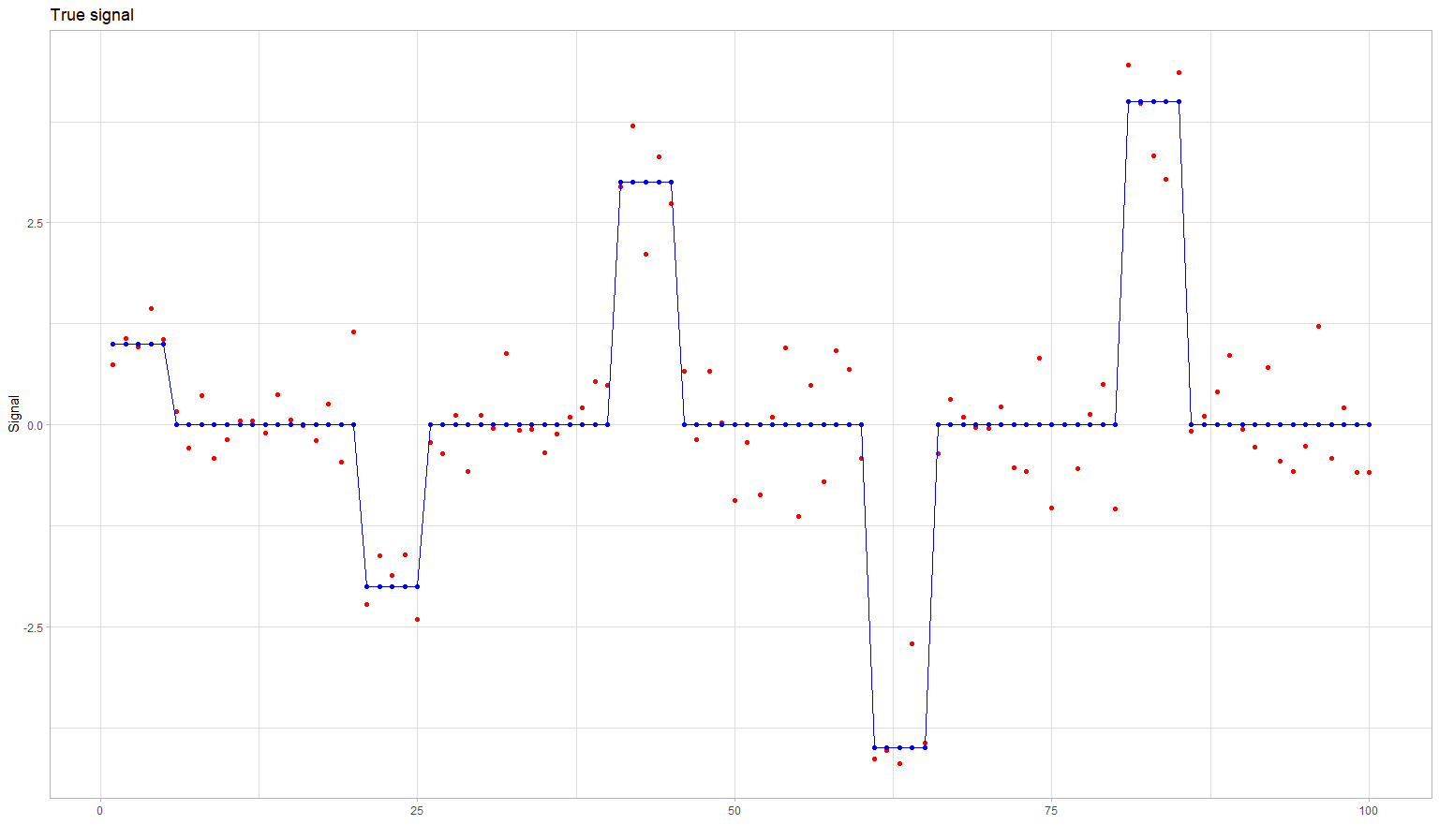} &   \includegraphics[width=40mm]{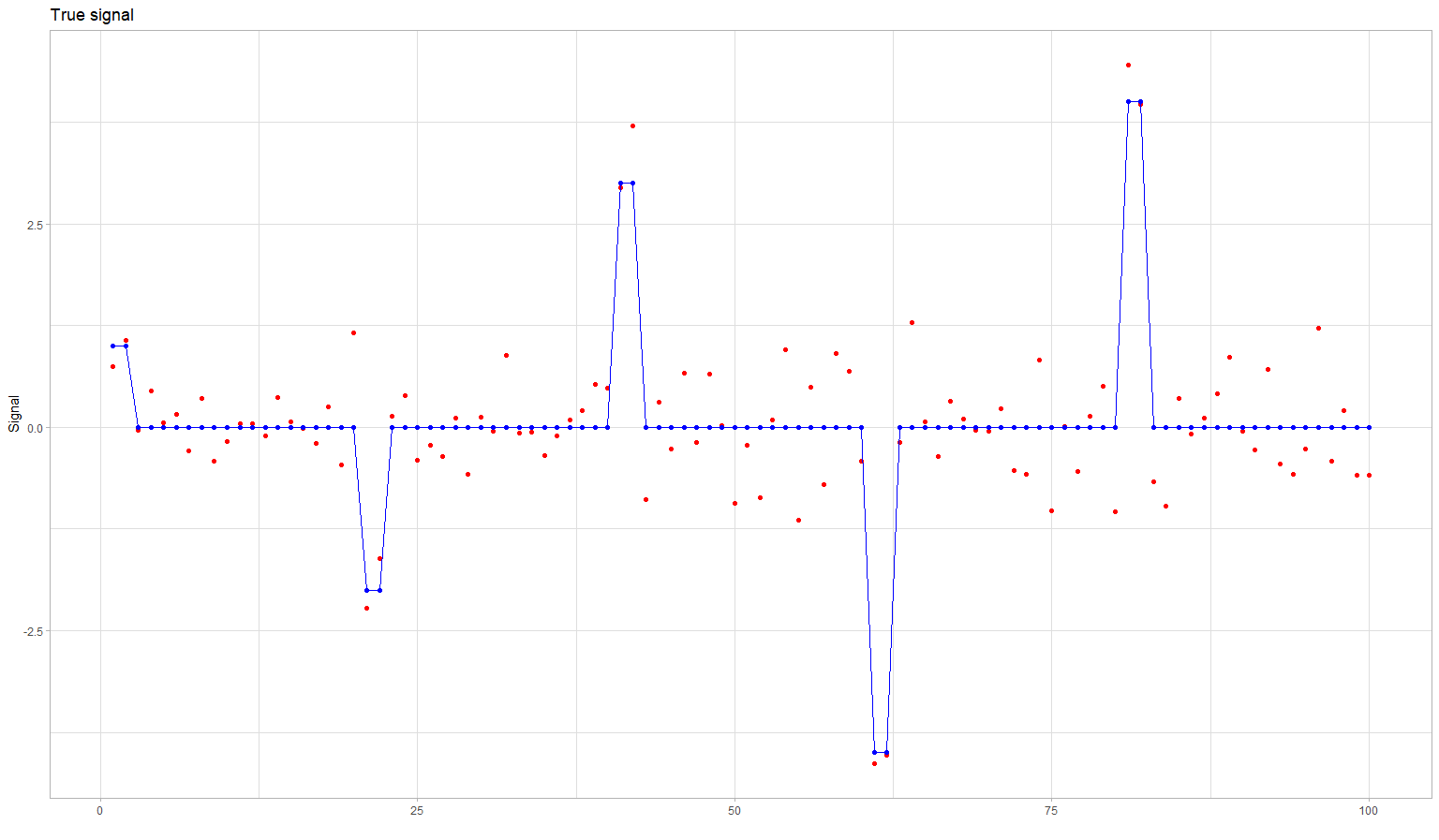} \\
		HS &&&\\
		 & \includegraphics[width=40mm]{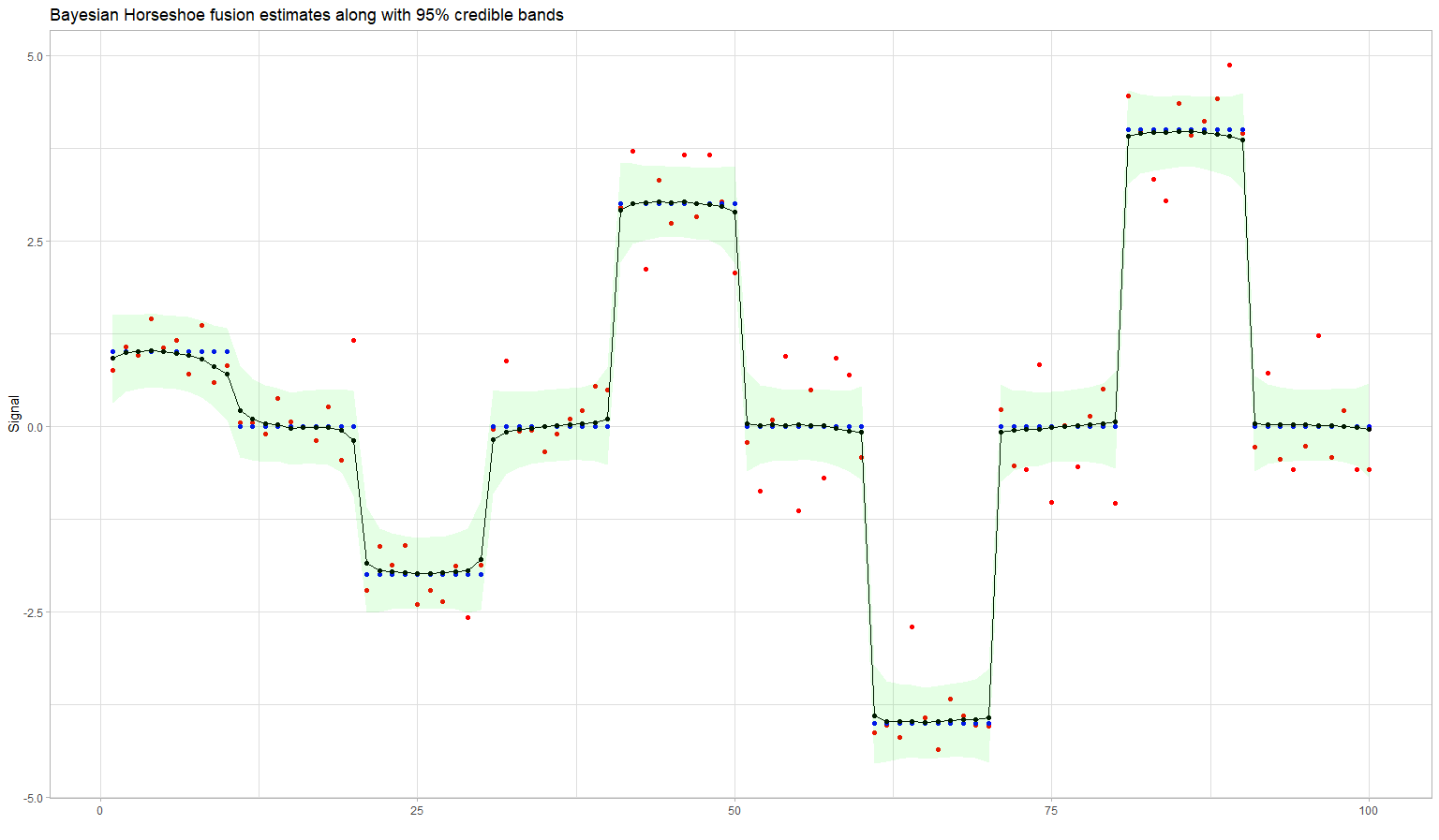} &   \includegraphics[width=40mm]{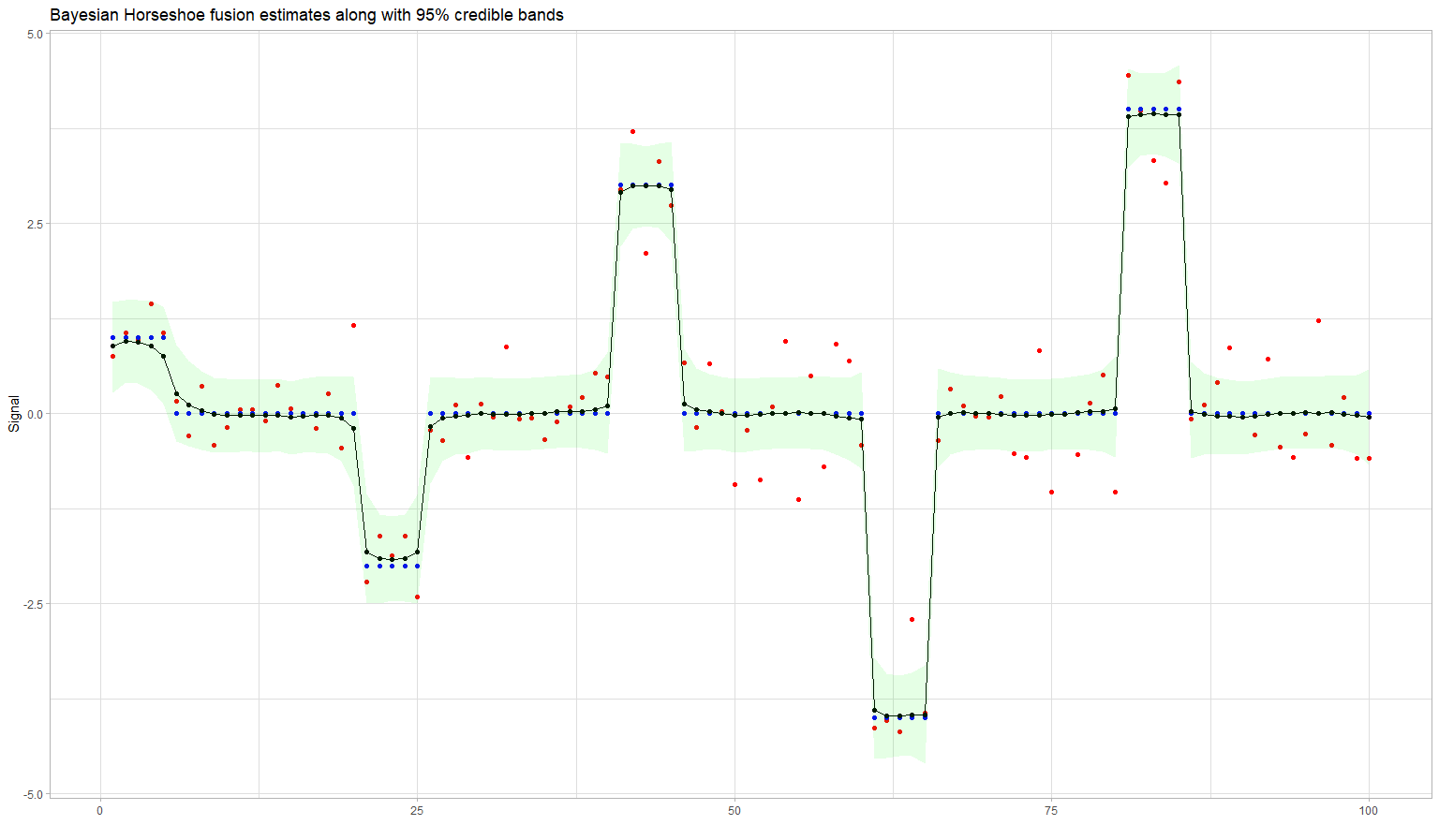} &   \includegraphics[width=40mm]{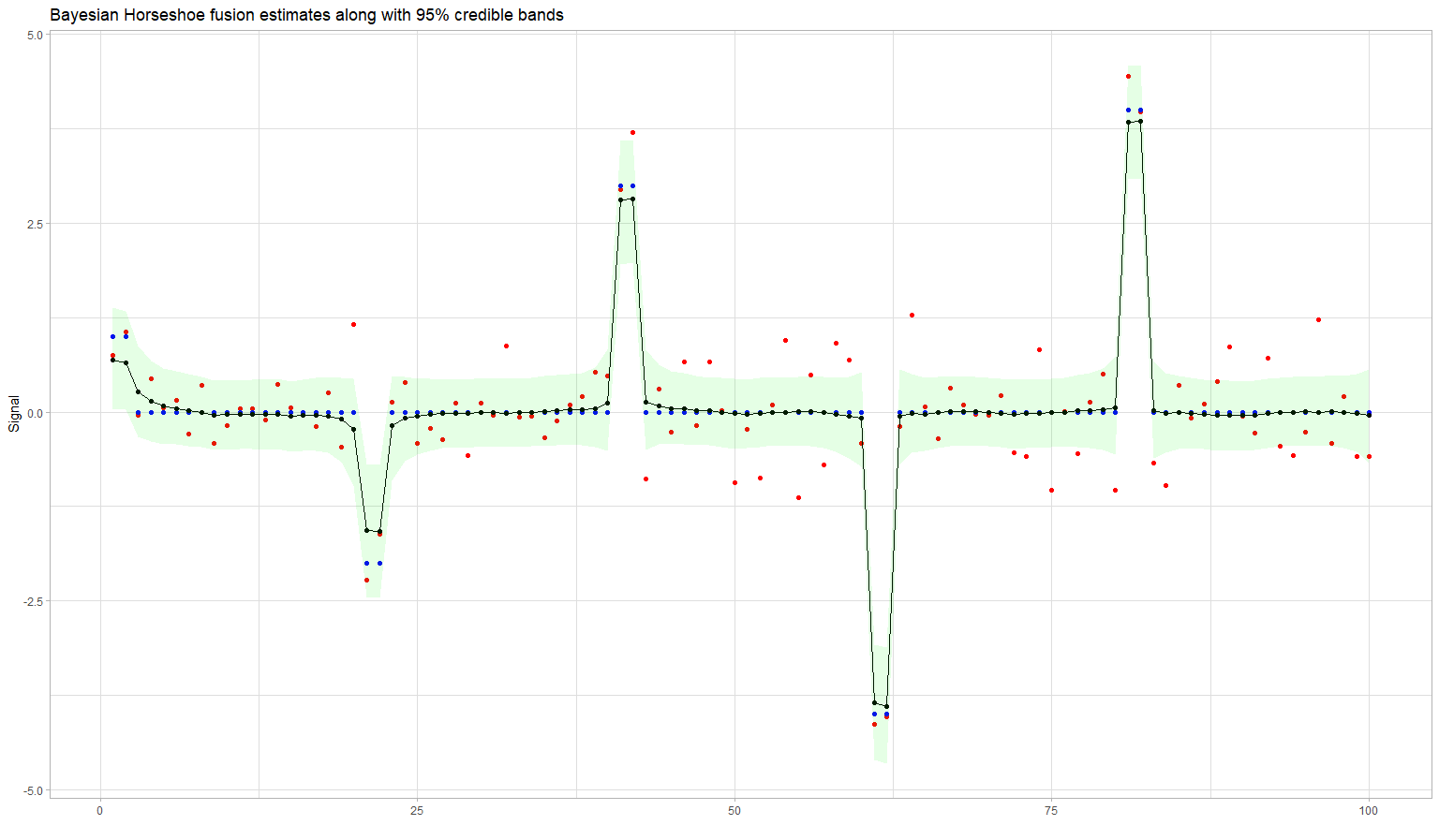} \\
		$t$ &&&\\
		 &\includegraphics[width=40mm]{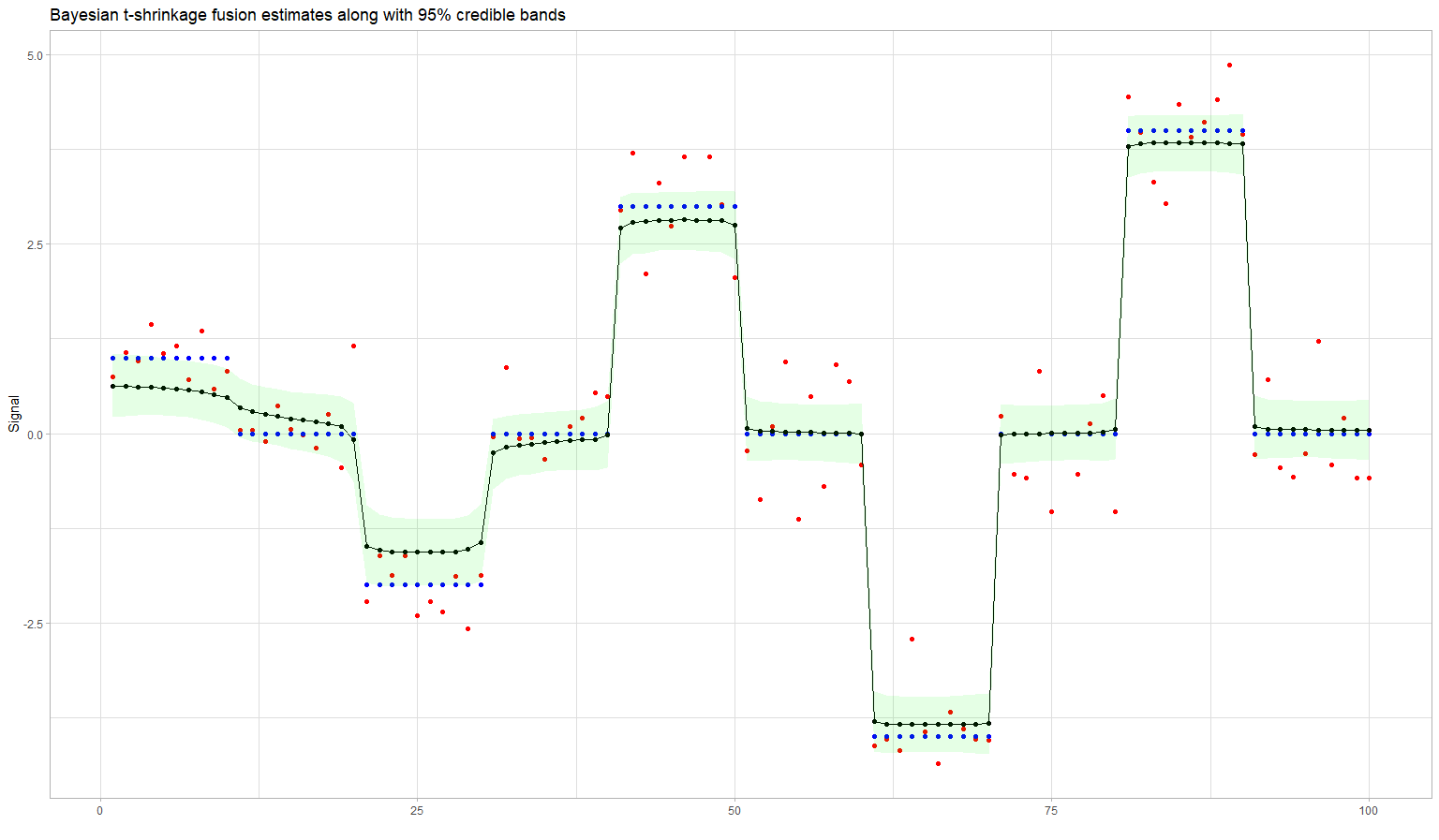} &   \includegraphics[width=40mm]{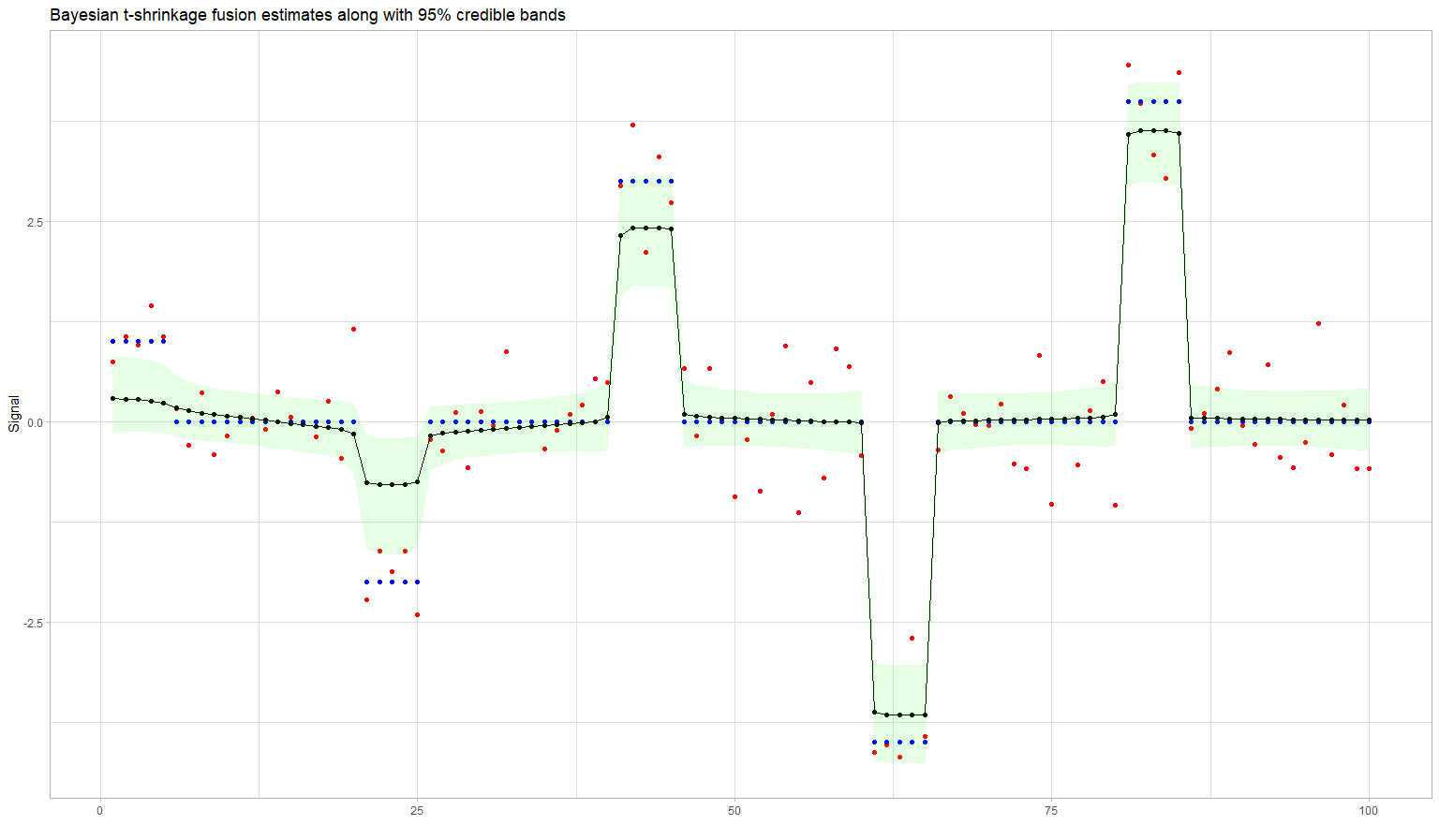} &   \includegraphics[width=40mm]{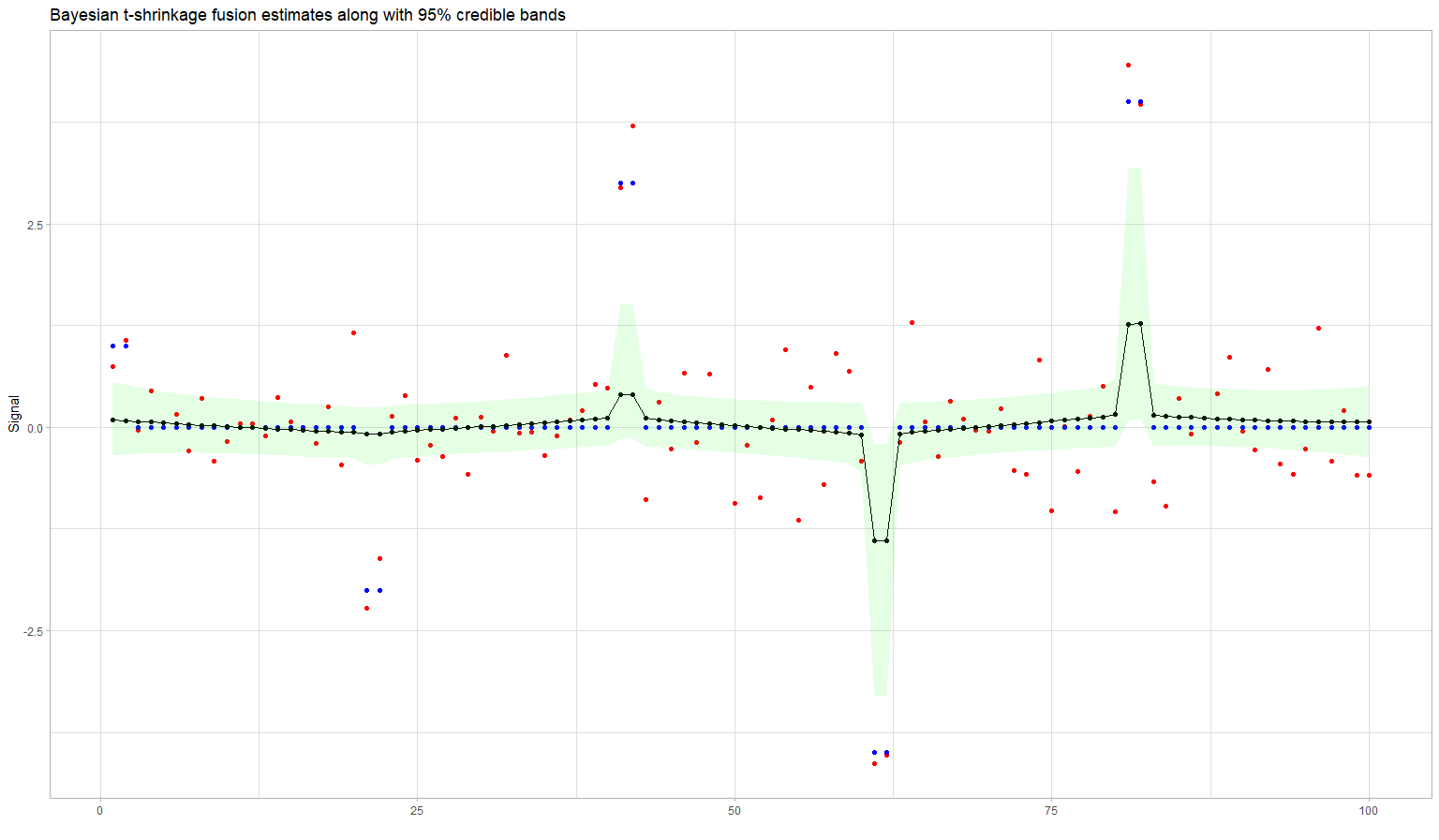} \\
		Laplace &&&\\
		 &	\includegraphics[width=40mm]{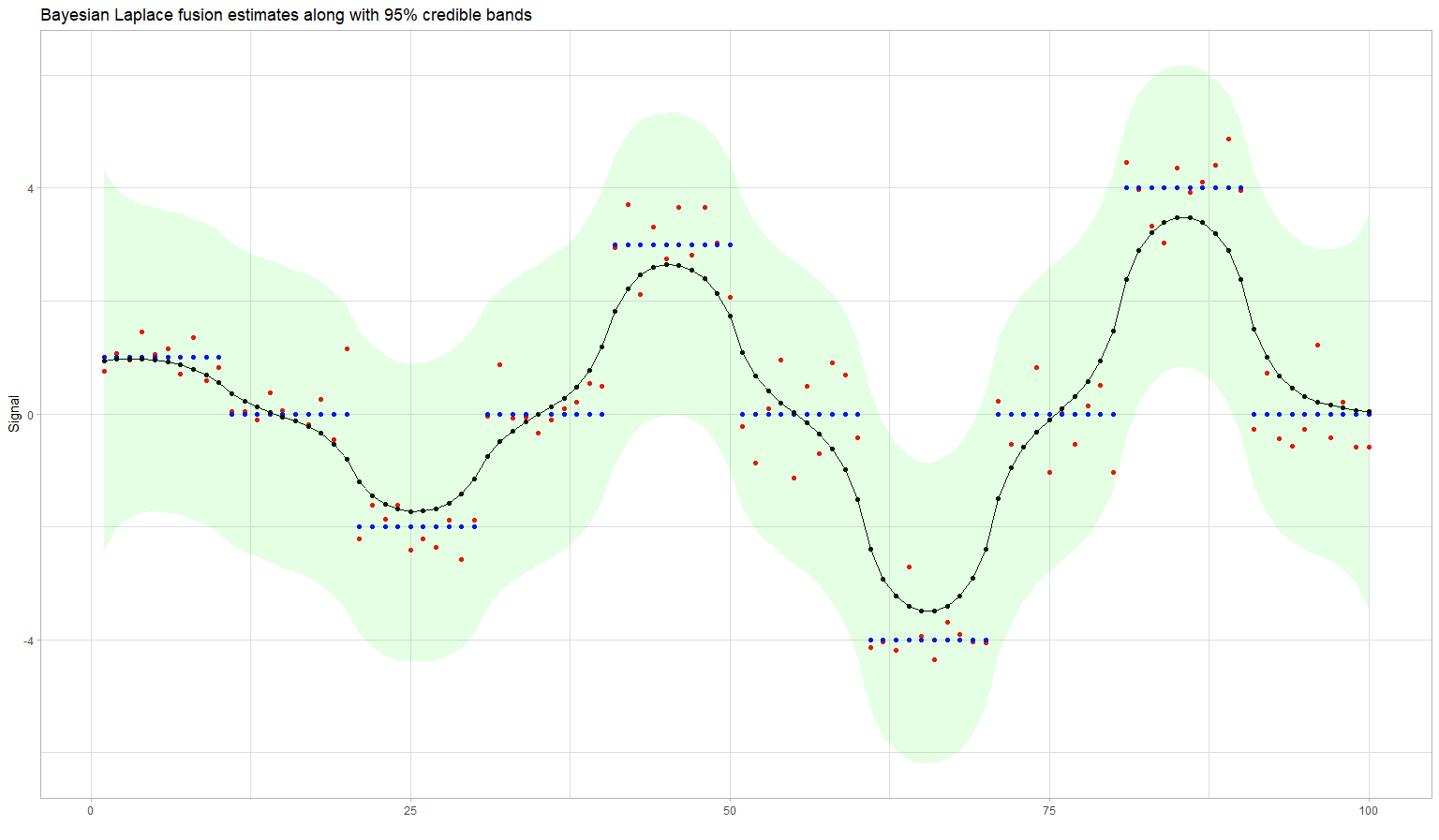} &   \includegraphics[width=40mm]{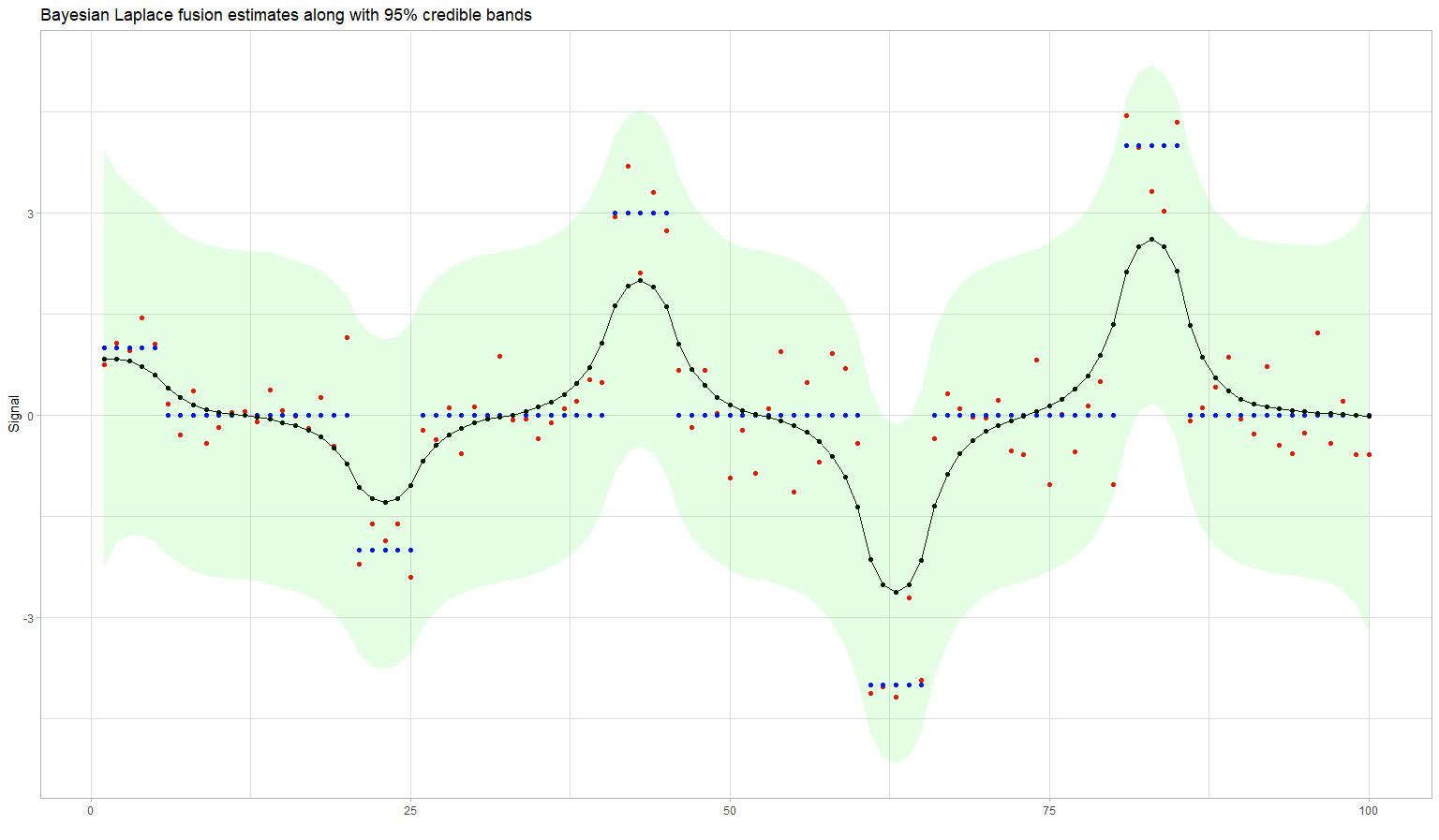} &   \includegraphics[width=40mm]{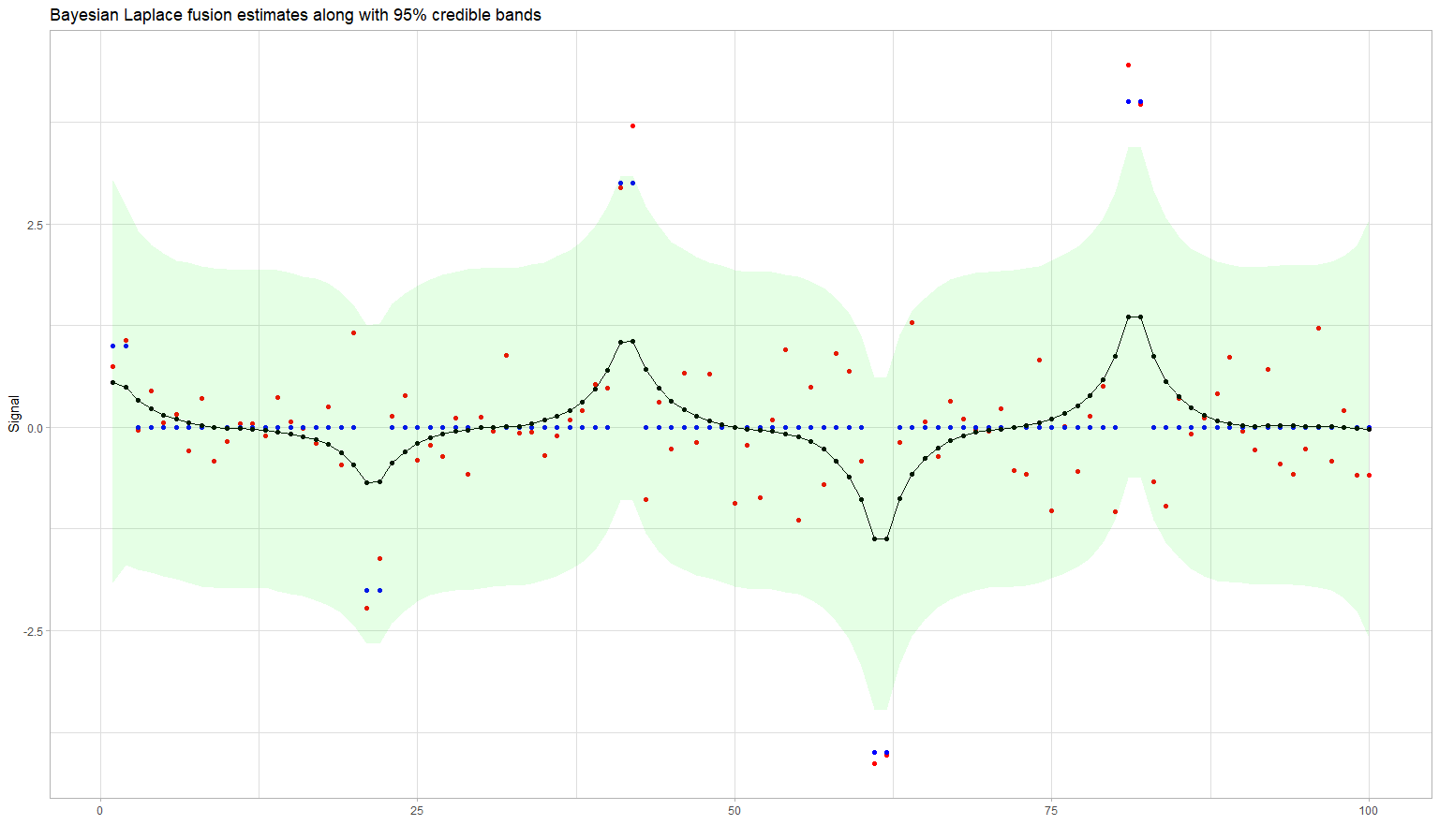} \\
		Fused &&&\\
		 & \includegraphics[width=40mm]{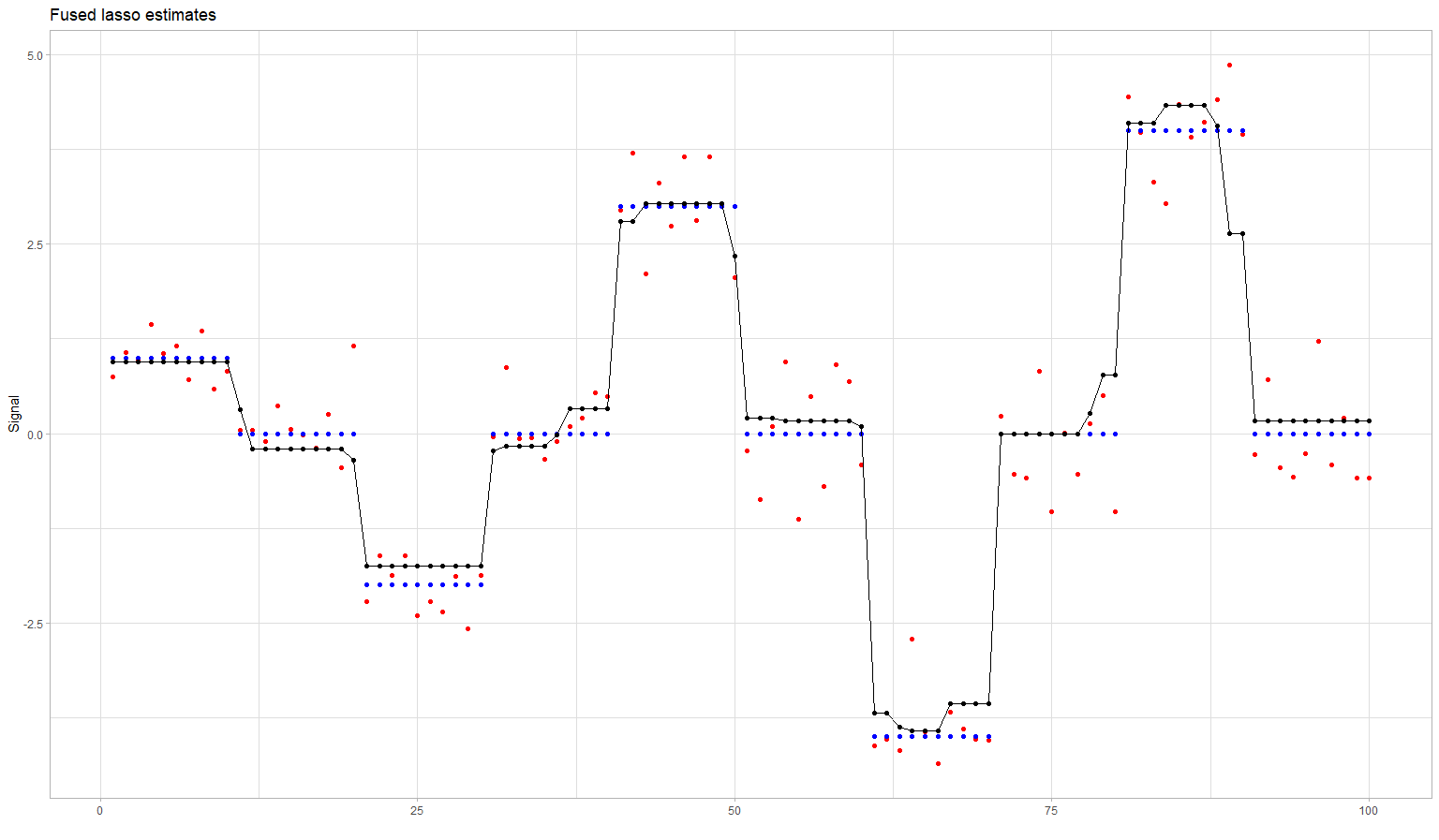} &   \includegraphics[width=40mm]{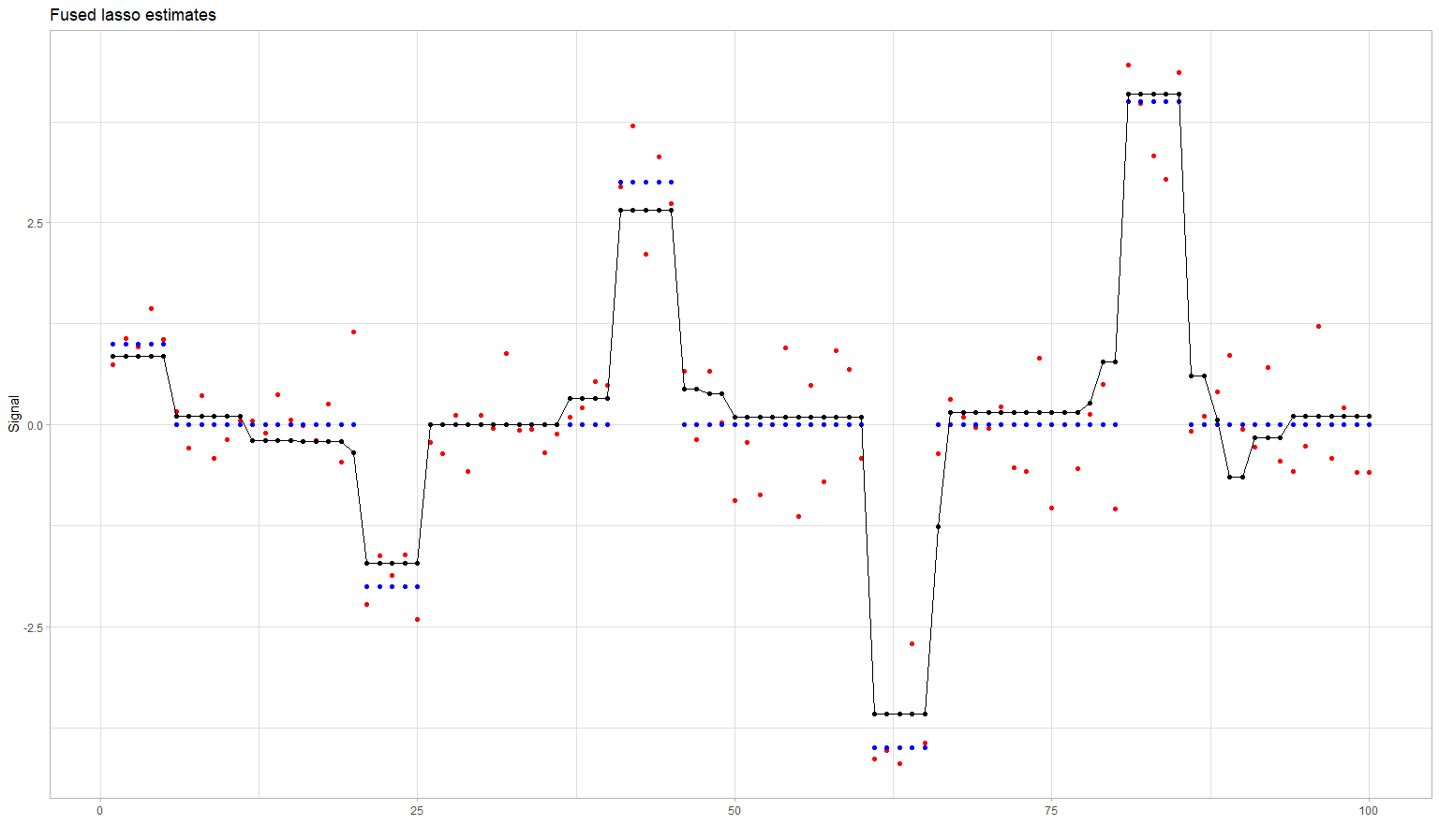} &   \includegraphics[width=40mm]{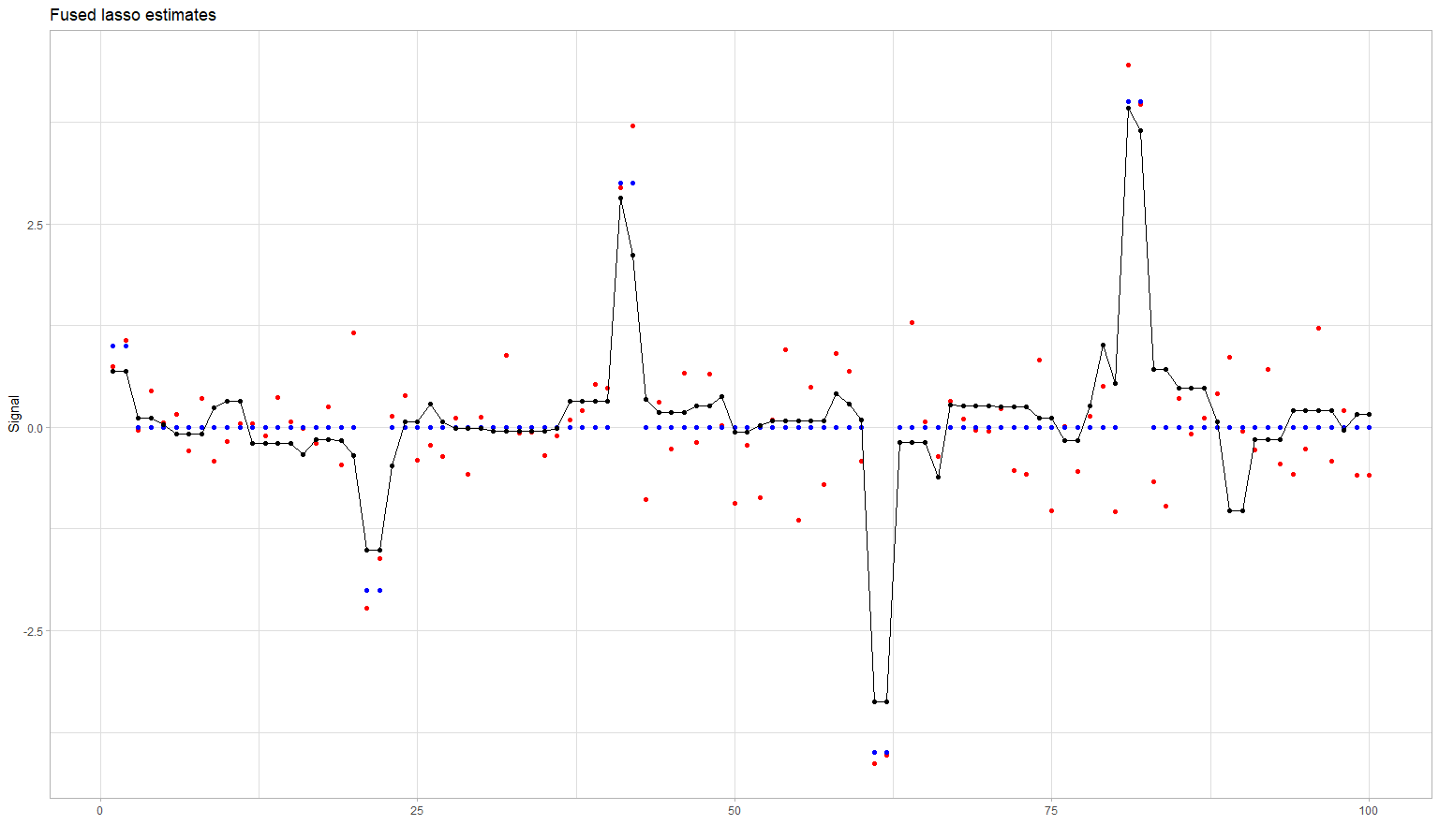} \\	
	\end{tabular}
	\caption{Fusion estimation performance with differently spaced signals and error sd $\sigma = 0.5.$ Observations are represented in red dots, true signals in blue dots, point estimates in black dots, and 95\% credible bands of the Bayesian procedures in green.}
	\label{fig3}
\end{figure}

\section{Real-data examples}
\label{sec:real-data-analysis}

We check the performance of our proposed method on two real-world data examples where the observed data can be modeled using piecewise constant functions. The first example involves DNA copy-number analysis of Glioblastoma multiforme (GBM) data, and the second one involves analysis of Solar X-Ray flux data for the period of October-November 2003 that encompasses the Halloween 2003 Solar storm events. 

\subsection{DNA copy-number analysis of Glioblastoma multiforme (GBM)}

Array Comparative Genomic Hybridization (aCGH) is a high-throughput technique used to identify chromosomal abnormalities in the genomic DNA. Identification and characterization of chromosomal aberrations are extremely important for pathogenesis of various diseases, including cancer. In cancerous cells, genes in a particular chromosome may get amplified or deleted owing to mutations, thus resulting in gains or losses in DNA copies of the genes. Array CGH data measure the $\log_2$-ratio between the DNA copy number of genes in tumor cells and those in control cells. The aCGH data thus may be considered as a piecewise constant signal with certain non-zero blocks corresponding to aberrations owing to additions/deletions at the DNA level.

Glioblastoma multiforme (GBM) is the most common malignant brain tumors in adults \citep{holland2000glioblastoma}. \cite{lai2005comparative} analyzed array CGH data in GBM using glioma data from \cite{bredel2005high} along with control samples. \cite{tibshirani2008spatial} used the data from \cite{lai2005comparative}, and created aCGH data corresponding to a pseudo-chromosome so that the resulting genome sequence have shorter regions of high amplitude gains (copy number additions) as well as a broad region of low amplitude loss (copy number deletions). To be precise, the combined array consists of aCGH data from chromosome 7 in GBM29 and chromosome 13 in GBM31. The aCGH data is presented in Figure~\ref{fig:cgh-data}, and is available from the archived version of the \texttt{cghFlasso} package in \texttt{R}. 

We apply the Horseshoe prior based Bayesian fusion estimation method and also compare the performance with the fused lasso method. We post-process the Bayesian estimates by discretization via the thresholding method as outlined in the appendix. The results are displayed in Figure~\ref{fig: CGH}. We find that our proposed Bayesian fusion estimation method has successfully detected the high amplitude regions of copy number additions, as well as the broad low amplitude region of copy number deletions. The results are very much comparable with the fused lasso estimates, and the findings are similar to that in \cite{lai2005comparative}. The Bayesian credible intervals obtained using our proposed method may further be utilized in addressing future research problems, like controlling false discovery rates.

\begin{figure}
	\centering
	\includegraphics[height = 3in, width=5in]{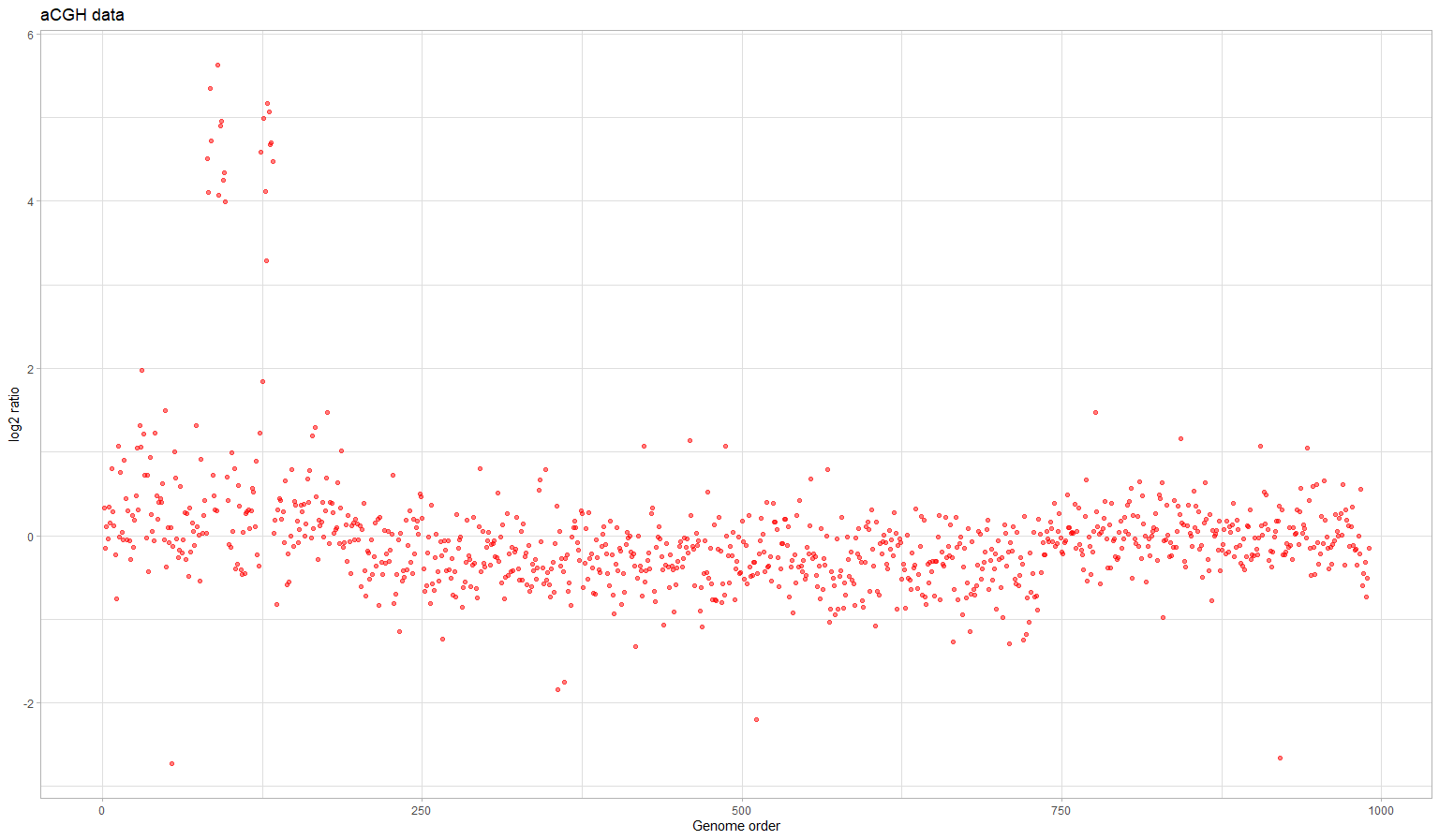}
	\caption{Figure showing array CGH data for Chromosome 7 in GBM29 and chromosome 13 in GBM31.}
	\label{fig:cgh-data}
\end{figure}

\begin{figure}
	\centering
	\begin{tabular}{c}
		\includegraphics[height=3in, width=5in]{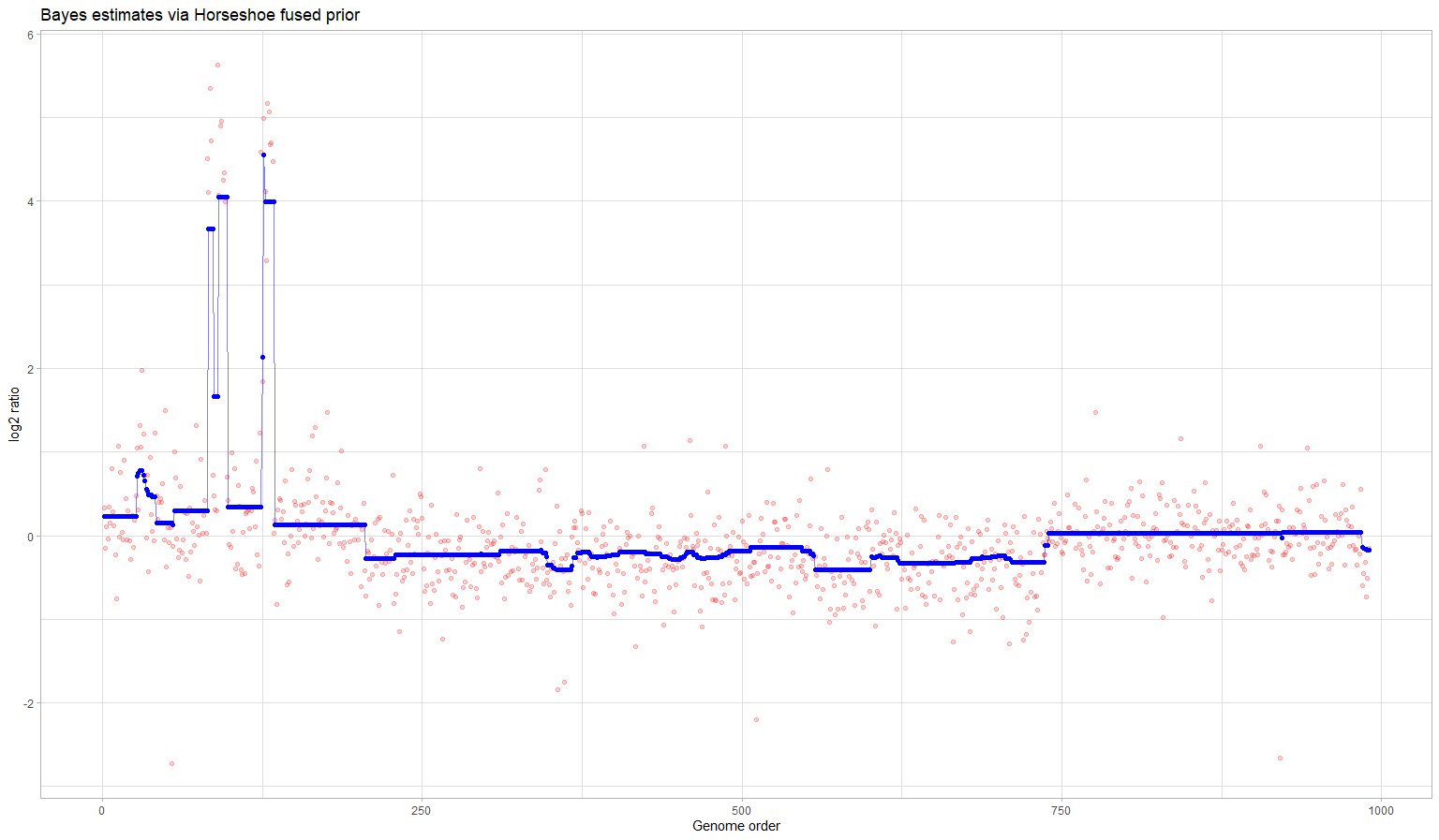}  \\
		(a) Bayesian Horseshoe fusion method  \\[1pt]
		\includegraphics[height=3in, width=5in]{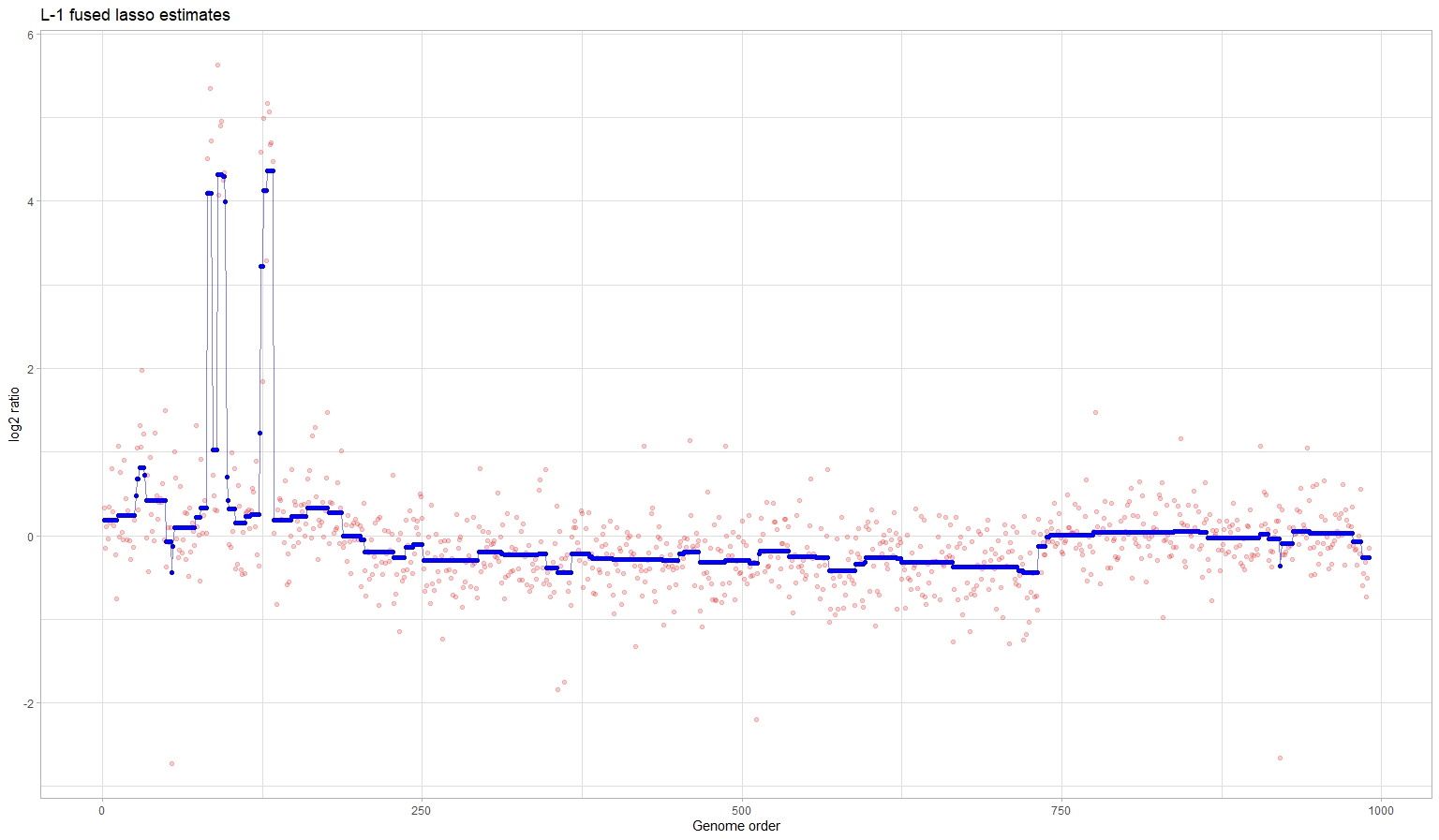} \\
		(b) $L_1$ fused lasso method \\[1pt]
	\end{tabular}
	\caption{Figure showing estimation performances of (a) the proposed Bayesian Horseshoe prior based fusion method, and (b) the $L_1$ fused lasso method. The blue lines represent the estimated signals, and the red dots represent the $\log_2$ ratios between the DNA copy numbers of genes in tumor cells and those in control cells.}
	\label{fig: CGH}
\end{figure}

\subsection{Solar X-Ray flux analysis: Halloween Solar Storms 2003}

The Earth's ionosphere is responsible for blocking high frequency radio waves on the side facing the Sun. Large solar flares affect the ionosphere, thus interrupting satellite communications, and posing threats of radiation hazards to astronauts and spacecrafts. These solar flares are also associated with Coronal Mass Ejections (CMEs) that have the potential to trigger geomagnetic storms, which in turn, have the ability to disrupt satellite and radio communications, cause massive black-outs via electric grid failures, and affect global navigation systems. The Geostationary Operational Environmental Satellite (GOES), operated by the United States' National Oceanic and Atmospheric Administration (NOAA), records data on solar flares. Measured in Watts-per-square-meter ($W/m^2$) unit, solar X-Ray flux can be classified in five categories, namely, A, B, C, M, and X, depending on the approximate peak flux with wavelengths in the range of $1-8$ Angstroms (long wave) and $0.5-4$ Angstroms (short wave). Solar flares in category A have the least peak flux ($< 10^{-7}\,W/m^2$), and those in caregory X are the most powerful ones ($>10^{-4}\,W/m^2$). 

Around Halloween in 2003, a series of strong solar storms affected our planet, leading to disruptions in satellite communication, damages to spacecrafts, and massive power outages in Sweden. The increased levels of solar activity made it possible to observe the Northern lights (Aurora Borealis) even from far south latitudes in Texas, USA. We consider the short wave solar X-ray flux data for the two months of October and November 2003. The data have been accessed from NOAA's website (\url{https://www.ngdc.noaa.gov/stp/satellite/goes/dataaccess.html}). We use the 5-minute average flux data for all the days in the two months, and extract the information for the flux level for shortwave X-rays.

The data contain some missing values, which we impute via a linear interpolation method. Also, we further use 1-hour averages of the X-ray flux data, so as to arrive at 24 data points for each day, thus leading to a total of 1464 data points. Figure \ref{fig:solar-data} shows the plot of the logarithm of the flux levels for the two months. As mentioned in \cite{little2011generalized}, we can approximate this data as a piecewise constant signal, and use fusion-based estimation methods to identify days of extreme solar activities. We apply our Horseshoe prior based Bayesian fusion estimation method, and post-process the estimates by using the thresholding approach as mentioned in the appendix. We also apply the $L_1$-fusion estimation method. The fitted estimates are shown in Figure \ref{fig: solar}. We find that our Horseshoe prior based Bayesian approach has been able to detect the periods of increased solar activity around October 28 and also on November 04, 2003, besides giving piecewise constant estimates. Our findings match with related analyses; see \cite{pulkkinen2005geomagnetic}. On the other hand, the fused lasso method suffers from serious over-fitting problems.

\begin{figure}
	\centering
	\includegraphics[height = 3in, width=5in]{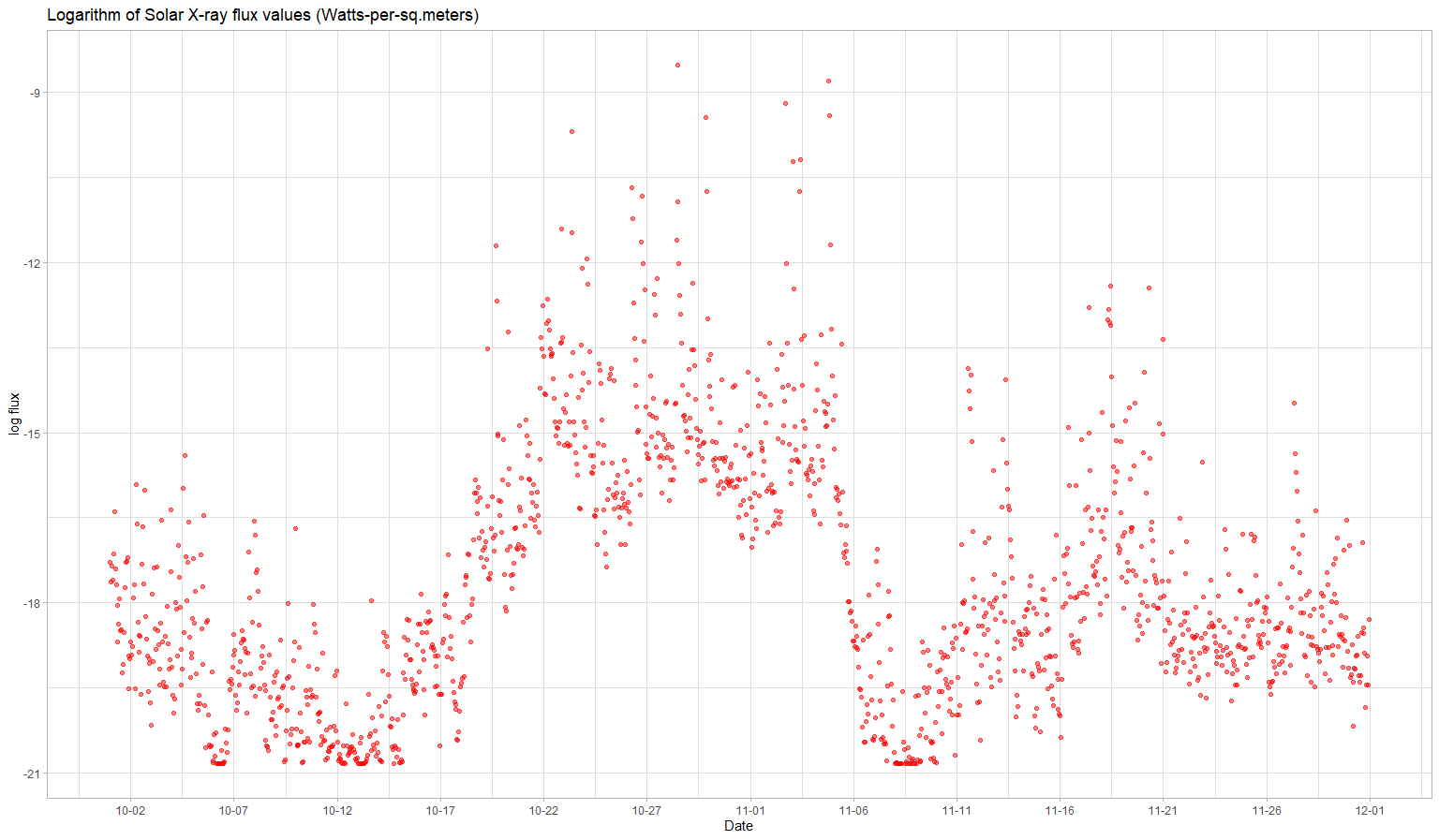}
	\caption{Figure showing scatter plot of logarithm of solar X-ray flare flux values for the months of October and November 2003, recorded by GOES.}
	\label{fig:solar-data}
\end{figure}

\begin{figure}
	\centering
	\begin{tabular}{c}
		\includegraphics[height=3in, width=5in]{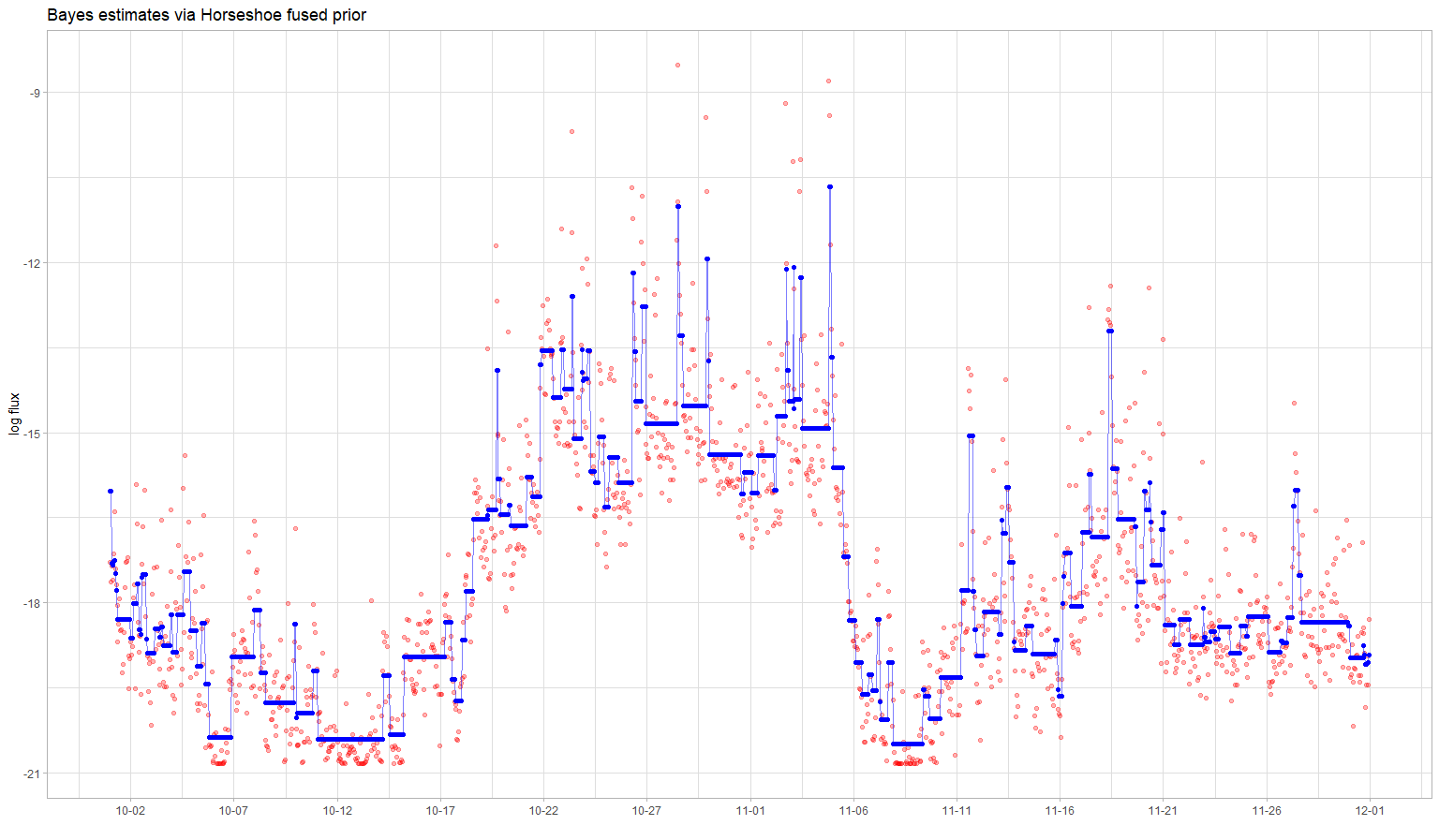}  \\
		(a) Bayesian Horseshoe fusion method  \\[1pt]
		\includegraphics[height=3in, width=5in]{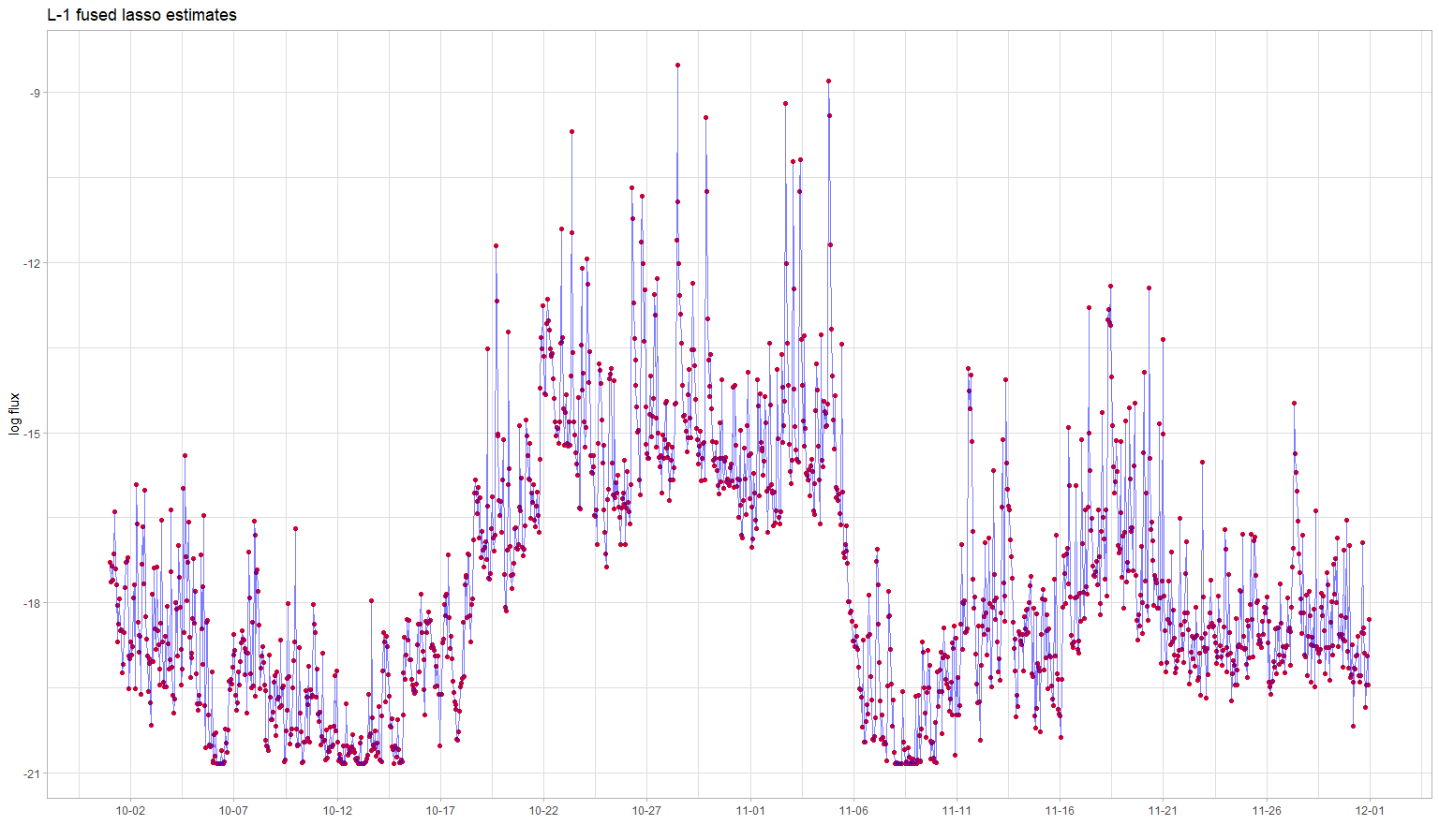} \\
		(b) $L_1$ fused lasso method \\[1pt]
	\end{tabular}
	\caption{Figure showing estimation performances of (a) the proposed Bayesian Horseshoe prior based fusion method, and (b) the $L_1$ fused lasso method. The blue lines represent the estimated signals, and the red dots represent the logarithm of the shortwave X-ray solar flare flux values as recorded by GOES.}
	\label{fig: solar}
\end{figure}

\section{Signal de-noising over arbitrary graphs}
\label{sec:graph de-noising}

In this section, we present an extension of our method for signal de-noising over arbitrary graphs. As mentioned earlier, we consider the normal sequence model as given by (\ref{eqn:gaussian-means}), where the signal $\theta = (\theta_1,\ldots,\theta_n)^T$ is defined over an undirected graph $G = (V,E)$. The true signal is then assumed to be piecewise constant over $G$, and the parameter space is given by $l_0[G,s] := \{\theta \in \mathbb{R}^n: \#\{(i,j) \in E: \theta_i \neq \theta_j \} \leq s, \, 0 \leq s(n) = s \leq n\}$. 

Motivated by the fused lasso approach of \cite{Padilla2017} to deal with the graph de-noising problem, \cite{banerjee2020graph} recently proposed a Bayesian fusion estimation method using $t$-shrinkage priors. In these two works, the actual graph $G$ is first reduced to a linear chain graph via the Depth-First-Search algorithm \citep{tarjan1972depth}. \cite{Padilla2017} provided theoretical justifications of reducing an arbitrary graph $G$ to a DFS-chain graph $G_C = (V,E_C)$, by showing that the total variation of the signal defined over $G_C$ achieves an upper bound of at most twice its total variation wrt the $L_1$-norm over $G$, that is,
$\sum_{(i,j) \in E_C}|\theta_i - \theta_j| \leq 2 \sum_{(i,j) \in E}|\theta_i - \theta_j|.$ Hence, if the true signal is smooth over the graph $G$, so that $\sum_{(i,j) \in E}|\theta_i - \theta_j| \leq \delta$, then, it is reasonable to perform de-noising over $G_C$ as well, in the sense that the total variation over $G_C$ would satisfy $\sum_{(i,j) \in E_C}|\theta_i - \theta_j| \leq 2\delta.$ In addition to this, a similar result holds for the total variation in the $L_0$-norm, so that, $\sum_{(i,j) \in E_C}\Ind\{\theta_i \neq \theta_j\} \leq 2\sum_{(i,j) \in E}\Ind\{\theta_i \neq \theta_j\}.$ A signal defined over a linear chain graph is then identical to the sequence model considered earlier in this paper, upto a permutation of the indices of the signal.

\subsection{Graph de-noising via Horseshoe fusion}

We use our proposed Horseshoe fusion estimation method for graph signal de-noising. Note that reducing a graph to a DFS-chain graph is not unique, and hence would result in different chain graphs. We need to first put a prior on the root vertex, say $\theta_r,\, r \in V$, and then put suitable priors on the successive differences $\theta_i - \theta_j$, for $(i,j) \in E_C$. For inference, we propose to use multiple root nodes to arrive at different DFS-chain graphs, and then combine the posterior samples to arrive at a final estimate. Thus, our proposed Bayesian approach can account for the uncertainty in fusion estimation as well as that in the DFS method. The prior specification is given by,
\begin{eqnarray}
	\theta_r \mid \lambda_0^2, \sigma^2 \sim  N(0, \lambda_0^2\sigma^2),\;\theta_i - \theta_j \mid \lambda_{k(i,j)}^2,\tau^2, \sigma^2 \stackrel{ind}{\sim}  N(0, \lambda_{k(i,j)}^2\tau^2\sigma^2), \; (i,j) \in E_C,\nonumber \\
	\lambda_i^2 \mid \nu_i \stackrel{ind}{\sim}  IG(1/2, 1/\nu_i),\; 1 \leq i \leq n-1,\;\tau^2 \mid \xi \sim  IG(1/2,1/\xi), \nonumber \\
	\nu_1,\ldots,\nu_{n-1}, \xi \stackrel{iid}{\sim}  IG(1/2,1),\; \sigma^2 \sim  IG(a_\sigma, b_\sigma).
	\label{eqn:prior-graph}
\end{eqnarray}
Here $k(i,j)$ refers to the index corresponding to the edge in $E_C$ that connects vertices $i$ and $j$. The hyperparameters in the above prior specification are identical to that in (\ref{eqn:prior-2}). The conditional posterior distributions of the parameters and the hyperparameters are given as follows. 
The normal means have the conditional posterior distribution
\begin{equation}
	\theta_i \mid \cdots \sim N(\mu_i, \zeta_i),\; 1 \leq i \leq n,
\end{equation}

where $\mu_i$ and $\zeta_i$ are given by,
$$\zeta_r^{-1} = \dfrac{1}{\sigma^2}\left(1 + \sum_{j=1}^{n}\dfrac{\Ind\{(r,j) \in E_C\}}{\lambda_{k(r,j)}^2\tau^2} + \dfrac{1}{\lambda_0^2} \right),\; \mu_r = \dfrac{\zeta_r}{\sigma^2}\left(y_r + \sum_{j=1}^{n}\dfrac{\theta_{j}\Ind\{(r,j) \in E_C\}}{\lambda_{k(r,j)}^2\tau^2}\right) ,$$
$$\zeta_i^{-1} = \dfrac{1}{\sigma^2}\left(1 + \sum_{j=1}^{n}\dfrac{\Ind\{(i,j) \in E_C\}}{\lambda_{k(i,j)}^2\tau^2} + \sum_{j=1}^{n}\dfrac{\Ind\{(j,i) \in E_C\}}{\lambda_{k(j,i)}^2\tau^2} \right), \; i \in V\backslash\{r\},$$
$$\mu_i = \dfrac{\zeta_i}{\sigma^2}\left(y_i + \sum_{j=1}^{n}\dfrac{\theta_{j}\Ind\{(i,j) \in E_C\}}{\lambda_{k(i,j)}^2\tau^2} + \sum_{j=1}^{n}\dfrac{\theta_{j}\Ind\{(j,i) \in E_C\}}{\lambda_{k(j,i)}^2\tau^2}\right),\; i \in V\backslash\{r\}.$$ 
The conditional posteriors for the rest of the parameters are given by,

\begin{eqnarray}
	\lambda_{k(i,j)}^2 \mid  \cdots &\sim & IG\left(1, \dfrac{1}{\nu_{k(i,j)}} + \dfrac{(\theta_i - \theta_{j})^2}{2\tau^2\sigma^2}  \right),\; (i,j) \in E_C, \nonumber \\
	\sigma^2 \mid \cdots &\sim & IG\left(n + a_\sigma, b_\sigma + \dfrac{1}{2}\left[\sum_{i=1}^{n}(y_i - \theta_i)^2 + \dfrac{1}{\tau^2}\sum_{(i,j)\in E_C}\dfrac{(\theta_i - \theta_{j})^2}{\lambda_{k(i,j)}^2} + \dfrac{\theta_r^2}{\lambda_0^2}\right]\right), \nonumber \\
	\tau^2 \mid \cdots &\sim & IG\left(\dfrac{n}{2}, \dfrac{1}{\xi} + \dfrac{1}{2\sigma^2}\sum_{(i,j) \in E_C}^{n}\dfrac{(\theta_i - \theta_{j})^2}{\lambda_{k(i,j)}^2}\right), \nonumber \\
	\nu_i \mid \cdots &\sim & IG\left(1, 1 + \dfrac{1}{\lambda_i^2}\right),\; 1 \leq i \leq n-1, \nonumber \\
	\xi \mid \cdots &\sim & IG\left(1, 1 + \dfrac{1}{\tau^2}\right). 
\end{eqnarray}

\subsection{Theoretical guarantees}

We now present the results on posterior concentration of the estimates obtained using the proposed Horseshoe based method. The assumptions on the true model and prior parameters (and, hyperparameters) are almost identical to that in Section~\ref{sec:posterior}, except for some modifications owing to the presence of the graph $G$. We present the assumptions in detail for the sake of completeness.
\begin{assumption}
	The number $s_0$ of true subgraphs of $G$ for which the components of the signal $\theta_0$ are piecewise constant in the model, satisfies $s_0 \prec n/\log n$.
	\label{assump-graph:true-blocks}
\end{assumption}

\begin{assumption}
	The true mean parameter vector $\theta_0 = (\theta_{0,1},\ldots, \theta_{0,n})^T$ and the true error variance $\sigma_0^2$ satisfy the following conditions:
	\begin{enumerate}
		\item[(i)] $\max_{(i,j) \in E} |\theta_{0,i} - \theta_{0,j}|/\sigma_0 < L,$ where $\log L = O(\log n).$
		\item[(ii)] For root node $r \in V$, $\theta_{0,r}/(\lambda_0^2\sigma_0^2) + 2\log \lambda_0 = O(\log n),$ where $\lambda_0$ is the prior hyperparameter appearing in the prior for $\theta_r$ in (\ref{eqn:prior-graph}).
	\end{enumerate}
	\label{assump-graph:true-param}
\end{assumption}

\begin{assumption}
	The global scale parameter $\tau$ in the prior specification (\ref{eqn:prior-graph}) satisfies $\tau < n^{-(2 + b)}$ for some constant $b > 0$, and $-\log \tau = O(\log n)$.
	\label{assump-graph:prior}
\end{assumption}

Under the above assumptions, we present the result on posterior concentration of the proposed estimator. 

\begin{theorem}
		Consider the Gaussian means model (\ref{eqn:gaussian-means}) over a graph $G = (V,E)$ with prior specification as in (\ref{eqn:prior-graph}), and suppose that assumptions \ref{assump-graph:true-blocks}, \ref{assump-graph:true-param} and \ref{assump-graph:prior} hold. Then the posterior distribution of $\theta$, given by $\Pi^n(\cdot \mid y)$, satisfies
	\begin{equation}
		\Pi^n(\|\theta - \theta_0\|_2/\sqrt{n} \geq M\sigma_0 \epsilon_n \mid y) \rightarrow 0, \; \mathrm{as}\, n \rightarrow \infty,
	\end{equation}
	in probability or in $L_1$ wrt the probability measure of $y$, for $\epsilon_n \asymp \sqrt{s_0\log n/n}$ and a constant $M > 0$.
	\label{theorem-graph:posteriorconvergencerate}
	
\end{theorem}

The above theorem implies that the mean squared error of $\theta$ is of the order $s_0\log n/n$, which differs from the oracle contraction rate $s_0/n$ under known graph sparsity by only a logarithmic factor in $n$. Note that our prior construction does not require any knowledge regarding the sparsity pattern of the graph, and hence the proposed Bayesian estimation procedure can adapt to the total variation metric(wrt the $L_0$-norm) and acheive the desired optimality in estimating the signal over the graph $G$. We discuss the proof of the above result in the appendix.

We now apply our method on a real-world network and check the performance of signal de-noising over the same.

\subsection{Application: Signal de-noising over Euroroad network}

We consider signals defined on a real-world graph, given by the Euroroad network \citep{nr}, obtained from the repository \url{http://networkrepository.com/inf-euroroad.php}. The Euroraod network is a real road network connecting various cities in Europe, having a total of around 1200 vertices and about 1400 edges. We first extract the largest connected component of the network, which results in 1039 vertices and 1305 edges. Next, to generate true piecewise constant signals on this network, we extract non-overlapping communities from the network using the Louvain community detection algorithm \citep{blondel2008fast}. The Louvain algorithm is one of the most popular community detection algorithms available in the literature which detects underlying communities in a network via optimising the modularity of the same. We use the \texttt{cluster\_louvain} function available in the \texttt{igraph} package in \texttt{R} to extract the communities. The method results in 25 connected communities, with the smallest community having 11 vertices and the largest community having 82 vertices as members. The true signal on vertex $j$ with community label $k$ is given by $\theta_{0,j} = k \Ind\{\mathrm{vertex}\, j \in \mathrm{community}\, k\},\, 1 \leq j \leq n.$ Finally we generate a noisy version of the signal on the network given by $y_j \sim N(\theta_{0,j}, 0.3^2),\, 1 \leq j \leq n.$

We apply our proposed Horseshoe fusion method on the above network by first reducing the same to a DFS-chain graph. We choose three root vertices at random from $V$, and construct three different DFS-chain graphs, followed by drawing MCMC samples from the posterior independently for the three chains. Finally, we merge the three sets of posterior samples after discarding the burn-in samples in each of the cases and find the posterior mean of the combined sample to arrive at the estimates. The MSE in estimation came out to be $0.022$ and the adjusted MSE is $0.0001$. The network visualization for the graph with true signal and the de-noised graph are provided in Figure~\ref{fig: euroroad-graph}, where the vertices of the former are color-coded according to their respective communities. Note that since the Horseshoe prior is a continuous shrinkage prior, no two estimated signals are of equal magnitude. Hence while coloring the vertices of the de-noised graph, we classify the estimated values in 25 bins of equal length, and color the vertices so that vertices having estimated signals falling in the same bin have identical colors. Additionally, in Figure~\ref{fig:step-plot-euroroad}, we provide visualizations of the true signal and the corresponding estimates as a step-function via permutation of the graph vertices to demonstrate the estimation performance of our proposed Bayesian fusion method. The figure re-iterates the excellent performance of our method to detect the piecewise constant signals over different communities in the actual graph. It is worthwhile to note in this regard that we work only with a transformed version of the original graph structure, yet we are able to learn the signals over the communities in an efficient manner, adapting to the sparsity level of the graphical structure.

\begin{figure}
	\centering
	\begin{tabular}{c}
		\includegraphics[height=3.5in, width=5in]{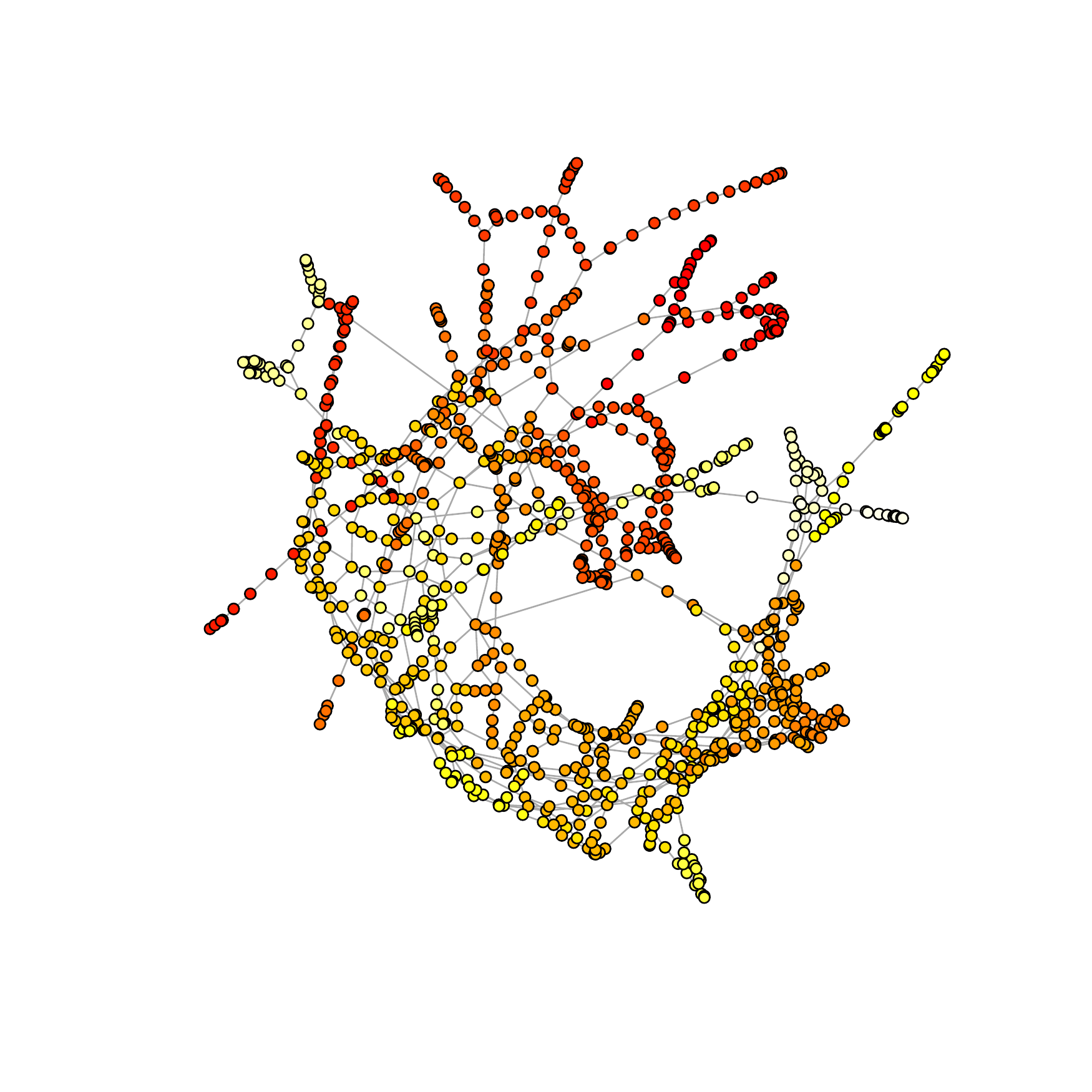}  \\
		(a) Euroroad network with true signal values  \\[1pt]
		\includegraphics[height=3.5in, width=5in]{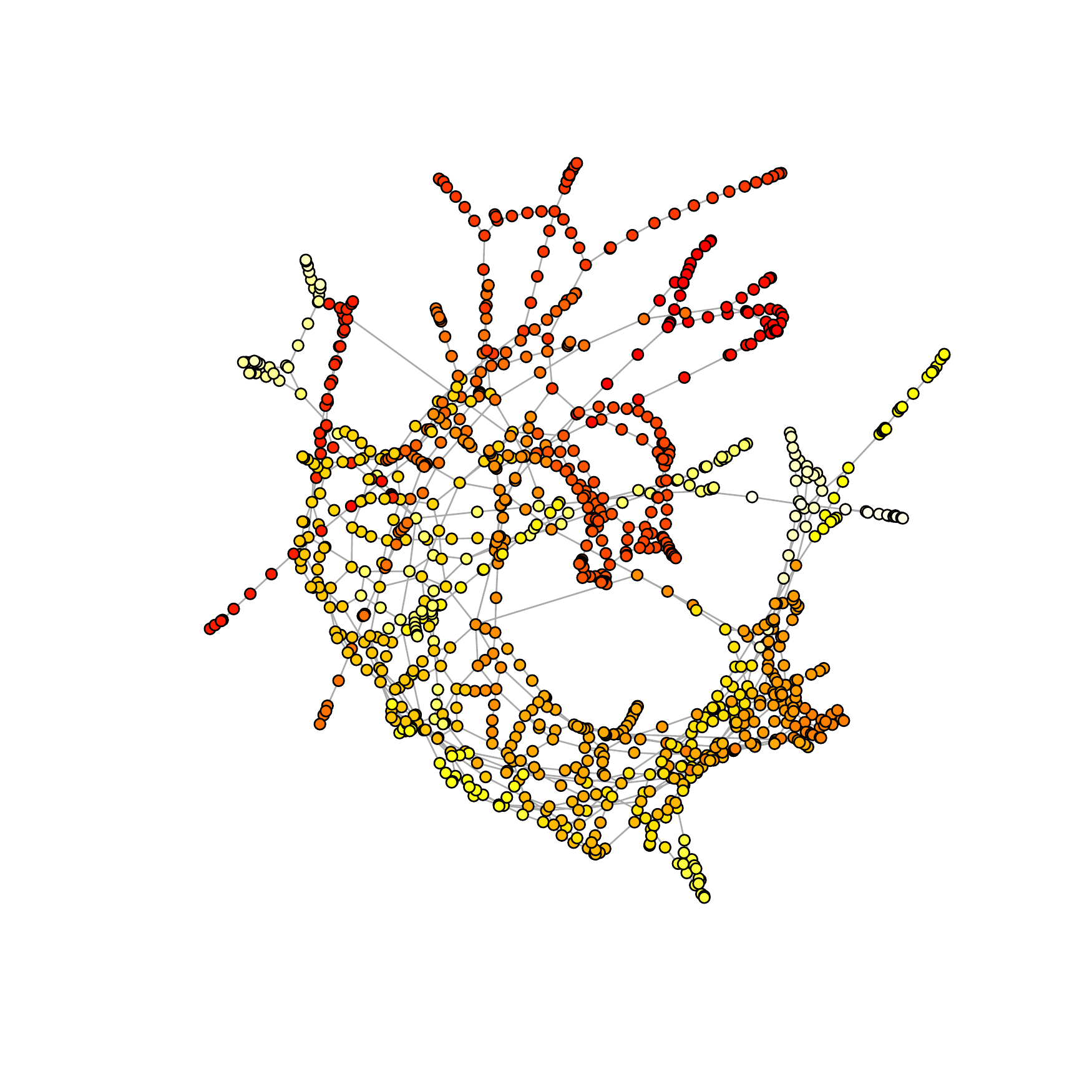} \\
		(b) Euroroad network with Bayesian Horseshoe fusion estimation-based values \\[1pt]
	\end{tabular}
	\caption{Figure showing estimation performances of the proposed Bayesian Horseshoe prior based fusion method.}
	\label{fig: euroroad-graph}
\end{figure} 

\begin{figure}
	\centering
	\includegraphics[width=0.7\linewidth]{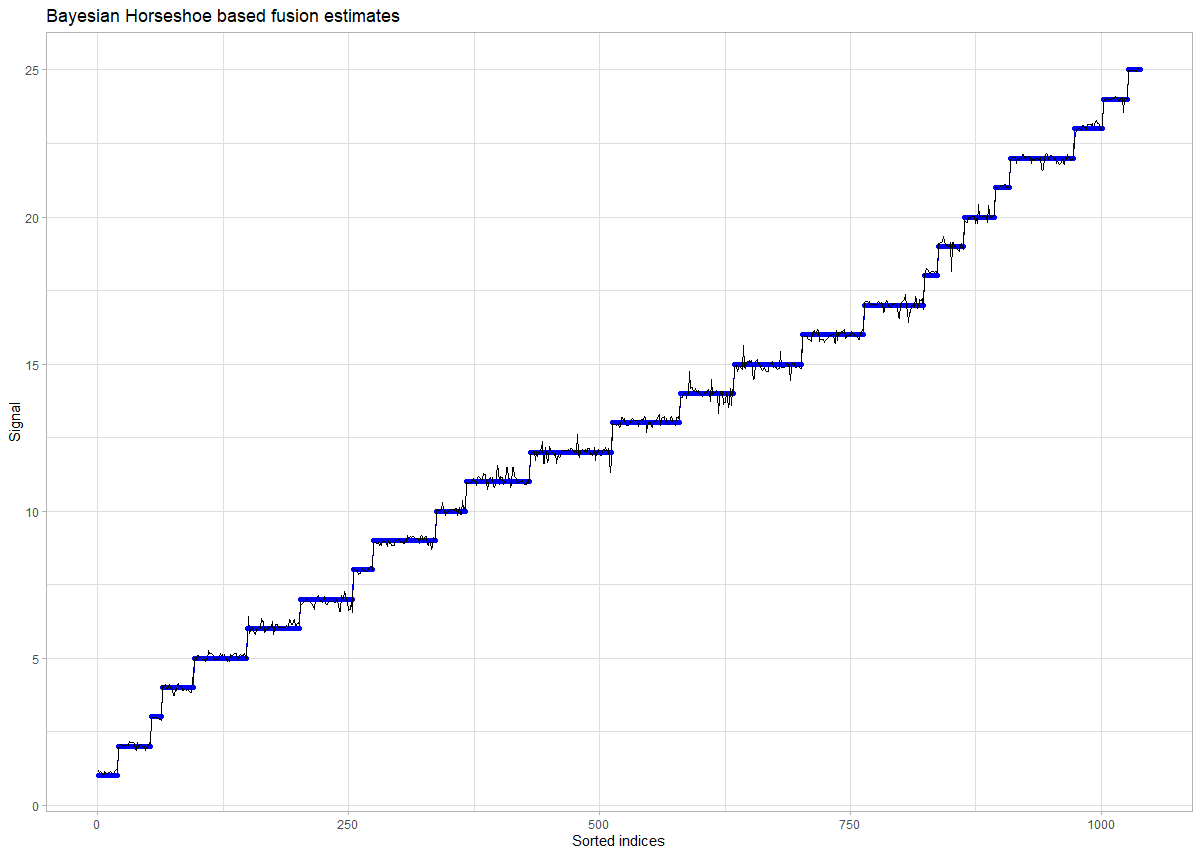}
	\caption{Step plot showing the Bayesian Horseshoe fusion performance for graph signal de-noising. The true signals are in blue and the estimated signals are in black.}
	\label{fig:step-plot-euroroad}
\end{figure}

\section{Discussion}
\label{sec:discussion}

In this paper, we have addressed the problem of fusion estimation in a normal sequence model where the underlying mean parameter vector is piecewise constant. Noise removal and block structure learning is an important and widely studied problem with applications in a plethora of areas. We developed a Horseshoe prior based fusion estimation method and provided theoretical guarantees via establishing posterior convergence and selection consistency. Numerical studies demonstrated excellent performance of our method as compared to other competing Bayesian and frequentist approaches. 

In many applications, we observe noisy signals defined over an undirected graph. The excellent performance of the Horseshoe prior based learning motivated us to extend the same to the graph signal denoising problem as well. We constructed a Gibbs sampler for efficient learning of graph signals via transformation of an arbitrary graph to a DFS-chain graph, and provided theoretical guarantees of the procedure.

There are several directions for extending our work. One important area is to move beyond normal sequence models and consider multiple graphical models. The graph fused lasso method for estimation and structure learning in multiple Gaussian graphical models \citep{danaher2014joint} has got much attention in recent years. The Horseshoe fusion prior may be utilized in this instance for precision matrix estimation over several groups and structure learning of respective graphical models. Additionally, our proposed approach may be extended to isotonic regression models, where there is a natural ordering in the means of the observations. We propose to explore these ideas as future work.

\section*{Data availability and software}

Data for the DNA copynumber analysis are available for free from the archived version of the \texttt{cghFlasso} package in \texttt{R}. Data on solar X-ray flux are available for free from the National Oceanic and Atmospheric Administration (NOAA)'s website at \url{https://www.ngdc.noaa.gov/stp/satellite/goes/dataaccess.html}. The Euroroad network data are available for free from the network repository \url{http://networkrepository.com/inf-euroroad.php}.

\texttt{R} codes to implement the fusion estimation procedure for the normal sequence model and for the graph signal de-noising problem, along with the codes to transform an arbitrary graph to a DFS-chain graph are available at \url{https://github.com/sayantanbiimi/HSfusion}.

\section*{Acknowledgements}
The author is supported by DST INSPIRE Faculty Award Grant No. 04/2015/002165, and also by IIM Indore Young Faculty Research Chair Award grant. 

\section*{Appendix}

\subsection*{Additional lemmas and proofs of results}
\begin{lemma}
	Under the assumptions on the true parameter $\theta_0$ and the prior distribution on $\theta$ as outlined in Theorem~\ref{theorem:posteriorconvergencerate}, for some $b' > 0$, we have,
	\begin{equation}
		\int_{-s_0\log n/n^2}^{s_0 \log n/n^2}p_{HS}(\eta; \tau)\, d\eta \geq 1 - n^{-b'},
	\end{equation}
where $p_{HS}(\eta; \tau)$ denotes the Horseshoe prior density on $\eta$ with hyperparameter $\tau > 0$.
\label{lemma:priormass}
\end{lemma}

\begin{proof}
	Define $a_n = s_0 \log n/n^2$. We will show that $\int_{-a_n}^{a_n}p_{HS}(\eta; \tau)\, d\eta \geq 1 - n^{-b}$ for some $b > 0$ under the assumptions outlined in Theorem~\ref{theorem:posteriorconvergencerate}. The Horseshoe prior density does not assume a closed form; however, \cite{carvalho2010horseshoe} showed that the density admits tight lower and upper bounds given by,
	\begin{equation}
		\label{eqn:HSbounds}
		\dfrac{1}{\tau(2\pi)^{3/2}}\log \left( 1 + \dfrac{4\tau^2}{\theta^2}\right) < p_{HS}(\theta \mid \tau) <  
		\dfrac{2}{\tau(2\pi)^{3/2}}\log \left( 1 + \dfrac{2\tau^2}{\theta^2}\right).
	\end{equation}
	Hence, we have,
	\begin{eqnarray}
		&&1 - \int_{-a_n}^{a_n}p_{HS}(\eta; \tau)\, d\eta = 2 \int_{a_n}^{\infty} p_{HS}(\eta; \tau)\, d\eta \nonumber \\
		&& < \left( \dfrac{2}{\pi^3}\right) ^{1/2}\dfrac{1}{\tau}\int_{a_n}^{\infty}\log \left( 1 + \dfrac{2\tau^2}{\eta^2} \right)\,d\eta \nonumber \\
		&& < \left( \dfrac{2}{\pi^3}\right) ^{1/2}\dfrac{1}{\tau} \int_{a_n}^{\infty} \dfrac{2\tau^2}{\eta^2}\,d\eta \nonumber \\
		&& \lesssim \dfrac{\tau}{a_n} \leq n^{-b'}, \nonumber
	\end{eqnarray}
where $0 < b' < b$, for $\tau \leq a_n n^{-b} \asymp n^{-(2 + b)}.$
This completes the proof.

\end{proof}

As mentioned earlier, the above result guarantees that the Horseshoe prior puts sufficient mass around zero for the successive differences $\eta_i$ so as to facilitate Bayesian shrinkage for the block-structured mean parameters. However, the prior should be able to retrieve the blocks effectively as well, so that successive differences that are non-zero should not be shrunk too much. The following result below guarantees that the Horseshoe prior is `thick' at non-zero parameter values, so that it is not too sharp.

\begin{lemma}
	Consider the prior structure and the assumptions as mentioned in Theorem~\ref{theorem:posteriorconvergencerate}. Then, we have,
	\begin{equation}
		- \log \left(\inf_{\eta/\sigma \in [-L,L]} p_{HS}(\eta;\tau)\right) = O(\log n).
	\end{equation}
\label{lemma:priorthickness}
\end{lemma}

\begin{proof}
	We utilise the tight lower bound for the Horseshoe prior density as in \ref{eqn:HSbounds}. We have,
	\begin{eqnarray}
		\inf_{\eta/\sigma \in [-L,L]} p_{HS}(\eta;\tau) &>& \inf_{\eta/\sigma \in [-L,L]} \dfrac{1}{\tau(2\pi)^{3/2}}\log \left( 1 + \dfrac{4\tau^2}{\eta^2}\right) \nonumber \\
			&\geq & \dfrac{1}{\tau(2\pi)^{3/2}}\log \left( 1 + \dfrac{4\tau^2}{L^2}\right)  \asymp \dfrac{\tau}{L^2} \nonumber \\
			&=& O(n^{-b''}) ,\; b''> 0, \nonumber 
	\end{eqnarray}
	for $-\log(\tau) = O(\log n)$ and $\log L = O(\log n).$
	Thus, we get, $- \log \left(\inf_{\eta/\sigma \in [-L,L]} p_{HS}(\eta;\tau)\right) = O(\log n),$ hence completing the proof.
	
\end{proof}
We now present the proofs of the main results in our paper.

\begin{proof}[Proofs of Theorem~\ref{theorem:posteriorconvergencerate} and Theorem~\ref{theorem:structure-recovery}]
	The proof readily follows from Theorems 2.1 and 2.2 in \cite{song2020bayesian}, if we can verify the conditions (see display (2.5) in the aforementioned paper) for posterior convergence specified therein. Lemma~\ref{lemma:priormass} and Lemma~\ref{lemma:priorthickness} imply that the first two conditions are satisfied. The third condition is satisfied for a Normal prior distribution on $\theta_1$ and for a fixed hyperparameter $\lambda_1 > 0.$ The fourth condition is trivially satisfied for fixed choices of (non-zero) hyperparamters $a_\sigma$ and $b_\sigma$, and for some fixed (unknown) true error variance $\sigma_0^2.$
\end{proof}

\begin{proof}[Proof of Theorem~\ref{theorem-graph:posteriorconvergencerate}]
 Note that, by Lemma~1 in \cite{Padilla2017}, the total variation of $\theta$ wrt to the $L_0$-norm in the DFS-chain graph $G_C$ is bounded by at most twice that in $G$, that is, $\theta \in l_0[G,s]$ always implies that $\theta \in l_0[2s].$ Also, assumptions \ref{assump-graph:true-blocks} - \ref{assump-graph:prior} for an arbitrary graph $G$ implies assumptions \ref{assump:true-blocks} - \ref{assump:prior} for a linear chain graph. Hence the result readily follows from that of Theorem~\ref{theorem:posteriorconvergencerate}.
\end{proof}

\subsection*{A practical solution to block structure recovery}
\label{sec:block-recovery}

The Horseshoe prior is a continuous shrinkage prior, and hence block structure recovery is not straight-forward. In Bayesian fusion estimation with Laplace shrinkage prior or with $t$-shrinkage prior, one may discretize the samples obtained via MCMC using a suitable threshold. \cite{song2020bayesian} recommended using the $1/2n$-th quantile of the corresponding prior for discretization of the scaled samples. In Section~\ref{sec:theory}, we proposed a threshold based on the posterior contraction rate. However, that would also require an idea of the true block size, which may not be available always, or may be difficult to ascertain beforehand. We provide a practical approach for discretization of the samples. Note that in a sparse normal sequence problem, the posterior mean of the mean parameter $\theta_i$ is given by $\kappa_i y_i$, where $\kappa_i$ is the shrinkage weight. These shrinkage weights mimic the posterior selection probability of the means, and hence thresholding the weights to 0.5 provide excellent selection performance, which is justified numerically \citep{carvalho2010horseshoe} and later theoretically \citep{datta2013asymptotic}. Motivated by the same, we propose to use the following thresholding rule for selection of the blocks for our proposed method: $\theta_{j_1}$ and $\theta_{j_2}$ are equal if $|\hat{\theta}_{j_1} - \hat{\theta}_{j_2}| < 0.5|y_{j_1} - y_{j_2}|$, where $\theta_{j_1}, \theta_{j_2}$ are the posterior means of $\theta_{j_1}$ and $\theta_{j_2}$respectively, for $1 \leq j_1 \neq j_2 \leq n.$ We thus estimate the block structure indicator $s_{j_1j_2} = \Ind\{\theta_{j_1} = \theta_{j_2}\}$ by $\hat{s}_{j_1j_2} = \Ind\{|\hat{\theta}_{j_1} - \hat{\theta}_{j_2}| < 0.5|y_{j_1} - y_{j_2}|\}$. To evaluate the performance of the structure recovery method using the proposed thresholding approach, we define the metrics $W$ and $B$, where $W$ is the average within-block variation defined as $W := \mathrm{mean}_{\{s_{j_1j_2} \neq 0\}}|\hat{\theta}_{j_1} - \hat{\theta}_{j_2}|$, and $B$ is the between-block separation defined as $B := \min_{\{s_{j_1j_2} = 0\}}|\hat{\theta}_{j_1} - \hat{\theta}_{j_2}|.$

This practical thresholding approach resulted in excellent block structure recovery performance with respect to the metrics $W$ and $B$. We compare our results with the other Bayesian fusion estimation methods and also with the frequentist fused lasso method. The fused lasso method results in exact structure recovery, and hence no thresholding is required. As mentioned earlier, the thresholding rules for the competing Bayesian methods are taken as suggested in \cite{song2020bayesian}. We notice that for low error variances, the within-block average variation is lower in case of the Horseshoe fusion and $t$-fusion priors, with the Horseshoe fusion method having larger between-block speration especially in case of higher noise, thus indicating superior structure learning capabilities. The results are summarized in Table~\ref{table:simu-results-block}.

Other possible directions for coming up with a threshold include using a multiple hypothesis testing approach for successive differences in the means, and also using simultaneous credible intervals. We leave the theoretical treatment of using (or improving, if possible) our practical thresholding approach and exploring other proposed methods as a future work.

\begin{sidewaystable}[h]
	\small
	\begin{tabular}{ll|cc|cc|cc|cc}
		\hline 
		&       & \multicolumn{2}{c}{HS-fusion} & \multicolumn{2}{c}{$t$-fusion} & \multicolumn{2}{c}{Laplace fusion} & \multicolumn{2}{c}{$L_1$ fusion}  \\
		Signal      & $\sigma$           &  W & B              & W              & B              & W & B  & W & B        \\
		\hline 
		& 0.1 & 0.039 (0.001) & 0.846  (0.005) & 0.035 (0.001) & 0.892  (0.010)  & 0.678 (0.002) & 0.003 (0.000) & 0.069 (0.003) & 0.739 (0.011) \\
		Even & 0.3 & 0.140  (0.003) & 0.425 (0.020)  & 0.121 (0.005) & 0.381 (0.035)  & 0.687 (0.003) & 0.003 (0.000)     & 0.181 (0.007) & 0.333 (0.023) \\
		& 0.5 & 0.250 (0.005) & 0.139 (0.016) & 0.221 (0.007) & 0.079 (0.017)  & 0.704 (0.004) & 0.003 (0.000)  & 0.285 (0.009) & 0.119 (0.016) \\
		&&&&&&&&&\\
		& 0.1 & 0.037 (0.001) & 0.846 (0.005)  & 0.034 (0.001) & 0.864 (0.010)  & 0.537 (0.002) & 0.006 (0.001)  & 0.070 (0.003) & 0.751 (0.010)  \\
		Uneven & 0.3 & 0.134 (0.003) & 0.457 (0.020)   & 0.106 (0.004) & 0.266 (0.031) & 0.548 (0.003) & 0.005 (0.001) & 0.182 (0.006) & 0.348 (0.021) \\
		& 0.5 & 0.238 (0.005) & 0.146 (0.015)  & 0.173 (0.006) & 0.022 (0.007)  & 0.570  (0.004) & 0.004 (0.000)     & 0.284 (0.009) & 0.131 (0.015) \\
		&&&&&&&&&\\
		& 0.1 & 0.036 (0.001) & 0.800 (0.007) & 0.034 (0.001) & 0.670 (0.022) & 0.325 (0.002) & 0.008 (0.001) & 0.106 (0.001) & 0.710 (0.010)  \\
	V. Uneven	& 0.3 & 0.128 (0.002) & 0.401 (0.025) & 0.092 (0.003) & 0.013 (0.004) & 0.344 (0.003) & 0.007 (0.001) & 0.276 (0.006) & 0.269 (0.023) \\
		& 0.5 & 0.222 (0.004) & 0.149 (0.019) & 0.124 (0.004) & 0.001 (0.000) & 0.376 (0.004) & 0.006 (0.001) & 0.391 (0.010)  & 0.105 (0.017)\\
		\hline 
	\end{tabular}
	\caption{Within (W) and between (B) block average variation (with associated standard errors in parentheses) for different fusion estimation methods. For good structure recovery, W values should be low and B values should be high.}
	\label{table:simu-results-block}
\end{sidewaystable}

\bibliographystyle{apalike}
\bibliography{HSfusion-bib.bib}

\end{document}